\documentclass[runningheads]{llncs}

\usepackage{graphicx}
\usepackage{marvosym}

\makeatletter
\RequirePackage[bookmarks,unicode,colorlinks=true]{hyperref}%
   \def\@citecolor{blue}%
   \def\@urlcolor{blue}%
   \def\@linkcolor{blue}%

\def\orcidID#1{\smash{\href{http://orcid.org/#1}{\protect\raisebox{-1.25pt}{\protect\includegraphics{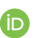}}}}}
\makeatother

\usepackage{xcolor}
\usepackage[all]{xy}
\usepackage{extarrows}
\usepackage{amssymb}
\usepackage{stmaryrd}
\usepackage{pst-node}
\usepackage{bussproofs}
\usepackage{sidenotes}

\newcommand\play[1]{\mathsf{play}(#1)}
\newcommand\trace[1]{\mathit{Tr}(#1)}
\newcommand\lang[1]{\mathit{L}(#1)}
\newcommand\abra[1]{\langle #1 \rangle}
\newcommand\cutout[1]{}

\newcommand\lock{\mathit{lock}}

\newcommand\mm{\mathsf{m}}
\newcommand\eq{\epsilon_\Q}
\newcommand\ea{\epsilon_\A}
\newcommand\imax{\mathit{max}}
\newcommand\off{\rho}

\newcommand\alp[1]{\mathcal{T}_{#1}}
\newcommand\q{\mathsf{q}}

\newcommand\colortext[2]{}

\newcommand\rl[1]{\colortext{green}{#1}}
\newcommand\rlnote[1]{\sidenote{\rl{#1}}}

\newcommand\igw[1]{\colortext{violet}{#1}}

\newcommand\N{\mathbb{N}}

\newcommand\sem[1]{{\llbracket}{#1}{\rrbracket}}
\newcommand\seq[2]{{#1} \vdash {#2}}

\newcommand\rarr\rightarrow
\newcommand\dom[1]{\mathsf{dom}(#1)}
\newcommand\makeset[1]{\{#1\}}
\newcommand\word[1]{\textrm{#1}}
\newcommand\clg[1]{\mathcal{#1}}

\newcommand\trans[1]{{\xlongrightarrow{#1}}}
\newcommand\Aut{\mathcal{A}}
\newcommand\la{\mathsf{LA}}
\newcommand\lla{\mathsf{LLA}}
\newcommand\sla{\lla}
\newcommand\D{\mathcal{D}}
\newcommand\pred[1]{\mathit{pred}(#1)}
\newcommand\predc{\mathit{pred}}

\newcommand\fica{\mathsf{FICA}}
\newcommand\lfica{\mathsf{LFICA}}
\newcommand\sfica{\lfica}
\newcommand\iatype[1]{{\bf #1}}
\newcommand\expt{\iatype{exp}}
\newcommand\comt{\iatype{com}}
\newcommand\vart{\iatype{var}}
\newcommand\semt{\iatype{sem}}
\newcommand\ctx{\mathcal{C}}
\newcommand\iaterm[1]{{\bf #1}}
\newcommand\ade[2]{\mathit{ad}_{#1}(#2)}

\newcommand{\aasg}{\,\raisebox{0.065ex}{:}{=}\,}

\newcommand\while[2]{\iaterm{while}\,#1\,\iaterm{do}\,#2}
\newcommand\cond[3]{\iaterm{if}\,#1\,\iaterm{then}\,#2\,\iaterm{else}\,#3}
\newcommand\skipcom{\iaterm{skip}}
\newcommand\divcom{\iaterm{div}}
\newcommand\newin[2]{\iaterm{newvar}\,#1\,\iaterm{in}\,#2}
\newcommand\newsem[2]{\iaterm{newsem}\,#1\,\iaterm{in}\,#2}
\newcommand\grb[1]{\iaterm{grab}(#1)}
\newcommand\rls[1]{\iaterm{release}(#1)}
\newcommand\mkvar[2]{\iaterm{mkvar}(#1,#2)}
\newcommand\mksem[2]{\iaterm{mksem}(#1,#2)}
\newcommand\arop[1]{\mathbf{op}(#1)}
\newcommand\vvec[1]{\overrightarrow{#1}}
\newcommand\mem{\mathcal{V}}
\newcommand\mss{\hspace*{0mm}}
\newcommand{\astep}[4]{\mem\vdash #2,\,#1 &\mss\longrightarrow\mss & #4,\,#3}
\newcommand{\step}[4]{\mem\vdash #2,\,#1\longrightarrow #4,\,#3}
\newcommand{\sqsubsim}{\,\raisebox{-.5ex}{$\stackrel{\textstyle\sqsubset}{\scriptstyle{\sim}}$}\,}

\newcommand\parc{||}

\newcommand\moveset{{\mathcal{M}}}
\newcommand\move[1]{\mathsf{#1}}
\newcommand\mread{\move{read}}
\newcommand\mwrite[1]{\move{write}(#1)}
\newcommand\mok{\move{ok}}
\newcommand\mrun{\move{run}}
\newcommand\mdone{\move{done}}

\newcommand\mrls{\move{rls}}
\newcommand\mgrb{\move{grb}}
\newcommand\mq{\move{q}}

\newcommand\wait{\mathsf{WAIT}}
\newcommand\fork{\mathsf{FORK}}

\newcommand{\comp}[1]{\textsf{comp}(#1)}

\newcommand\justf[3][]{\nccurve[arrowsize=2.5pt,nodesep=.5pt,offsetB=-2pt,linewidth=0.4pt,angleA=110,angleB=30,linecolor=darkgray#1]{->}{#2}{#3}}

\newcommand\justg[3][]{\nccurve[arrowsize=2.5pt,nodesep=.5pt,offsetB=-2pt,linewidth=0.4pt,angleA=130,angleB=30,linecolor=darkgray#1]{->}{#2}{#3}}

\newcommand\justh[3][]{\nccurve[arrowsize=2.5pt,nodesep=.5pt,offsetB=-2pt,linewidth=0.4pt,angleA=130,angleB=30,linecolor=darkgray#1]{->}{#2}{#3}}

\newcommand\justj[3][]{\nccurve[arrowsize=2.5pt,nodesep=.5pt,offsetB=-2pt,linewidth=0.4pt,angleA=160,angleB=30,linecolor=darkgray#1]{->}{#2}{#3}}

\newcommand\justn[4][]{\nccurve[arrowsize=2.5pt,nodesep=.5pt,offsetB=-2pt,linewidth=0.4pt,angleA=#4,angleB=30,linecolor=darkgray#1]{->}{#2}{#3}}

\newcommand\Q{\mathsf{Q}}
\newcommand\A{\mathsf{A}}

\begin{document}

%%%%%%%%%%%%%%%%%%%%%%%%%%%%%%%%%%%%%%%%%%%%%%%%%%%%%%%%%%%%%%%%%%%%%%%%%%%%%%%%
% FRONT MATTER

\title{Leafy automata for higher-order concurrency}

% an abbreviated paper title here:

\author{Alex Dixon\inst{1}\orcidID{0000-0003-3048-5128} (\Letter) \and
Ranko Lazi\'c\inst{2}\orcidID{0000-0003-3663-5182} \and
Andrzej S. Murawski\inst{3}\orcidID{0000-0002-4725-410X} \and
Igor~Walukiewicz\inst{4}}

\authorrunning{A.~Dixon, R.~Lazi\'c, A.~S.~Murawski and I.~Walukiewicz}

\institute{University of Warwick, UK, \texttt{alexander.dixon@warwick.ac.uk}\and
University of Warwick, UK, \texttt{r.s.lazic@warwick.ac.uk} \and
University of Oxford, UK, \texttt{andrzej.murawski@cs.ox.ac.uk} \and
CNRS, Universit\'e de Bordeaux, France
}

\maketitle

\begin{abstract}

Finitary Idealized Concurrent Algol ($\fica$) is a prototypical programming language combining functional, imperative, and concurrent computation. There exists a fully abstract game model of $\fica$, which in principle can be used to prove equivalence and safety of $\fica$ programs. Unfortunately, the problems are undecidable for the whole language, and only very rudimentary decidable sub-languages are known. 

We propose leafy automata as a dedicated automata-theoretic formalism for representing the game semantics of $\fica$. The automata use an infinite alphabet with a tree structure. 
We show that the game semantics of any $\fica$ term can be represented by traces of a leafy automaton. Conversely, the traces of any leafy automaton can be represented by a $\fica$ term.
Because of the close match with $\fica$, we view leafy automata as a promising starting point for finding decidable subclasses of the language and, more generally, to provide a new perspective on models of higher-order concurrent computation. 

Moreover, we identify a fragment of $\fica$ that is amenable to verification by translation into a particular class of leafy automata.  Using a locality property of the latter class, where communication between levels is restricted and every other level is bounded, we show that their emptiness problem is decidable by reduction to Petri nets reachability.

%\igw{Abstract modified}
%\am{So we don't mention "local leafy automata" by name here and don't brag about the decidability result?}

%Finitary Idealized Concurrent Algol ($\fica$) is a prototypical programming language combining functional, imperative, and concurrent computation. There exists a fully abstract game model of $\fica$, which in principle can be used to prove equivalence and safety of $\fica$ programs. 

% AM: I added a similar sentence later to delay the "bad news".
%Unfortunately, the problems are undecidable for the whole language and only very rudimentary decidable sub-languages are known. 

%We propose leafy automata as a dedicated automata-theoretic formalism for representing the game semantics of $\fica$. The automata use an infinite alphabet with a tree structure. To demonstrate suitability of this automata model, we show that:

%\begin{enumerate}
%\item  the game semantics of any $\fica$ term can be represented by traces of a leafy automaton, and 
%\item  the traces of any leafy automaton can be represented by a $\fica$ term.
%\end{enumerate}

%\adnote{Perhaps just ``The emptiness problem for leafy automata is undecidable.''}
%Unfortunately, the emptiness problem for leafy automata turns out to be undecidable. To address this, we propose a decidable variant, called \emph{local},
%where communication between levels is restricted and every other level is bounded.
%We also identify a fragment of $\fica$ that is amenable to verification by translation  into the local variant.

\end{abstract}

%%%%%%%%%%%%%%%%%%%%%%%%%%%%%%%%%%%%%%%%%%%%%%%%%%%%%%%%%%%%%%%%%%%%%%%%%%%%%%%%
% MAIN BODY

% TODO:
% - Reformat relevant parts according to lncs-resources/style-guide.tex

\section{Introduction}

Game semantics is a versatile paradigm for giving semantics to a wide spectrum of
programming languages~\cite{AM98b,MT16b}.
It is well-suited for studying the observational equivalence of programs and, more generally, the
behaviour of a program in an arbitrary context. 
About 20 years ago, it was discovered that the game semantics of a program can sometimes be expressed by a finite automaton or another simple computational model~\cite{GMcC00}. 
This led to algorithmic uses of game semantics 
for program analysis and verification~\cite{AGMO04,DGL06,GM06,BG08,HO09,HMO12,KMOWW12,MRT15,Dim16,Dim17}. 
%\igwnote{more references needed}.\amnote{done}
Thus far, these advances concerned mostly languages without concurrency.

In this work, we consider Finitary Idealized Concurrent Algol ($\fica$) and its fully abstract game semantics~\cite{GM08}.
It is a call-by-name language with higher-order features, side-effects, and concurrency implemented by a parallel composition operator and semaphores.
It is finitary since, as it is common in this context, base types are restricted to finite domains.
Quite surprisingly, the game semantics of this language is arguably simpler than that for the language without concurrency.
The challenge comes from algorithmic considerations.
%Unlike for the sequential case~\cite{GMcC00}, observational equivalence of programs is undecidable even for second-order programs without recursion~\cite{GMO06}.

%Our long-term goal is to find fragments of $\fica$, for which some program analysis problems become decidable.
Following the successful approach from the sequential case~\cite{GMcC00,Ong02,Mur04,MW08,CBMO19}, the first step is to find an automaton model abstracting the phenomena appearing in the semantics. 
The second step is to obtain program fragments from structural restrictions on the automaton model. 
In this paper we take both steps.

We propose \emph{leafy automata}:  an automaton model working on nested data.
Data are used to represent pointers in plays, while the nesting of data reflects structural dependencies in the use of pointers. 
Interestingly, the structural dependencies in plays boil down to imposing a tree structure on the data.
We show a close correspondence between the automaton model and the game semantics of $\fica$.
For every program, there is a leafy automaton whose traces (data words) represent precisely the plays in the semantics of the program (Theorem~\ref{thm:trans}).
Conversely, for every leafy automaton, there is a program whose semantics consists of plays representing the traces of the automaton (Theorem~\ref{thm:toalgol}).
(The latter result holds modulo a saturation condition we explain later.) 
This equivalence shows that leafy automata are a suitable model for studying decidability
questions for $\fica$.

Not surprisingly, due to their close connection to $\fica$, leafy automata turn out to have an undecidable emptiness problem.
We use the undecidability argument to identify the source, namely communication across several unbounded levels, i.e., levels in which nodes can produce an unbounded number of children during the lifetime of the automaton.
To eliminate the problem, we introduce a restricted variant of leafy automata, called \emph{local},
in which every other level is bounded and communication is allowed to cross only one unbounded node.
Emptiness for such automata can be decided via reduction to a number of instances of Petri net reachability problem.

We also identify a fragment of $\fica$, dubbed \emph{local} $\fica$ ($\sfica$), which maps onto
local leafy automata. 
It is based on restricting the distance between semaphore and variable
declarations and their uses inside the term.
This is a first non-rudimentary fragment of $\fica$ for which some verification tasks are decidable.
Overall, this makes it possible to use local leafy automata to analyse $\sfica$ terms and decide associated verification tasks.

\paragraph{Related work} Concurrency, even with only first-order recursion, leads to
undecidability~\cite{Ram00}. 
Intuitively, one can encode the intersection of languages of two pushdown automata.
From the automata side, much research on decidable cases has concentrated on bounding interactions between stacks representing different threads of the program~\cite{QR05,TMP09,AGK14}.
From the game semantics side, the only known decidable fragment of $\fica$ is Syntactic Control
of Concurrency (SCC)~\cite{GMO06}, which imposes bounds 
on the number of threads in which arguments can be used.
This restriction makes it possible to represent the game semantics of programs by finite automata. 
In our work, we propose automata models that correspond to unbounded interactions with
arbitrary $\fica$ contexts, and importantly that remains true also when we restrict the terms to $\sfica$.
%\amnote{I wanted to stress that in both cases we consider interactions with arbitrary $\fica$ contexts
%and highlight that we are not restricting contexts, as was done in the "rudimentary" SCC.} 
%\adnote{Perhaps make that more explicit here (if there is space!)}
Leafy automata are a model of computation over an infinite alphabet.
This area has been explored extensively, partly motivated by applications to database theory, notably XML~\cite{Sch07}. 
In this context, nested data first appeared in~\cite{BB07}, where the authors considered shuffle expressions as 
the defining formalism. Later on, data automata~\cite{BDMSS11} and class memory automata~\cite{BS10} have been adapted to nested data in~\cite{DHLT14,CMO15}.
They are similar to leafy automata in that the automaton is allowed to access states related to previous uses of data values at various depths. 
% What distinguishes leafy automata is the fact that the lifetime of a data value follows a question/answer discipline and is quite restricted. 
What distinguishes leafy automata is that the lifetime of a data value is precisely defined and follows a question and answer discipline in correspondence with game semantics.
Leafy automata also feature run-time ``zero-tests'', activated when reading answers.

For most models over nested data, the emptiness problem  is undecidable. 
To achieve decidability, the authors in~\cite{DHLT14,CMO15} relax the acceptance conditions so that the emptiness problem can eventually be recast as a coverability problem for a well-structured transition system.
In~\cite{CHMO15}, this result was used to show decidability of equivalence for a first-order (sequential) fragment of Reduced ML.
On the other hand, in~\cite{BB07} the authors relax the order of letters in words, which leads to an analysis based on semi-linear sets. 
Both of these restrictions are too strong to permit the semantics of $\fica$, because of the game-semantic $\wait$ condition, 
which corresponds to waiting until all sub-processes terminate.

Another orthogonal strand of work on concurrent higher-order programs is based on higher-order recursion schemes~\cite{Hague13,KI13}. Unlike $\fica$, they feature recursion but the computation is purely functional over a single atomic type~$o$.

\paragraph{Structure of the paper:} 
In the next two sections we recall $\fica$ and its game semantics from~\cite{GM08}.
The following sections introduce leafy automata ($\la$) and their local variant ($\sla$),
where we also analyse the associated decision problems and, in particular, show that the non-emptiness problem for $\sla$ is decidable.
Subsequently, we give a translation from $\fica$ to $\la$ (and back) and define a fragment $\sfica$ of $\fica$ which can be translated into $\sla$.
% !TEX root =  main.tex

\section{Finitary Idealized Concurrent Algol ($\fica$)}
\label{sec:fica}

Idealized Concurrent Algol~\cite{GM08} is a paradigmatic language combining higher-order with imperative computation in the style of Reynolds~\cite{Rey78}, extended to concurrency with parallel composition ($\parc$) and binary semaphores.
We consider its finitary variant $\fica$ over the finite datatype $\makeset{0,\ldots,\imax}$ ($\imax\ge 0$)
with loops but no recursion.
Its types $\theta$ are generated by the grammar 
\[
\theta::=\beta\mid \theta\rarr\theta\qquad\qquad
  \beta::=\comt\mid\expt\mid\vart\mid\semt
\]
where 
$\comt$ is the type of commands;
$\expt$ that of $\makeset{0,\ldots,\imax}$-valued expressions;
$\vart$ that of assignable variables;
and $\semt$ that of semaphores.
The typing judgments are displayed in Figure~\ref{fig:icatypes}.
$\skipcom$ and $\divcom_\theta$ are constants representing termination and divergence respectively,
$i$ ranges over $\{0,$ $\cdots,$ $\imax\}$,
and $\mathbf{op}$ represents unary arithmetic operations, such as successor or predecessor (since we work over a finite datatype, operations of bigger arity can be defined using conditionals).
Variables and semaphores can be declared locally via $\mathbf{newvar}$ and $\mathbf{newsem}$.
Variables are dereferenced using $!M$, and semaphores are manipulated using two (blocking) primitives,
$\grb{s}$ and $\rls{s}$, which  grab and release the semaphore respectively. 
\begin{figure}[t]
\begin{center}
  \AxiomC{$\phantom{\beta}$}
 \UnaryInfC{$\Gamma\vdash\skipcom:\comt $}
\DisplayProof\quad
  \AxiomC{$\phantom{\beta}$}
 \UnaryInfC{$\Gamma\vdash\divcom_\theta:\theta $}
\DisplayProof\quad
  \AxiomC{$\phantom{\beta}$}
 \UnaryInfC{$\Gamma\vdash i:\expt$}
\DisplayProof\quad
  \AxiomC{$\seq{\Gamma}{M:\expt}$}
  \UnaryInfC{$\seq{\Gamma}{\arop{M}:\expt}$}
  \DisplayProof\\[2ex]
%  \AxiomC{$\{\Gamma\vdash N_1:\theta_1\}$}
%  \AxiomC{$\Gamma\vdash N_2:\theta_2$}
%\RightLabel{$\square$ is any sequential combinator}
%  \BinaryInfC{$\Gamma\vdash \{M_1\}\,\square\, M_2:\theta_3$}
%  \DisplayProof\\[2ex]
  \AxiomC{$\Gamma\vdash M:\comt$}
  \AxiomC{$\Gamma\vdash N:\beta$}
  \BinaryInfC{$\Gamma \vdash M;N:\beta$}
  \DisplayProof\quad
    \AxiomC{$\Gamma\vdash M:\comt$}
  \AxiomC{$\Gamma\vdash N:\comt$}
  \BinaryInfC{$\Gamma \vdash M\parc N:\comt$}
  \DisplayProof\\[2ex]
      \AxiomC{$\Gamma\vdash M:\expt$}
  \AxiomC{$\Gamma\vdash N_1,N_2:\beta$}
  \BinaryInfC{$\Gamma\vdash \cond{M}{N_1}{N_2}:\beta$}
  \DisplayProof\quad
    \AxiomC{$\Gamma\vdash M:\expt$}
  \AxiomC{$\Gamma\vdash N:\comt$}
  \BinaryInfC{$\Gamma\vdash \while{M}{N}:\comt$}
  \DisplayProof\\[2ex]
  \AxiomC{$\phantom{\beta}$}
  \UnaryInfC{$\Gamma, x:\theta \vdash x: \theta$}
  \DisplayProof\quad
  \AxiomC{$\Gamma,x:\theta\vdash M:\theta'$}
  \UnaryInfC{$\Gamma\vdash\lambda x. M:\theta\rarr\theta' $}
  \DisplayProof\quad
  \AxiomC{$\Gamma\vdash M:\theta\rarr\theta'$}
  \AxiomC{$\Gamma\vdash N:\theta$}
  \BinaryInfC{$\Gamma \vdash M N:\theta'$}
  \DisplayProof\\[2ex]
      \AxiomC{$\Gamma\vdash M:\vart$}
  \AxiomC{$\Gamma\vdash N:\expt$}
  \BinaryInfC{$\Gamma \vdash M\aasg N:\comt$}
  \DisplayProof\quad
  \AxiomC{$\Gamma\vdash M:\vart$}
  \UnaryInfC{$\Gamma \vdash !M:\expt$}
  \DisplayProof\\[2ex]
  \AxiomC{$\Gamma\vdash M:\semt$}
  \UnaryInfC{$\Gamma \vdash \rls{M}:\comt$}
  \DisplayProof\,\,
  \AxiomC{$\Gamma\vdash M:\semt$}
  \UnaryInfC{$\Gamma \vdash \grb{M}:\comt$}
  \DisplayProof\\[2ex]
    \AxiomC{$\Gamma, x:\vart\vdash M:\comt,\expt$}
  \UnaryInfC{$\Gamma\vdash \newin{x\aasg i}{M:\comt,\expt}$}
  \DisplayProof\quad
  \AxiomC{$\Gamma,x:\semt\vdash M:\comt,\expt$}
  \UnaryInfC{$\Gamma\vdash \newsem{x\aasg i}{M:\comt,\expt}$}
  \DisplayProof
\end{center}
\caption{$\fica$ typing rules\label{fig:icatypes}}
\vspace{-4mm}
\end{figure}
The small-step operational semantics of $\fica$ is reproduced in Appendix~\ref{apx:opsem}.
In what follows, we shall write $\divcom$ for $\divcom_\comt$.
\medskip

We are interested in \emph{contextual equivalence} of terms. 
Two terms are contextually equivalent if there is no context that can distinguish them with respect to may-termination. 
More formally, a term $\seq{}{M:\comt}$ is said to terminate,  written $M\!\Downarrow$, if there
exists a terminating evaluation sequence from $M$ to $\skipcom$.
Then \emph{contextual (may-)equivalence} ($\Gamma\vdash M_1\cong M_2$) is defined by:
for all contexts $\ctx$ such that $\seq{}{\ctx[M]:\comt}$,
$\ctx[M_1]\!\Downarrow$ if and only if $\ctx[M_2]\!\Downarrow$.
The force of this notion is quantification over all contexts.

Since contextual equivalence becomes undecidable for $\fica$ very quickly~\cite{GMO06}, we will look at 
the special case of testing equivalence with  terms that always diverge,
e.g. given $\Gamma\vdash M: \theta$, is it the case that
$\seq{\Gamma}{M\cong\divcom_\theta}$?
Intuitively, equivalence with an
always-divergent term means that $\ctx[M]$ will never converge (must diverge) if $\ctx$ uses $M$.
At the level of automata, this will turn out to correspond to the emptiness problem.

\label{ex:verification} In verification tasks, with the above equivalence test, we can check whether uses of $M$ can ever lead to undesirable states. For example, for a given term $\seq{x:\vart}{M:\theta}$, the term
\[
\seq{f:\theta\rarr\comt}{\newin{x\aasg 0}{(f (M)\, ||\, \cond{!x=13}{\,\skipcom\,}{\,\divcom}})}
\]
will be equivalent
to $\divcom$ only when $x$ is never set to $13$ during a terminating execution.
Note that, because of quantification over all contexts, $f$ may use $M$ an arbitrary number of times,
also concurrently or in nested fashion, which is a very expressive form of quantification.

\cutout{
The automata-theoretic emptiness problems that will be studied in the paper correspond to 
a special case  of contextual equivalence, where a term is compared with a term that always diverges,
i.e. given $\Gamma\vdash M:\theta_h\rarr\cdots\rarr\theta_1\rarr\comt$, is it the case that
$\seq{\Gamma}{M\cong\lambda x_h\cdots x_1.\divcom}$? Intuitively, equivalence with an 
always-divergent term means that $\ctx[M]$ will never converge (must diverge) if $\ctx$ uses $M$.
\rlnote{Can we provide a bit more motivation for focussing on emptiness, in addition to the kind of example that follows here?  At a first reading, these two paragraphs currently come across as a little contrived.  We should provide as robust advertising as possible here, since the criticism about applicability to program verification of what we are doing is perhaps one of the most likely ones.}

In verification tasks, the above property could be used to check whether uses of $M$ can ever lead to undesirable states. For example, given $\seq{x:\vart}{M:\theta}$, 
\[
\seq{f:\theta\rarr\comt}{\newin{x\aasg 0}{(f (M)\, ||\, \cond{!x=13}{\,\skipcom\,}{\,\divcom}})}
\]
will be equivalent
to $\divcom$ if $x$ is never set to $13$ during a terminating execution.
Note that, because of quantification over all contexts, $f$ may use $M$ an arbitrary number of times,
also concurrently or in nested fashion, which is a very expressive form of quantification.
}
% !TEX root =  main.tex

\section{Game semantics\label{sec:gs}}

Game semantics for programming languages involves
two players, called Opponent (O) and Proponent (P),
and the sequences of moves made by them can be viewed as interactions between 
a program (P) and a surrounding context (O).
In this section, we briefly present 
the fully abstract game model for $\fica$ from~\cite{GM08}, which we rely on in the paper.
The games are defined using an auxiliary concept of an arena.
\begin{definition}
An \emph{arena} $A$ is a
triple $\langle{M_A,\lambda_A,\vdash_A}\rangle$ where:
\begin{itemize}
\item $M_A$ is a set of \emph{moves};
\item $\lambda_A:M_A\rarr\makeset{O,P}\times\makeset{Q,A}$
is a function determining for each $m\in M_A$ whether
it is an \emph{Opponent} or a \emph{Proponent move}, 
and a \emph{question} or an \emph{answer};
we write $\lambda_A^{OP},\lambda_A^{QA}$ for the composite
of $\lambda_A$ with respectively the first and second projections;
\item $\vdash_A$ is a binary relation on $M_A$, called \emph{enabling},
satisfying: if $m\vdash_A n$ for no $m$ then $\lambda_A (n)= (O,Q)$,
if $m\vdash_A n$ then $\lambda_A^{OP}(m)\neq\lambda_A^{OP}(n)$,
and if $m\vdash_A n$ then $\lambda_A^{QA}(m)=Q$.
\end{itemize}
\end{definition}
%If $m\vdash_A n$ we say that $m$ \emph{enables} $n$.
We shall write $I_A$ for the set of all moves of $A$ which
have no enabler; such moves are called \emph{initial}.
Note that an initial move must be an Opponent question.
In arenas used to interpret base types all questions are initial and P-moves
answering them are detailed in the table below, where $i\in\makeset{0,\cdots,\imax}$.
\[\renewcommand\arraystretch{0.9}\begin{array}{c|c|c||c|c|c}
~\word{Arena}~ & ~\word{O-question}~  & ~\word{P-answers}~ &
~\word{Arena}~ & ~\word{O-question}~  &~\word{P-answers}~\\
\hline
\sem{\comt} & \mrun & \mdone &
\sem{\expt} & \mq &   i\\[1ex]
\hline
\sem{\vart} & \mread & i & \sem{\semt}
          & \mgrb & \mok \\
          & \mwrite{i} & \mok &
          & \mrls & \mok
\end{array}
\]
More complicated types are interpreted inductively using
the \emph{product} ($A\times B$) 
and \emph{arrow} ($A\Rightarrow B$) constructions, given below.
\[\begin{array}{rcl}
M_{A\times B}      &=& M_A+M_B\\
\lambda_{A\times B}&=& [\lambda_A,\lambda_B]\\
\vdash_{A\times B} &=& \vdash_A+\vdash_B\\
\end{array}\qquad
\begin{array}{rcl}
M_{A\Rightarrow B}      &=& M_A+M_B\\
\lambda_{A\Rightarrow B}&=& [\abra{\lambda_A^{PO},\lambda_A^{QA}},\lambda_B]\\
\vdash_{A\Rightarrow B} &=& \vdash_A+\vdash_B+\makeset{\,(b,a)\mid b\in I_B\textrm{ and }a\in I_A}\\
\end{array}\]
where $\lambda_A^{PO}(m)= O$ iff $\lambda_A^{OP}(m)=P$.
We write $\sem{\theta}$ for the arena corresponding to type $\theta$. Below we draw (the enabling relations of) $A_1=\sem{\comt\rarr\comt\rarr\comt}$
and $A_2=\sem{(\vart\rarr\comt)\rarr\comt}$ respectively, 
using superscripts to distinguish copies of the same move
(the use of superscripts is consistent with our future use of tags in Definition~\ref{def:tags}).
\[
\xymatrix@C=1mm@R=1mm{O &&&\mrun\ar@{-}[d]\ar@{-}[ld]\ar@{-}[lld]\\
P &\mrun^2\ar@{-}[d] &\mrun^1\ar@{-}[d] &\mdone\\
O &\mdone^2 &\mdone^1}\qquad\qquad
\xymatrix@C=1mm@R=1mm{O &&&&\mrun\ar@{-}[d]\ar@{-}[ld]\\
P && &\mrun^1\ar@{-}[d]\ar@{-}[ld]\ar@{-}[lld] &\mdone\\
O &\mread^{11}\ar@{-}[d]&\mwrite{i}^{11}\ar@{-}[d] & \mdone^1\\
P &i^{11}& \mok^{11}}
\]
%An arena is called \emph{flat} if its questions are all initial
%(consequently the P-moves can only be answers). 

Given an arena $A$, we specify next what it means to be a legal play in $A$.
For a start, the moves that players exchange will have to form
a \emph{justified sequence}, which is a finite sequence of moves
of $A$ equipped with pointers. Its first move is always initial and has no
pointer, but each subsequent move $n$ must have a unique pointer to an
earlier occurrence of a move $m$ such that $m\vdash_A n$.  We say that
$n$ is (explicitly) justified by $m$ or, when $n$ is an answer, that
$n$ answers $m$. \cutout{Note that interleavings of several justified sequences 
may not be justified sequences; instead we shall call them \emph{shuffled sequences}.}
If a question does not have an answer in a justified sequence, we say
that it is \emph{pending} in that sequence.  
%and $m_A$ a move from $M_A$.
Below we give two justified sequences from $A_1$ and $A_2$ respectively.

\[
\rnode{A}{\mrun}\,\,\rnode{B}{\mrun}^1\justh{B}{A}\,\,\rnode{C}{\mrun}^2\justh{C}{A}\, \,\rnode{D}{\mdone^1}\justn{D}{B}{140}\,\, \rnode{E}{\mdone^2}\justn{E}{C}{140}\, \,\rnode{F}{\mdone}\justn{F}{A}{155}\qquad
\rnode{A}{\mrun} \,\, \rnode{B}{\mrun^1}\justj{B}{A}\,\, \rnode{C}{\mread^{11}}\justh{C}{B}\,\, \rnode{D}{0^{11}}\justf{D}{C}\, \,
\rnode{E}{\mwrite{1}^{11}}\justh{E}{B}\,\, \rnode{F}{\mok^{11}}\justh{F}{E}\,\, \rnode{G}{\mread^{11}}\justn{G}{B}{160} \,\, \rnode{H}{1^{11}}
\justh{H}{G}
\]

Not all justified sequences are valid.  In order to constitute a legal
play, a justified sequence must satisfy a well-formedness condition
that reflects the ``static'' style of concurrency of our programming
language: any started sub-processes must end before the parent process terminates.
\cutout{any process starting sub-processes must wait for the
children to terminate in order to continue. In game terms: if a
question is answered then that question and all questions justified by
must have been answered (exactly once).} This is formalised as follows,
where the letters $q$ and $a$ to refer to question- and answer-moves
respectively, while $m$ denotes arbitrary moves.
\begin{definition}
The set $P_A$ of  \emph{plays over $A$} 
consists of the justified sequences $s$ over $A$ that satisfy
the two conditions below.
\begin{description}
\item[FORK]: In any prefix $s'= \cdots\rnode{A}{q} \cdots\rnode{B}{m}\justf{B}{A}$ of $s$, the question $q$ must be pending when $m$ is played.
\item[WAIT]: In any prefix $s'= \cdots\rnode{A}{q} \cdots\rnode{B}{a}\justf{B}{A}$ of $s$, all questions justified by $q$ must be answered.
\end{description}
\end{definition}
\cutout{For two shuffled sequences $s_1$ and $s_2$, $s_1\amalg s_2$ denotes
the set of all interleavings of $s_1$ and $s_2$.
For two sets of shuffled sequences $S_1$ and $S_2$,
$S_1\amalg S_2=\bigcup_{s_1\in S_1,s_2\in S_2}s_1\amalg s_2$.
Given a set $X$ of shuffled sequences, we define $X^0=X$, $X^{i+1}=X^i \amalg X$. 
Then $X^\circledast$, called \emph{iterated shuffle} of $X$, is defined to
be $\bigcup_{i\in\N}X^i$. }
It is easy to check that the justified sequences given above are plays.
A subset $\sigma$ of $P_A$ is \emph{O-complete} if $s\in \sigma$
and $s o\in P_A$ imply $so\in\sigma$, when $o$ is an O-move.
\begin{definition}
  A \emph{strategy} on $A$, written $\sigma:A$, is a
  prefix-closed O-complete subset of $P_A$.
\end{definition}
Suppose
$\Gamma=\{x_1:\theta_1,\cdots, x_l:\theta_l\}$ 
and $\seq{\Gamma}{M:\theta}$ is a $\fica$-term.
Let us write $\sem{\seq{\Gamma}{\theta}}$ for the arena $\sem{\theta_1}\times\cdots\times\sem{\theta_l}\Rightarrow\sem{\theta}$.
In~\cite{GM08} it is shown how to assign 
a strategy on $\sem{\seq{\Gamma}{\theta}}$ to any $\fica$-term
$\seq{\Gamma}{M:\theta}$. 
We write $\sem{\seq{\Gamma}{M}}$ to refer to that strategy.
For example, $\sem{\seq{\Gamma}{\divcom}}=\{\epsilon, \mrun\}$
and $\sem{\seq{\Gamma}{\skipcom}} = \{\epsilon,\mrun,\rnode{A}{\mrun}\, \rnode{B}{\mdone}\justf{B}{A}\}$.
Given a strategy $\sigma$,
we denote by $\comp\sigma$ the set of non-empty \emph{complete} plays of $\sigma$, i.e. those in which all questions have been
answered. 
The game-semantic interpretation $\sem{\cdots}$
turns out to provide a fully abstract model
in the following sense.

%\footnote{The language introduced in~\cite{GM08} also features variable and  semaphore constructors $\mathbf{mkvar}$ and $\mathbf{mksem}$, which play a technical role in the full abstraction argument, similarly to~\cite{AM97a}. We omit them in the main body of the paper, because they do not present technical challenges, but they are covered in the Appendix for the sake of completeness.}.
\begin{theorem}[\cite{GM08}]\label{thm:full}
\cutout{$\Gamma\vdash M_1\sqsubsim  M_2$ iff
$\comp{\sem{\Gamma\vdash M_1}}\subseteq \comp{\sem{\Gamma\vdash M_2}}$.}
$\Gamma\vdash M_1\cong  M_2$ iff  $\comp{\sem{\Gamma\vdash M_1}}=\comp{\sem{\Gamma\vdash M_2}}$.
\end{theorem}
In particular, since we have $\comp{\sem{\seq{\Gamma}{\divcom_\theta}}}=\emptyset$,
$\seq{\Gamma}{M:\theta}$ is equivalent to $\divcom_\theta$ iff
$\comp{\sem{\seq{\Gamma}{M}}}=\emptyset$.

% !TEX root =  main.tex

\section{Leafy automata\label{sec:leafy}}

%\adnote{The jump from LA to game semantics is abrupt. Is there a nice way to smooth the transition?}

We would like to be able to represent the game semantics of $\fica$ using automata.
To that end, we introduce \emph{leafy automata} ($\la$).
They are a variant of automata over nested data, i.e. a type of automata that read finite sequences of letters of the form $(t,d_0 d_1\cdots d_j)$ ($j\in \N$), where $t$ is a \emph{tag} from  a finite  set $\Sigma$ and each $d_i$ ($0\le i\le j$) is a \emph{data value}  from an infinite set $\D$.%\footnote{In process algebra, data values are often referred to as \emph{names}.}

In our case, $\D$ will have the structure of a countably infinite forest and the sequences $d_0\cdots d_j$ will correspond to branches of a tree.
Thus,  instead of $d_0\cdots d_j$, we can simply write $d_j$, because $d_j$ uniquely determines its ancestors: $d_0,\dots,d_{j-1}$.
The following definition captures the technical assumptions on~$\D$.
\begin{definition}
$\D$ is a countably infinite set equipped with a function $\predc:\D\rarr\D\cup\{\bot\}$ (the \emph{parent} function) such that the following conditions hold.
\begin{itemize}
\item Infinite branching: $\predc^{-1}(\{d_\bot\})$ is infinite for any $d_\bot\in\D\cup\{\bot\}$.
\item Well-foundedness: for any $d\in\D$, there exists $i\in\N$, called the \emph{level of $d$}, such that $\predc^{i+1}(d)=\bot$.
Level-$0$ data values will be called \emph{roots}.
\end{itemize}

\end{definition}
In order to define configurations of leafy automata, we will rely on finite subtrees of $\D$, whose nodes will be labelled with states.
% We say that $T\subseteq \D$ is a finite subtree of $\D$ iff $T$ is finite, $\pred{x}\in T\cup\{\bot\}$ for any $x\in T$ (closure), and $\pred{x}=\bot$ for at most one $x\in T$ (single root).
We say that $T\subseteq \D$ is a subtree of $\D$ iff $T$ is closed ($\forall x \in T \colon \pred{x}\in T\cup\{\bot\}$) and rooted ($\exists!x\in T\colon\pred{x}=\bot$).

Next we give the formal definition of a level-$k$ leafy automaton.
Its set of states $Q$ will be divided into layers, written $Q^{(i)}$ ($0\le i\le k$), which will be used to label level-$i$ nodes.
We will write $Q^{(i_1,\cdots, i_k)}$ to abbreviate $Q^{(i_1)}\times\cdots \times Q^{(i_k)}$, excluding any components $Q^{(i_j)}$ where $i_j < 0$. We distinguish $Q^{(0,-1)} = \{\dagger\}$.%\igwnote{This paragraph does not bring anything. We can delete it}\adnote{I have moved the definition of $Q^{(i,\cdots,j)}$ to here from the $\sla$ section, so it is now meaningful}

\begin{definition}
A level-$k$ leafy automaton ($k$-$\la$) is a tuple $\Aut=\abra{\Sigma,k,Q,\delta}$, where
\begin{itemize}
\item $\Sigma=\Sigma_\Q+\Sigma_\A$ is a finite alphabet, partitioned into questions and answers;
\item $k\geq 0$ is the level parameter;
\item $Q= \sum_{i=0}^k Q^{(i)}$ is a finite set of states, partitioned into sets $Q^{(i)}$ of level-$i$ states;
\item $\delta=\delta_\Q+\delta_\A$ is a finite transition function, partitioned into question- and answer-related transitions;
% \item $\delta_\Q=\sum_{i=0}^k \delta^{(i)}_{\Q}$, where $\delta^{(i)}_{\Q} \subseteq (\prod_{j=0}^{i-1} Q^{(j)})\times \Sigma_{\Q} \times (\prod_{j=0}^i Q^{(j)})$ for $0\le i\le k$;
% \item  $\delta_{\A}=\sum_{i=0}^k \delta^{(i)}_{\A}$, where $\delta^{(i)}_\A \subseteq (\prod_{j=0}^{i} Q^{(j)})\times \Sigma_{\A} \times  (\prod_{j=0}^{i-1} Q^{(j)})$ for $0\le i\le k$.
% \item $\delta_\Q=\sum_{i=0}^k \delta^{(i)}_{\Q}$, where $\delta^{(i)}_{\Q} \subseteq  Q^{[0,i-1]}\times \Sigma_{\Q} \times Q^{[0,i]}$ for $0\le i\le k$\igwnote{What about this notation?}\adnote{We already use $Q^{(i,i+1,...,j)}$ elsewhere in the paper. Perhaps that can be repurposed here for consistency.};
% \item  $\delta_{\A}=\sum_{i=0}^k \delta^{(i)}_{\A}$, where $\delta^{(i)}_\A \subseteq Q^{[0,i]})\times \Sigma_{\A} \times  Q^{[0,i-1]}$ for $0\le i\le k$.
\item $\delta_\Q=\sum_{i=0}^k \delta^{(i)}_{\Q}$, where $\delta^{(i)}_{\Q} \subseteq  Q^{(0,1,\cdots,i-1)}\times \Sigma_{\Q} \times Q^{(0,1,\cdots,i)}$ for $0\le i\le k$;%\\\igwnote{What about this notation?}\adnote{We already use $Q^{(i,i+1,...,j)}$ elsewhere in the paper. Perhaps that can be repurposed here for consistency.}\adnote{I have done this now}
\item  $\delta_{\A}=\sum_{i=0}^k \delta^{(i)}_{\A}$, where $\delta^{(i)}_\A \subseteq Q^{(0,1,\cdots,i)}\times \Sigma_{\A} \times  Q^{(0,1,\cdots,i-1)}$ for $0\le i\le k$.
\end{itemize}
% We assume $\prod_{j=0}^{-1} Q^{(j)}=\{\dagger\}$.

\igw{Remark for the final version: use $Q^{[0,i]}$ notation}

\end{definition}%\amnote{changed $T$ to $E$ in response to Ranko, to avoid the $t$-$T$ association. In what follows I try to use $t$ consistently to refer to word tags.}
Configurations of $\la$ are of the form $(D,E,f)$, where $D$ is a finite subset of~$\D$ (consisting of data values that have been encountered so far),
$E$ is a finite subtree of $\D$, and $f:E\rarr Q$ is a level-preserving function, i.e.\ if $d$ is a level-$i$ data value then $f(d)\in Q^{(i)}$.
A leafy automaton starts from the empty configuration $\kappa_0=(\emptyset,\emptyset,\emptyset)$ and proceeds according to $\delta$,
making two kinds of transitions. Each kind manipulates a single leaf: for questions one new leaf is added, for answers one leaf is removed.
Let  the current configuration be $\kappa=(D,E,f)$.%\igwnote{Reformulated descriptions of transitions}
\begin{itemize}
\item
On reading a letter $(t,d)$ with $t\in\Sigma_\Q$ and $d\not\in D$ a fresh level-$i$ data, the automaton adds a new leaf $d$ in a configuration and updates the states on the branch to $d$. 
So it changes its configuration to $\kappa'=(D\cup\{d\},E\cup\{d\},f')$ provided that $\pred{d}\in E$ and $f'$ satisfies:
\[
(f(\predc^i(d)),\cdots,f(\pred{d}), t, f'(\predc^i(d)),\cdots,f'(\pred{d}),f'(d))\in\delta^{(i)}_{\Q},
\] 
$\dom{f'}=\dom{f} \cup\{d\}$, and $f'(x)=f(x)$ for all $x\not\in\{\pred{d},\cdots, \predc^i(d)\}$.

\item On reading a letter $(t,d)$ with $t\in\Sigma_\A$ and $d\in E$ a level-$i$ data which is a leaf, the automaton deletes $d$ and updates the states on the branch to $d$.
So it changes its configuration to $\kappa'=(D,E\setminus\{d\},f')$ where $f'$ satisfies:
\[
(f(\predc^i(d)),\cdots,f(\pred{d}),  f(d), t, f'(\predc^i(d)),\cdots,f'(\pred{d}))\in\delta^{(i)}_{\A}, 
\]
$\dom{f'}=\dom{f}\setminus\{d\}$ and $f'(x)=f(x)$ for all $x\not\in\{\pred{d},\cdots, \predc^i(d)\}$.\item Initially $D$,$E$, and $f$ are empty; we proceed to $\kappa' = (\{d\},\{d\},\{d \mapsto q^{(0)}\})$ if $(t,d)$ is read where $\dagger \trans{t} q^{(0)} \in \delta^{(0)}_{\Q}$. The last move is treated symmetrically.
\end{itemize}
% 
% \begin{itemize}
% \item
% Suppose $t\in\Sigma_\Q$ and $d\not\in D$ is a fresh level-$i$ data value with $\pred{d}\in E$. The automaton can read $(t,d)$ and move to $\kappa'=(D\cup\{d\},E\cup\{d\},f')$, where $f'$ must satisfy:
% \[
% (f(\predc^i(d)),\cdots,f(\pred{d}), t, f'(\predc^i(d)),\cdots,f'(\pred{d}),f'(d))\in\delta^{(i)}_{\Q},
% \] 
% $\dom{f'}=\dom{f} \cup\{d\}$, and $f'(x)=f(x)$ for all $x\not\in\{\pred{d},\cdots, \predc^i(d)\}$.
% Intuitively, a new leaf $d$ is added to the tree under an existing node and states on the branch to $d$ are updated.
% \item Suppose $t\in\Sigma_\A$ and $d\in \dom{f}$ is a level-$i$ data with $d\not\in \predc^{-1}(E)$ ($d$ is a leaf in $E$). The automaton can read $(t,d)$ and move to $\kappa'=(D,E\setminus\{d\},f')$ where $f'$ must satisfy:
% \[
% (f(\predc^i(d)),\cdots,f(\pred{d}),  f(d), t, f'(\predc^i(d)),\cdots,f'(\pred{d}))\in\delta^{(i)}_{\A}, 
% \]
% $\dom{f'}=\dom{f}\setminus\{d\}$ and $f'(x)=f(x)$ for all $x\not\in\{\pred{d},\cdots, \predc^i(d)\}$.
% Intuitively,  an existing leaf $d$ is removed and 
% states on the branch that led to $d$ are updated.
% \end{itemize}
In all cases, we write $\kappa\trans{(t,d)}\kappa'$.
Note that a single transition can only change states on the branch ending in $d$. Other parts of the tree remain unchanged.
\begin{example}
Below we illustrate the effect of $\la$ transitions.
Let $D_1=\{d_0,d_1,d_1'\}$ and $d_2\not\in D_1$.
Let $\kappa_1=(D_1,E_1, f_1)$, 
$\kappa_2=(D_1\cup\{d_2\},E_2, f_2)$, $\kappa_3=(D_1\cup\{d_2\}, E_1,f_1)$,
where the trees $E_1, E_2$ are displayed below and node annotations of the form $(q)$ correspond to values
of $f_1, f_2$, e.g. $f_1(d_0)=q^{(0)}$.
\[
\xymatrix@C=1mm@R=2mm{
&&d_0 (q^{(0)})\ar@{-}[ld]\ar@{-}[rd] &\\
E_1,f_1: &d_1' (q) & & d_1 (q^{(1)})\\
}\qquad\qquad
\xymatrix@C=1mm@R=3mm{
&&d_0 (r^{(0)})\ar@{-}[ld]\ar@{-}[rd] &\\
E_2,f_2: &d_1' (q) & & d_1 (r^{(1)})\ar@{-}[d]\\
&&& d_2(r^{(2)})\\}
\]
For $\kappa_1$ to evolve into $\kappa_2$ (on $(t,d_2)$), 
we need $(q^{(0)}, q^{(1)}, t, r^{(0)}, r^{(1)}, r^{(2)})\in \delta^{(2)}_\Q$.
On the other hand, to go from $\kappa_2$ to $\kappa_3$ (on $(t,d_2)$),
we want $(r^{(0)},$ $r^{(1)},$ $r^{(2)},$ $t,$ $q^{(0)},$ $q^{(1)})\in \delta^{(2)}_\A$.
\end{example}
\begin{definition}
A \emph{trace} of a leafy automaton $\Aut$ is a sequence
$w=l_1\cdots l_h\in (\Sigma\times\D)^\ast$ such that $\kappa_0\trans{l_1}\kappa_1\dots\kappa_{h-1}\trans{l_h}\kappa_h$
where $\kappa_0=(\emptyset,\emptyset,\emptyset)$.
A configuration $\kappa=(D,E,f)$  is \emph{accepting} if 
$E$ and $f$ are empty.
A trace $w$ is accepted by $\Aut$ if there is a non-empty sequence of transitions as above with $\kappa_h$  accepting.  
The set of traces (resp. accepted traces) of $\Aut$ is denoted
by $\trace{\Aut}$ (resp. $\lang{\Aut}$).
\end{definition}

\begin{remark}
When writing states, we will often use superscripts $(i)$ to indicate the intended level.
So, $(q^{(0)},\cdots, q^{(i-1)}) \trans{t} (r^{(0)},\cdots, r^{(i)})$ 
refers to
$(q^{(0)},\cdots, q^{(i-1)}, t,$ $r^{(0)},$ $\cdots, r^{(i)})\in\delta^{(i)}_{\Q}$; similarly for $\delta^{(i)}_{\A}$ transitions. 
For $i=0$, this degenerates to $\dagger\trans{t} r^{(0)}$ and $r^{(0)}\trans{t} \dagger$.
\end{remark}

\begin{example}\label{ex:la}
Consider the $1$-$\la$ over $\Sigma_\Q=\{\move{start},\move{inc}\}, \Sigma_\A=\{\move{dec},\move{end}\}$.
Let $Q^{(0)}=\{0\}$, $Q^{(1)}=\{0\}$ and define $\delta$ by:
$\dagger\trans{\move{start}} 0$,\quad $0\trans{\move{inc}} (0,0)$,\quad $(0,0)\trans{\move{dec}} 0$,\quad $0\trans{\move{end}} \dagger$.
The accepted traces of this $1$-$\la$ have the form $(\move{start},d_0)\,\, (||_{i=0}^n (\move{inc}, d_1^i)$ $(\move{dec},d_1^i))\,\, (\move{end},d_0)$,
i.e.\ they are valid histories of a single {non-negative} counter (histories such that the counter starts and ends at 0). In this case, all traces are simply prefixes of such words.
\end{example}

\begin{remark}\label{rem:lawork}
Note that, whenever a leafy automaton reads $(t,d)$ ($t\in\Sigma_\Q$)
and the level of $d$ is greater than $0$, then it must have read 
a unique question $(t',\pred{d})$ earlier.
Also, observe that an $\la$ trace contains at most two occurrences of the same data value, such that the first is paired with a question and the second is paired with an answer. Because the question and the answer share the same data value, we can think of the answer as answering the question, like in game semantics. Indeed, justification pointers from answers to questions will be represented in this way in Theorem~\ref{thm:trans}. Finally, we note that $\la$ traces are invariant under tree automorphisms of $\D$.
\end{remark}

\begin{lemma}\label{lem:la2-1}
The emptiness problem for $2$-$\la$ is undecidable.  For $1$-$\la$, it is reducible to the reachability problem for VASS in polynomial time and there is a reverse reduction in exponential time, so it is decidable in Ackermannian time~\cite{LerouxS19} but not elementary~\cite{CzerwinskiLLLM19}.
\end{lemma}
\begin{proof}

For $2$-$\la$ we reduce from the halting problem on two-counter-machines. 
Two counters can be simulated using configurations of the form
\[\xymatrix@C=.5em@R=.5em{
  &            &     &q\ar@{-}[lld]\ar@{-}[rrd] & &\\
  & c_1\ar@{-}[ld]\ar@{-}[d]\ar@{-}[rd] &     &  &   & c_2\ar@{-}[ld]\ar@{-}[d]\ar@{-}[rd] \ar@{-}[rrd] &\\
\star & \star         & \star &  & \star  & \star        &\star & \star
}\]
where there are two level-$1$ nodes, one for each counter. 
The number of children at level $2$ encodes the counter value.
Zero tests can be implemented by removing the corresponding level-$1$ node and creating a new one.
This is possible only when the node is a leaf, i.e., it does not have children at level~$2$.
The state of the 2-counter machine can be maintained at level~$0$, the states at level $1$ indicate the name of the counter, and the level-$2$ states are irrelevant. 

The translation from $1$-$\la$ to VASS is straightforward and based on representing $1$-$\la$ configurations by the state at level~$0$ and, for each state at level~$1$, the count of its occurrences.  The reverse translation is based on the same idea and extends the encoding of a non-negative counter in Example~\ref{ex:la}, where the exponential blow up is simply due to the fact that vector updates in VASS are given in binary whereas $1$-$\la$ transitions operate on single branches.\qed
\end{proof}

\begin{lemma}\label{lem:la1}
$1$-$\la$ equivalence is undecidable.
\end{lemma}
\begin{proof}

We provide a direct reduction from the halting problem for 2-counter machines, where both counters are required to be zero initially as well as finally.  The main obstacle is that implementing zero tests as in the proof of the first part of Lemma~\ref{lem:la2-1} is not available because we are restricted to leafy automata with levels $0$ and $1$ only.  To overcome it, we exploit the power of the equivalence problem where one of the $1$-$\la$ will have the task not of correctly simulating zero tests but recognising zero tests that are incorrect. The full argument can be found in Appendix~\ref{apx:leafy}.\qed
\end{proof}

\section{Local leafy automata ($\lla$)}
\label{sec:lla}

Here we identify a restricted variant of $\la$
for which the emptiness problem is decidable. We start with a technical definition.
\begin{definition}
A $k$-$\la$ is \emph{bounded} at level $i$ ($0\le i\le k$) if 
there is a bound $b$ such that each node at level $i$ can create at most $b$ children during a run.
We refer to $b$ as the \emph{branching bound}.
\end{definition}
% Note that above we do not mean a ``local" bound of the number of children in a configuration, but a ``global" bound on the number of children produced by transitions in a whole run.
Note that we are defining a ``global'' bound on the number of children that a node at level $i$ may create across a whole run, rather than a ``local'' bound on the number of children a node may have in a given configuration.

To motivate the design of $\lla$, we observe that the undecidability argument (for the emptiness problem) for $2$-$\la$ used two consecutive levels ($0$ and $1$) that are not bounded.
For the node at level $0$, this corresponded to the number of zero tests, while an unbounded counter is simulated at level $1$. 
In the following we will eliminate consecutive unbounded levels by introducing an alternating pattern of bounded and unbounded levels. Even-numbered layers ($i=0, 2, ...$) will be bounded, while odd-numbered layers will be unbounded. Observe in particular that the root (layer $0$) is bounded. As we will see later, this alternation reflects the term/context distinction in game semantics: the levels corresponding to terms are bounded, and the levels coresponding to contexts are unbounded.

%changed 'reasonable' to 'certain' to make the statement more neutral

With this restriction alone, it is possible to reconstruct the undecidability argument for $4$-$\la$, as two unbounded levels may still communicate. Thus we introduce a restriction on how many levels a transition can read and modify.
\begin{itemize}
    \item when adding or removing a leaf at an odd level $2i+1$, the automaton will be able to
    access levels $2i$, $2i-1$ and $2i-2$; while
    \item when adding or removing a leaf at an even level $2i$, the automaton
    will be able to access levels $2i-1$ and $2i-2$.
\end{itemize}
In particular, when an odd level produces a leaf, it will not be able to see the previous odd level.
The above constraints mean that the transition functions $\delta^{(i)}_{\Q}, \delta^{(i)}_{\Q}$ can be
presented in a more concise form, given below.

\[
\delta^{(i)}_{\Q}\subseteq \begin{cases}
        Q^{(i-2, i-1)} \times \Sigma_\Q \times Q^{(i-2, i-1, i)} 
            & \text{if $i$ is even} \\
        Q^{(i-3, i-2, i-1)} \times \Sigma_\Q \times Q^{(i-3, i-2, i-1, i)} 
            & \text{if $i$ is odd}
            \end{cases}
\]
   \[
    \delta^{(i)}_\A \subseteq\begin{cases}
        Q^{(i-2, i-1, i)} \times \Sigma_\A \times Q^{(i-2, i-1)} 
            & \text{if $i$ is even} \\
        Q^{(i-3, i-2, i-1, i)} \times \Sigma_\A \times Q^{(i-3, i-2, i-1)} 
            & \text{if $i$ is odd}
    \end{cases}
    \]
    %Note that the  relations above contain smaller tuples than the transition relation for %standard $\la$. 
    %This corresponds to the constraint that $\sla$ may only access ancestors close to a newly %created node or a node to be removed, 
    %and may not make transitions on the basis of information that is not available at that level. 
In terms of the previous notation developed for $\la$, 
$(q^{(i-2)}, q^{(i-1)},x,r^{(i-2)}, r^{(i-1)},$ $r^{(i)})\in\delta^{(i)}_\Q$ represents
all tuples of the form $(\vec{q}, q^{(i-2)}, q^{(i-1)}, x, \vec{q}, r^{(i-2)},r^{(i-1)}, r^{(i)})$,
where $\vec{q}$ ranges over $Q^{(0,\cdots,i-3)}$.
\begin{definition}
A level-$k$ \emph{local leafy automaton} ($k$-$\lla$) is a $k$-$\la$ whose transition function admits the above-mentioned
presentation and which is bounded at all even levels.
\end{definition}

\newcommand{\interrupt}{\ensuremath{\mathsf{INT}}}

\begin{theorem}
\label{thm:sla-decidable}
The emptiness problem for $\lla$ is decidable.
\end{theorem}
\begin{proof}[Sketch]
Let $b$ be a bound on the number of children created by each even node during a run.

The critical observation is that, once a node $d$ at even level $2i$ has been created, all subsequent actions of descendants of $d$ access (read and/or write) the states at levels $2i-1$ and $2i-2$ at most $2b$ times. The shape of the transition function dictates that this can happen only when child nodes at level $2i+1$ are added or removed. 
In addition, the locality property ensures that the automaton will never access levels $< 2i-2$ at the same time as node $d$ or its descendants. 
% Since we know that the number of $d$'s children is bounded by $b$, there will be at most $2b$ instances at which the automaton will read/write states at levels $2i-1$ and $2i-2$ when operating on descendants of $d$, and it will never access levels $< 2i-2$ at the same time as $d$ or its descendant. 

We will make use of these facts to construct \emph{summaries} for nodes on even levels which completely describe such a node's lifetime, from its creation as a leaf until its removal, and in between performing at most $2b$ reads-writes of the parent and grandparent states.
A summary is a sequence quadruples of states: two pairs of states of levels $2i-2$ and $2i-1$. The first pair are the states we expect to find on these levels, while the second are the states to which we update these levels. Hence a summary at level $2i$ is a complete record of a valid sequence of read-writes and stateful changes during the lifetime of a node on level $2i$.

% To capitalise on the observation, we are going to compute summaries for nodes at even nodes, describing their lifetime, i.e. from the point of creation to a point where they are leaves again.

% Thus, our summaries will take the form
% \[
% S^{(2i)}_x = Q^{(2i)} \times (Q^{(2i-2, 2i-1)} \times Q^{(2i-2, 2i-1)})^{\leq 2b} \times Q^{(2i)}
% \]
% where the first and last component are the initial and final states of the summary, and the at most $2b$ components in-between represent at most $2b$ reads-writes of the parent and grandparent states.

We proceed by induction and show how to calculate the complete set of summaries at level $2i$ given the complete set of summaries at level $2i+2$. We construct a program for deciding whether a given sequence is a summary at level $2i$.
This program can be evaluated via Vector Addition Systems with States (VASS). Since we can finitely enumerate all candidate summaries at level $2i$, this gives us a way to compute summaries at level $2i$. Proceeding this way, we finally calculate summaries at level $2$.
At this stage, we can reduce the emptiness problem for the given $\lla$ to a reachability test on a VASS.

% , Then the reachability problem for VASS
% allows us to decide whether a particular tuple 
% is a genuine summary at level $2k$, assuming 
% summaries at level $2k+2$ have already been computed correctly.
The complete argument is given in Appendix~\ref{apx:sla}.
\qed
\end{proof}

Let us remark also that the problem becomes undecidable if we remove either boundedness restriction, or allow transitions to look one level further.
% !TEX root =  main.tex
\section{From FICA to LA}
\label{sec:tola}

%\igw{We need an introduction here. I am not sure if we follow my suggestion to move also the other direction of the correspondance here. }\am{OK, Section 8 stays where it is.}

Recall from Section~\ref{sec:gs} that, to interpret base types, game semantics uses moves from the set
\[\begin{array}{rcl}
\moveset &=& M_{\sem{\comt}}\cup M_{\sem{\expt}}\cup M_{\sem{\vart}} \cup M_{\sem{\semt}}\\
&=&\{\, \mrun,\, \mdone,\, \mq,\, \mread,\, \mgrb,\, \mrls,\, \mok\, \}\cup \{\,i,\, \mwrite{i}{}\,|\, 0\le i \le \max\,\}.
\end{array}\]
The game semantic interpretation of 
a term-in-context $\seq{\Gamma}{M:\theta}$ is a strategy over the arena $\sem{\seq{\Gamma}{\theta}}$,
which  is obtained through product and arrow constructions,  starting from arenas corresponding to base types.
As both constructions rely on the disjoint sum, the moves from $\sem{\seq{\Gamma}{\theta}}$ are derived
from  the base types present in types inside $\Gamma$ and $\theta$.
To indicate the exact occurrence of a base type from which each move originates, we will  annotate elements of $\moveset$ with
a specially crafted scheme of superscripts.
Suppose  $\Gamma=\{x_1:\theta_1,\cdots, x_l:\theta_l\}$.
The superscripts will have one of the two forms,  where $\vec{i}\in\N^\ast$ and $\rho\in\N$:
\begin{itemize}
\item $(\vec{i},\rho)$ will be used to represent moves from $\theta$;
\item $(x_v\vec{i}, \rho)$ will be used to represent moves from $\theta_v$ ($1\le v\le l$).
\end{itemize}
The annotated moves will be written as $m^{(\vec{i},\rho)}$ or $m^{(x_v\vec{i},\rho)}$, where $m\in\moveset$.
We will sometimes omit $\rho$ on the understanding that this represents $\rho=0$.
Similarly, when $\vec{i}$ is omitted, the intended value is~$\epsilon$. Thus, $m$ stands for $m^{(\epsilon,0)}$.

The next definition explains how the $\vec{i}$ superscripts are
linked to moves from $\sem{\theta}$.
Given $X\subseteq \{ m^{(\vec{i},\rho)} \,|\, \vec{i}\in\N^\ast,\,\rho\in\N\}$ and $y\in \N\cup \{x_1,\cdots, x_l\}$, 
we let $yX = \{m^{(y\vec{i},\rho)}\,|\, m^{(\vec{i},\rho)}\in X\}$.%\amnote{definition modified in response to Ranko}
\begin{definition}\label{def:tags}
Given a type $\theta$, the corresponding alphabet $\alp{\theta}$ is defined as follows
\[\begin{array}{rcl}
\alp{\beta}&=&\{\, m^{(\epsilon,\rho)}\,|\, m\in M_{\sem{\beta}},\,\rho\in\N\,\}\qquad \beta=\comt,\expt,\vart,\semt\\
\alp{\theta_h\rarr\ldots\rarr\theta_1\rarr\beta}&=& \bigcup_{u=1}^h (u\alp{\theta_u}) \cup \alp{\beta}
\end{array}\]
For $\Gamma=\{x_1:\theta_1,\cdots, x_l:\theta_l\}$,
the alphabet $\alp{\seq{\Gamma}{\theta}}$ is defined to be 
$\alp{\seq{\Gamma}{\theta}}=\bigcup_{v=1}^l (x_v \alp{\theta_v}) \cup \alp{\theta}$.
\end{definition}
\begin{example}
The alphabet  $\alp{\seq{f:\comt\rarr\comt, x:\comt}{\comt}}$ is
  \[\{ \mrun^{(f1,\rho)}, \mdone^{(f1,\rho)}, \mrun^{(f,\rho)}, \mdone^{(f,\rho)}, \mrun^{(x,\rho)},\mdone^{(x,\rho)},\mrun^{(\epsilon,\rho)},\mdone^{(\epsilon,\rho)}\,|\, \rho\in \N \}.\]
\end{example}%\amnote{added a sentence explaining 'finite subsets'}
To represent the game semantics of terms-in-context, of the form $\seq{\Gamma}{M:\theta}$,
we are going to use \emph{finite subsets} of $\alp{\seq{\Gamma}{\theta}}$ as alphabets in leafy automata.
The subsets will be finite, because $\rho$ will be bounded.
Note that $\alp{\theta}$ admits a natural partitioning into questions and answers, depending on whether the underlying move is a question or answer. 

%\igwnote{I would not make (Play represenation a deinfition. It is long, difficult to read in italic, and mixed between formal and explanations. Just make it normal text. I understand this creates some problem in Thm21 but it can be dealt with. I have not implemented this as I understand that this is controversial.}

%\paragraph{Play representation}
We will represent plays using data words in which the underpinning
sequence of tags will come from an alphabet as defined above.
Superscripts and data are used to represent justification pointers.
Intuitively, we represent occurrences of questions with data values.
Pointers from answers to questions just refer to these values.
Pointers from questions use bounded indexing with the help of~$\rho$.
%correspond to the sequence of moves in the play.
%Next we explain how $\rho$ is used to represent justification pointers.

%\igwnote{My rewriting of the next para}
Initial question-moves do not have a pointer and to represent such questions we simply use $\rho=0$. 
For non-initial questions,
we rely on the tree structure of $\D$ and use $\rho$ to indicate the ancestor of the currently read data value that we mean to point at.
Consider a trace $w (t_i,d_i)$ ending in a non-initial question, where $d_i$ is a level-$i$  data value and $i>0$. 
In our case, we will have $t_i\in\alp{\seq{\Gamma}{\theta}}$, i.e.
$t_i=m^{(\cdots, \rho)}$.
By Remark~\ref{rem:lawork}, trace $w$ contains unique occurrences of questions
$(t_0,d_0), \cdots, (t_{i-1},d_{i-1})$ such that $\pred{d_j}=d_{j-1}$ for $j=1,\cdots, i$.
The pointer from $(t_i,d_i)$ goes to one of these questions, and we use $\rho$ to represent
the scenario in which the pointer goes to  $(t_{i-(1+\rho)},d_{i-(1+\rho)})$.

Pointers from answer-moves to question-moves are represented 
simply by using the same data value in both moves (in this case we use $\rho=0$).

% Initial question-moves do not have a pointer and to represent such questions we simply use $\rho=0$. Next we discuss non-initial questions.
% Recall from Remark~\ref{rem:lawork}
% that whenever a leafy automaton has read a trace $w (t_i,d_i)$, where $d_i$ is a level-$i$ data value and $i>0$, then $w$ will contain unique occurrences of questions
% $(t_0,d_0), \cdots, (t_{i-1},d_{i-1})$ such that $\pred{d_j}=d_{j-1}$ for $j=1,\cdots, i$.
%  In our case, we will have $t_i\in\alp{\seq{\Gamma}{\theta}}$, i.e.
%  $t_i=m^{(\cdots, \rho)}$.
%  We shall use the $\rho$ superscript to represent justification pointers from (non-initial) question-moves as follows:
%  if $t_i=m^{(\cdots, \rho)}$ and 
%  $m$ is a question, then $\rho$ represents a justification pointer from $(t_i,d_i)$ to $(t_{i-(1+\rho)},d_{i-(1+\rho)})$,
%  i.e., $\rho$ indicates that, to find the pointer, 
%  we should move $\rho+1$ levels up from $d_i$ in $\D$.
% Finally, pointers from answer-moves to question-moves will be represented by the fact that the corresponding tags from $\alp{\seq{\Gamma}{\theta}}$ 
%  will occur with the same data value in the representing data word.
%  In this case, we will use $\rho=0$.
 
 We will also use  $\epsilon$-tags $\eq$ (question) and $\ea$ (answer), which do not contribute moves to the represented play. Each $\eq$ will always be answered with $\ea$. Note that the use of $\rho,\eq,\ea$ means that several data words may represent the same play (see Examples~\ref{ex:play},~\ref{ex:play2}).
   \begin{example}\label{ex:play}
   Suppose that $d_0=\pred{d_1}, d_1=\pred{d_2}=\pred{d_2'}, d_2=\pred{d_3}$, $d_2'=\pred{d_3'}$.
Then the data word
$(\mrun,d_0)$ $(\mrun^f,d_1)$ $(\mrun^{f1}, d_2)$ $(\mrun^{f1}, d_2')$ $(\mrun^{(x,2)},d_3)$ $(\mrun^{(x,2)}, d_3')$ $(\mdone^x,d_3)$,
which is short for
$(\mrun^{(\epsilon,0)},d_0)$ $(\mrun^{(f,0)},d_1)$ $(\mrun^{(f1,0)}, d_2)$ $(\mrun^{(f1,0)}, d_2')$ $(\mrun^{(x,2)},d_3)$ $(\mrun^{(x,2)}, d_3')$ $(\mdone^{(x,0)},d_3)$,
represents the play
\medskip
\[\begin{array}{ccccccc}
\rnode{Z}{\mrun} &
\rnode{A}{\mrun^f}\justf{A}{Z} &
\rnode{B}{\mrun^{f1}}\justf{B}{A} &
\rnode{C}{\mrun^{f1}}\justn{C}{A}{140} &
\rnode{D}{\mrun^{x}}\justn{D}{Z}{150} &
\rnode{E}{\mrun^{x}} \justn{E}{Z}{150} &
\rnode{F}{\mdone^x}\justn{F}{D}{140}\\
O &P&O&O&P &P & O.
\end{array}\]
\end{example}
\begin{example}
Consider the $\la$ $\Aut=\abra{Q,3,\Sigma,\delta}$,
where $Q^{(0)}=\{0,1,2\}$, $Q^{(1)}=\{0\}$, 
$Q^{(2)}=\{0,1,2\}$, $Q^{(3)}=\{0\}$,
$\Sigma_\Q=\{\mrun,\mrun^f,\mrun^{f1},\mrun^{(x,2)}\}$,
$\Sigma_\A=\{\mdone,\mdone^f,$ $\mdone^{f1},\mdone^x\}$,
and $\delta$ is given by
\[\begin{array}{c}
\dagger\trans{\mrun} 0\qquad
0 \trans{\mrun^f}{(1,0)}\qquad
(1,0) \trans{\mdone^f} 2 \qquad
2 \trans{\mdone} \dagger \qquad
(1,0) \trans{\mrun^{f1}} (1,0,0)\\
(1,0,0)\trans{\mrun^{(x,2)}} (1,0,1,0)\qquad
(1,0,1,0) \trans{\mdone^{(x,0)}} (1,0,2)\qquad
(1,0,2) \trans{\mdone^{f1}} (1,0)
\end{array}\]
Then traces from $\trace{\Aut}$  represent
all plays from $\sigma=\llbracket f:\comt\rarr\comt,\, x:\comt \,\vdash\, f x
% :\com
\rrbracket$, including the play from Example~\ref{ex:play},
and $\lang{\Aut}$ represents $\comp{\sigma}$.
\end{example}

\begin{example}\label{ex:play2}
One might wish to represent plays of  
$\sigma$ from the previous Example
using data values  
$d_0,d_1,d_1',d_1'',d_2,d_2'$ such that
$d_0=\pred{d_1}=\pred{d_1'}=\pred{d_1''}$, $d_1=\pred{d_2}=\pred{d_2'}$,
so that the play from Example~\ref{ex:play} is represented
by
$(\mrun^{(\epsilon,0)},d_0)$ $(\mrun^{(f,0)},d_1)$ $(\mrun^{(f1,0)}, d_2)$ $(\mrun^{(f1,0)}, d_2')$ $(\mrun^{(x,0)},d_1')$ $(\mrun^{(x,0)}, d_1'')$ $(\mdone^{(x,0)},d_1')$.
Unfortunately, it is impossible to construct a $2$-$\la$ that would 
accept all representations of such plays. To achieve this, 
the automaton would have to make sure that the number of $\mrun^{f1}$s is the same
as that of $\mrun^x$s. Because the former are labelled with level-$2$
values and the latter with incomparable level-$1$ values,
the only point of communication (that could be used for comparison)
is the root. However, the root cannot
accommodate unbounded information, while plays of $\sigma$
can feature an unbounded number of $\mrun^{f1}$s, which could well be consecutive.
\end{example}

%\igwnote{My rewriting of the paragraphs before the theorem}
%\amnote{rephrased a bit - there is no P=even, O=odd correspondence, we have even=OQ,PA, odd=PQ,OA}

Before we state the main result linking $\fica$ with leafy automata, we note some structural properties 
of the automata.
Questions will create a leaf, and answers will remove a leaf.
P-moves add leaves at odd levels (questions) and remove leaves at even levels (answers),
while O-moves have the opposite effect at each level.
Finally, when removing nodes at even levels we will not need to check if a node is a leaf.
We call the last property \emph{even-readiness}.

Even-readiness is a consequence of the WAIT condition in the game semantics.
The condition captures well-nestedness of concurrent interactions -- 
a term  can terminate only after subterms terminate.
In the leafy automata setting, this is captured by the requirement that only leaf nodes can be removed, i.e. a node can be removed only if
all of its children have been removed beforehand. 
It turns out that, for \emph{P-answers} only, 
this property will come for free. 
Formally, whenever the automaton arrives at a configuration
$\kappa=(D,E,f)$, where $d\in E$ and there is a transition
\[
(f(\predc^{(2i)}(d)),\cdots,f(\pred{d}),  f(d), t, f'(\predc^{(2i)}(d)),\cdots,f'(\pred{d}))\in\delta^{(2i)}_{\A}, 
\]
then $d$ is a leaf.
In contrast, our automata will not satisfy the same property for O-answers (the environment) and for such transitions it is crucial that
the automaton actually checks that only leaves can be removed.

\begin{theorem}\label{thm:trans}
For any $\fica$-term $\seq{\Gamma}{M:\theta}$, there exists an  even-ready leafy automaton $\Aut_M$ 
over a finite subset of $\alp{\seq{\Gamma}{\theta}}+\{\eq,\ea\}$ such that the set of plays represented by data words from $\trace{\Aut_M}$ 
is exactly $\sem{\seq{\Gamma}{M:\theta}}$. Moreover,
$\lang{\Aut_M}$ represents $\comp{\sem{\seq{\Gamma}{M:\theta}}}$
in the same sense.
\end{theorem}
\begin{proof}[Sketch]
Because every $\fica$-term can be converted to $\beta\eta$-normal form, we use induction on the structure of such normal forms. 
The base cases are: $\seq{\Gamma}{\skipcom:\comt}$ ($Q^{(0)}= \{0\}$;
$\dagger \trans{\mrun} 0$,
$0 \trans{\mdone} \dagger$), $\seq{\Gamma}{\divcom:\comt}$ ($Q^{(0)}= \{0\}$; $\dagger \trans{\mrun} 0$),
and $\seq{\Gamma}{i:\expt}$ ($Q^{(0)}= \{0\}$; 
$\dagger \trans{\q} 0$, $0 \trans{i} \dagger$).

The remaining cases are inductive.
When referring to the inductive hypothesis for a subterm $M_i$,
we shall use subscripts $i$ to refer to the automata components, 
e.g. $Q_i^{(j)}$,  $\trans{\mm}_i$ etc.
In contrast, $Q^{(j)}$, $\trans{\mm}$ will refer to the automaton that is being constructed.
Inference lines $\frac{\qquad}{\qquad}$ will indicate that the transitions listed under the line should be added
to the new automaton provided the transitions listed above the line are present in the automaton obtained via induction hypothesis. We discuss a selection of technical cases below.

\paragraph{$\seq{\Gamma}{M_1|| M_2}$}
In this case we need to run the automata for $M_1$ 
and $M_2$ concurrently. To this end, their level-$0$ states will be combined ($Q^{(0)} =  Q_1^{(0)} \times Q_2^{(0)}$),
but not deeper states ($Q^{(j)}= Q_1^{(j)}+Q_2^{(j)}, 1\le j\le k$).
The first group of transitions activate and terminate the two components respectively:
$\frac{\dagger\trans{\mrun}_1 q_1^{(0)}\qquad \dagger\trans{\mrun}_2 q_2^{(0)}}{\dagger\trans{\mrun}(q_1^{(0)},q_2^{(0)}) }$,
$\frac{q_1^{(0)}\trans{\mdone}_1 \dagger\qquad q_2^{(0)}\trans{\mdone}_2 \dagger}{(q_1^{(0)},q_2^{(0)})\trans{\mdone}\dagger}$.
The remaining transitions advance each component:
$\frac{(q_1^{(0)}, \cdots, q_1^{(j)}) \trans{\mm}_1 (r_1^{(0)},\cdots, r_1^{(j')})\qquad q_2^{(0)}\in Q_2^{(0)}}{((q_1^{(0)},q_2^{(0)}), \cdots, q_1^{(j)}) \trans{\mm} ((r_1^{(0)},q_2^{(0)}),\cdots, r_1^{(j')})}$,
$\frac{q_1^{(0)}\in Q_1^{(0)}\qquad (q_2^{(0)}, \cdots, q_2^{(j)}) \trans{\mm}_2 (r_2^{(0)},\cdots, r_2^{(j')})
}{((q_1^{(0)},q_2^{(0)}), \cdots, q_2^{(j)}) \trans{\mm} ((q_1^{(0)},r_2^{(0)}),\cdots, r_2^{(j')})}$, where $\mm\neq\mrun,\mdone$.

\paragraph{$\seq{\Gamma}{\newin{x\aasg i}{M_1}}$} 
%\igwnote{reformulated}
By~\cite{GM08}, the semantics of this term is obtained from the semantics of $\sem{\seq{\Gamma,x}{M_1}}$ by 
\begin{enumerate}
\item  restricting to plays in which the moves $\mread^x$, $\mwrite{n}^x$ are followed  immediately by answers,  
\item selecting those plays in which each answer to a $\mread^x$-move is consistent
with the preceding $\mwrite{n}^x$-move (or equal to $i$, if no $\mwrite{n}^x$ was made),
\item erasing all moves related to $x$, e.g. those of the form $m^{(x,\rho)}$. 
\end{enumerate}
To implement 1., we will lock the automaton after each $\mread^x$- or $\mwrite{n}^x$-move, so that only an answer to that move can be played next. Technically, this will be done by adding  an extra bit (lock) to the level-$0$ state.
To deal with 2., we keep track of the current value of $x$, also at level $0$. 
This makes it possible to ensure that
 answers to $\mread^x$  are consistent  with the stored value and that $\mwrite{n}^x$ transitions cause the right change.
%\igwnote{reformulated}
Erasing from condition 3 is implemented by replacing all moves with the $x$ subscript with $\eq,\ea$-tags.

Accordingly, 
we have $Q^{(0)}=(Q_1^{(0)} + (Q_1^{(0)}\times \{\mathit{lock}\})) \times\{0,\cdots,\imax\}$ and $Q^{(j)} = Q_1^{(j)}$ ($1\le j\le k$).
As an example of a transition, we give the transition related to writing:
$\frac{(q_1^{(0)},\cdots, q_1^{(j)})\trans{\mwrite{z}^{(x,\rho)}}_1 (r_1^{(0)},\cdots, r_1^{(j')})\qquad 0\le n,z\le \imax}{
((q_1^{(0)},n),\cdots, q_1^{(j)})\trans{\eq} ((r_1^{(0)},\lock, z),\cdots, r_1^{(j')})}$.

\paragraph{$\seq{\Gamma}{f M_h \cdots M_1:\comt}$ with $(f: \theta_h\rarr\cdots\rarr\theta_1\rarr\comt)$}
Here we will need $Q^{(0)} = \{0,1,2\}$, $Q^{(1)}=\{0\}$,  $Q^{(j+2)}= \sum_{u=1}^{h} Q_u^{(j)}$ ($0\le j\le k$).
The first group of transitions corresponding to calling 
and returning from $f$: $\dagger \trans{\mrun} 0$,\quad  $0\trans{\mrun^f} (1,0)$,\quad 
$(1,0)\trans{\mdone^f} 2$,\quad $2\trans{\mdone} \dagger$.
Additionally, in state $(1,0)$ we want to enable the environment to spawn an unbounded number of copies of each of $\seq{\Gamma}{M_u:\theta_u}$ ($1\le u\le h$).
This is done through rules that embed the actions of the automata for $M_u$ while (possibly) relabelling the moves in line with our convention for representing moves from game semantics. 
Such transitions have the general form
$\frac{ (q_u^{(0)},\cdots, q_u^{(j)}) \trans{m^{(t,\rho)}}_u (q_u^{(0)},\cdots, q_u^{(j')})}{(1,0,q_u^{(0)},\cdots, q_u^{(j)}) \trans{m^{(t',\rho')}} (1,0,q_u^{(0)},\cdots, q_u^{(j')})}$.
Note that this case also covers $f:\comt$ ($h=0$).

More details and the remaining cases are covered in Appendix~\ref{apx:tola}.
In Appendix~\ref{apx:example} we give an example of a term and the corresponding $\la$.
\qed
\end{proof}

\section{Local $\fica$}
\label{sec:tosla}

In this section we identify a family of $\fica$ terms
that can be translated into $\sla$ rather than $\la$. To achieve boundedness at even levels,
we remove $\mathsf{while}$\footnote{The automaton for $\while{M}{N}$ may repeatedly visit the automata for
$M$ and $N$, generating an unbounded number of children at level $0$ in the process.}.
To achieve restricted communication, we will constrain
the distance between a variable declaration and its use.
Note that in the translation, the application of function-type variables increases $\la$ depth.
So in $\sfica$ we will allow the link 
between the binder $\mathbf{newvar}/\mathbf{newsem}\, x$ and each use of $x$ to ``cross" at most one occurrence of a free variable.
For example, the following terms
\begin{itemize}
\item    $\newin{x\aasg 0}{x\aasg 1\, ||\, f(x\aasg 2)}$,
\item $\newin{x\aasg 0}{f(\newin{y}{f(y\aasg 1)\, ||\, x:=!y})}$
\end{itemize}
will be allowed, but not $\newin{x\aasg 0}{f(f(x\aasg 1))}$.

To define the fragment formally, given a term $Q$ in $\beta\eta$-normal form, 
we use a notion of the \emph{applicative depth of a variable $x:\beta$ ($\beta=\vart,\semt$) %\igwnote{I would remove $\beta$}
inside $Q$}, written $\ade{x}{Q}$ and defined inductively by the table below. The applicative depth is increased whenever a functional 
identifier is applied to a term containing~$x$.
\[\begin{array}{lcl}
\textrm{shape of $Q$} && \ade{x}{Q} \\
\hline
x && 1\\
y\, (y\neq x),\, \skipcom,\,\divcom,\, i &\quad & 0\\
\arop{M},\, !M,\, \rls{M},\, \grb{M} && \ade{x}{M}\\
M;N,\, M||N,\, M\aasg N,\,\while{M}{N} & & \max(\ade{x}{M},\ade{x}{N})\\
{\cond{M}{N_1}{N_2}} && \max(\ade{x}{M},\ade{x}{N_1},\ade{x}{N_2})\\
{\lambda y.M}, \newin{\textbf{/newsem}\,y\aasg i}{M} && \ade{x}{M[z/y]},\textrm{where $z$ is fresh}\\
{f M_1\cdots M_k} && 1+ \max(\ade{x}{M_1},\cdots,\ade{x}{M_k})
\end{array}
\]
%\igwnote{proposal for the last line}

Note that in our examples above, in the first two cases the applicative depth of $x$ is $2$; and in the third case it is $3$.

% Note that  $\ade{x}{x\aasg 1 || f(x\aasg 2)} = 
% \ade{x}{f(\newin{y}{(f(y\aasg 1) || x:=!y)})}=2$ and
% $\ade{x}{f(f(x\aasg 1))}=3$.

\begin{definition}[Local $\fica$]
A  $\fica$-term $\seq{\Gamma}{M:\theta}$ is \emph{local} if its $\beta\eta$-normal form 
does not contain any occurrences of $\mathbf{while}$ and, 
for every subterm of the normal form of the shape $\newin{/\mathbf{newsem}\, x\aasg i}{N}$, we have $\ade{x}{N} \le 2$.
We write $\lfica$ for the set of local $\fica$ terms.
\end{definition}

\begin{theorem}\label{thm:trans2}
For any $\lfica$-term $\seq{\Gamma}{M:\theta}$,
the automaton $\clg{A}_M$ obtained from the translation in Theorem~\ref{thm:trans} 
can be presented as a $\sla$.
\end{theorem}
\begin{proof}[Sketch]
We argue by induction that the constructions from Theorem~\ref{thm:trans} preserve presentability as a $\sla$.

The case of parallel composition involves running copies of $M_1$ and $M_2$ in parallel without communication,
with their root states stored as a pair at level $0$. Note, though, that each of the automata transitions independently
of the state of the  other automaton.
In consequence, if the automata $M_1$ and $M_2$ are $\sla$, 
so will be the automaton for $M_1 || M_2$. 
The branching bound after the construction is the sum of the two bounds  for $M_1$ and $M_2$.

For $\seq{\Gamma}{\newin{x\aasg i}{M}}$, because the term is in $\lfica$,
so is $\seq{\Gamma,x:\vart}{M}$ and we have $\ade{x}{M}\le 2$.
Then we observe that in the translation of Theorem~\ref{thm:trans} ($\seq{\Gamma,x:\vart}{M:\theta}$) the questions related to $x$,
(namely $\mwrite{i}^{(x,\rho)}$ and $\mread^{(x,\rho)}$) correspond to creating leaves at levels $1$ or $3$, while the corresponding answers ($\mok^{(x,\rho)}$ and $i^{(x,\rho)}$ respectively)
correspond to removing such leaves.
In the construction for $\seq{\Gamma}{\newin{x}{M}}$,
such transitions need access to the root (to read/update the current state) and the root is indeed
within the allowable range: in an $\sla$
transitions creating/destroying leaves at level $3$ can read/write at level $0$.
All other transitions (not labelled by $x$) proceed as in $M$ and
need not consult the root for additional information about the current state, as it is propagated. Consequently, if $M$ is represented by a $\sla$ then the interpretation of $\newin{x\aasg i}{M}$ is also a $\sla$. The construction does not affect the branching bound, because
the resultant runs can be viewed as a subset of runs of the automaton for $M$, i.e. those in which reads and writes are related.

For $f M_h\cdots M_1$, we observe that the construction first 
creates two nodes at levels $0$ and $1$, and the node at level $1$ is used
to run an unbounded number of copies of (the automaton for) $M_i$.
The copies do not need access to the states stored at levels $0$ and $1$, 
because they are never modified when the copies are running.
Consequently, if each $M_i$ can be translated into a $\sla$,
the outcome of the construction in Theorem~\ref{thm:trans} is also a $\sla$. The new branching bound is the maximum over bounds from $M_1,\cdots, M_h$, because at even levels children are produced as in $M_i$ and 
level $0$ produces only $1$ child.
\cutout{The other constructions involve running an automaton for a subterm (to completion), to be followed by an automaton
for another subterm. This does not violate the property of being a $\sla$, as the switch happens through a level-$1$ transition,
which can read/write the root. Boundedness is also preserved:
the new bound is the maximum of bounds obtained from the inductive hypothesis.}
\qed
\end{proof}
\begin{corollary}
For any $\sfica$-term $\seq{\Gamma}{M:\theta}$, the problem of
determining whether $\comp{\sem{\seq{\Gamma}{M}}}$ is empty
is decidable.
\end{corollary}
Theorems~\ref{thm:full} and \ref{thm:sla-decidable} imply the above. Thanks to Theorem~\ref{thm:full}, it is decidable if a $\sfica$ term is equivalent to a term  that always diverges (cf.\ example on page~\pageref{ex:verification}).
In case of inequivalence, our results could also be
applied to extract the distinguishing context, first by
extracting the witnessing trace from the argument 
underpinning Theorem~\ref{thm:sla-decidable} and then feeding it to
the Definability Theorem (Theorem 41~\cite{GM08}). This is 
a valuable property given that in the concurrent setting bugs
are difficult to replicate.

\section{From LA to FICA}
\label{sec:tofica}

In this section, we show how to represent leafy automata in $\fica$. Let $\Aut=\abra{\Sigma,k,Q,\delta}$ be a leafy automaton. 
We shall assume that $\Sigma,Q\subseteq\{0,\cdots,\imax\}$ so that we can
encode the alphabet and states using type $\expt$.
We will represent a trace $w$ generated by $\Aut$ by a play $\play{w}$, which simulates
each transition with two moves, by $O$ and $P$ respectively. The child-parent links in $\D$ will be represented by justification pointers. We refer the reader to Appendix~\ref{apx:tofica} for details. Below we just state the lemma that 
identifies the types that correspond to our encoding, where 
we write $\theta^{\imax+1}\rarr\beta$ for $\underbrace{\theta \rarr\cdots\rarr\theta}_{\imax+1}\rarr\beta$.

\begin{lemma}
Let $\Aut$ be a $k$-$\la$ and $w\in\trace{\Aut}$. Then $\play{w}$ is a play in $\sem{\theta_k}$,
where $\theta_0=\comt^{\imax+1}\rarr\expt$ and $\theta_{i+1}=(\theta_i\rarr\comt)^{\imax+1}\rarr\expt$ ($i\ge 0$).
\end{lemma}
Before we state the main result, we recall from~\cite{GM08} that strategies corresponding to $\fica$ terms 
satisfy a closure condition known as~\emph{saturation}: swapping two adjacent moves in a play belonging 
to such a strategy yields another play from the same strategy, 
as long as the swap yields a play and it is not the case
that the first move is by O and the second one by P.
Thus, saturated strategies express causal dependencies of P-moves on O-moves.
Consequently, one cannot expect to find a $\fica$-term such that the corresponding
strategy is the smallest strategy containing $\{\,\play{w}\,|\,w\in \trace{\Aut}\,\}$. 
Instead, the best one can aim for is the following result.
\begin{theorem}\label{thm:toalgol}
Given a $k$-$\la$ $\Aut$, there exists a $\fica$ term $\seq{}{M_\Aut:\theta_k}$ such that $\sem{\seq{}{M_\Aut:\theta_k}}$
is the smallest saturated strategy containing $\{\,\play{w}\,|\,w\in \trace{\Aut}\,\}$.
\end{theorem}
\begin{proof}[Sketch]
Our assumption $Q\subseteq\{0,\cdots,\imax\}$ allows us to maintain $\Aut$-states in the memory of $\fica$-terms.
To achieve $k$-fold nesting, we rely on the  higher-order structure of the term:
$\lambda f^{(0)}.f^{(0)}(\lambda f^{(1)}.f^{(1)}(\lambda f^{(2)}.f^{(2)}(\cdots \lambda f^{(k)}. f^{(k)})))$.
In fact, instead of the single variables $f^{(i)}$, we shall use sequences
$f^{(i)}_0\cdots f^{(i)}_\imax$, so that a question $t_{\Q}^{(i)}$ read by $\Aut$ at level $i$ can be simulated
by using variable $f^{(i)}_{t_{\Q}^{(i)}}$ (using our assumption $\Sigma\subseteq\{0,\cdots,\imax\}$).
Additionally, the term contains state-manipulating code that enables moves only if they are
consistent with the transition function of $\Aut$.\qed
\end{proof}
\section{Conclusion and further work}

%\igwnote{My take on conclusions}
We have introduced leafy automata, $\la$, and shown that they correspond to the game
semantics of Finitary Idealized Concurrent Algol ($\fica$).
The automata formulation makes combinatorial challenges posed by the equivalence problem explicit. 
This is exemplified by a very transparent undecidability proof of the emptiness problem for $\la$. 
Our hope is that $\la$ will allow to discover interesting fragments of $\fica$ for which some variant of the equivalence problem is decidable. 
We have identified one such instance, namely local leafy automata ($\lla$), and a fragment of $\fica$ that can be translated to them. 
The decidability of the emptiness problem for $\lla$ implies decidability of a simple instance of the equivalence problem.
This in turn allows to decide some verification questions as in the example on page~\pageref{ex:verification}. 
Since these types of questions involve quantification over all contexts, the use of a fully-abstract semantics appears essential to solve them. 

The obvious line of future work is to find some other subclasses of $\la$ with decidable emptiness problem. 
Another interesting target is to find an automaton model for the call-by-value setting, where answers enable questions~\cite{AM97b,HY97}. 
It would also be worth comparing our results with abstract machines~\cite{FG13}, the
Geometry of Interaction~\cite{LTY17}, and the $\pi$-calculus~\cite{BHY01}.

% \igwnote{old verion}We have introduced leafy automata and shown that they correspond to the game
% semantics of Finitary Idealized Concurrent Algol ($\fica$).
% This makes them an automata-theoretic foundation for studying and verifying higher-order concurrent computation.
% We have made a first step in this direction by indentifying a decidable subclass $\lla$
% and a fragment of $\fica$ that can be translated into $\lla$.
% The results are immediately applicable to the problem of deciding a simple instance of the equivalence problem,
% which can be applied to specify safety-related verification tasks.

% An interesting target for future work would be to check if we can find an
% automaton model for a call-by-value setting, where answers enable questions~\cite{AM97b,HY97}.
% It would also be worth comparing our results with abstract machines~\cite{FG13}, the
% Geometry of Interaction~\cite{LTY17} and the $\pi$-calculus~\cite{BHY01}.

% \input{sections/intro}
% \input{sections/fica}
% \input{sections/games}
% \input{sections/leafy}
% \input{sections/from-algol}
% \input{sections/sfica}
% \input{sections/to-algol}
% \input{sections/end}

%%%%%%%%%%%%%%%%%%%%%%%%%%%%%%%%%%%%%%%%%%%%%%%%%%%%%%%%%%%%%%%%%%%%%%%%%%%%%%%%
% BIBLIOGRAPHY

\bibliographystyle{lncs-resources/splncs04}
\bibliography{bibliography}

%%%%%%%%%%%%%%%%%%%%%%%%%%%%%%%%%%%%%%%%%%%%%%%%%%%%%%%%%%%%%%%%%%%%%%%%%%%%%%%%
% OBLIGATORY OPEN ACCESS PORTION

%%%%% To display Open Access text and logo, Please add below text and copy the cc_by_4-0.eps in the manuscript package %%%

% \vfill

% {\small\medskip\noindent{\bf Open Access} This chapter is licensed under the terms of the Creative Commons\break Attribution 4.0 International License (\url{http://creativecommons.org/licenses/by/4.0/}), which permits use, sharing, adaptation, distribution and reproduction in any medium or format, as long as you give appropriate credit to the original author(s) and the source, provide a link to the Creative Commons license and indicate if changes were made.}

% {\small \spaceskip .28em plus .1em minus .1em The images or other third party material in this chapter are included in the chapter's Creative Commons license, unless indicated otherwise in a credit line to the material.~If material is not included in the chapter's Creative Commons license and your intended\break use is not permitted by statutory regulation or exceeds the permitted use, you will need to obtain permission directly from the copyright holder.}

% \medskip\noindent\includegraphics{llncs2e-resources/cc_by_4-0.eps}

%%%%%%%%%%%%%%%%%%%%%%%%%%%%%%%%%%%%%%%%%%%%%%%%%%%%%%%%%%%%%%%%%%%%%%%%%%%%%%%%
% APPENDIX

\appendix
\section{Additional material for Section~\ref{sec:fica}}
\label{apx:opsem}

\subsection{Operational semantics of $\fica$}

{
The operational semantics is defined using a (small-step) transition
relation $\step{s}{M}{s'}{M'}$, where $\mem$ is a set of variable names
denoting active \emph{memory cells} and \emph{semaphore locks}.
$s,s'$ are states, i.e.\ functions $s,s':\mem\rightarrow\makeset{0,\cdots,\imax}$, and $M,M'$ are
terms.  We write $s\otimes (v\mapsto i)$ for the state obtained by augmenting $s$ with $(v\mapsto i)$, assuming $v\not\in \dom{s}$.
The basic reduction rules are given in Figure~\ref{fig:os},
where $c$ stands for any language constant ($i$ or $\skipcom$)
and $\widehat{\mathbf{op}}:\{0,\cdots,\imax\}\rarr\{0,\cdots,\imax\}$
is the function corresponding to $\mathbf{op}$.
In-context reduction is given by the schemata:

\begin{center}
\AxiomC{$\mem,v\vdash  M[v/x],s\otimes(v\mapsto i)\longrightarrow M',s'\otimes(v\mapsto i') $ \quad $M\neq c$}
\UnaryInfC{$\mem\vdash\newin{x\aasg i}{M},s\longrightarrow \newin{x\aasg i'}{M'[x/v]}, s' $}
\DisplayProof\\[2ex]
\AxiomC{$\mem,v\vdash  M[v/x],s\otimes(v\mapsto i)\longrightarrow M',s'\otimes(v\mapsto i') $\quad $M\neq c$}
\UnaryInfC{$\mem\vdash\newsem{x\aasg i}{M},s\longrightarrow \newsem{x\aasg i'}{M'[x/v]}, s' $}
\DisplayProof\\[2ex]
  \AxiomC{$\step{s}{M}{s'}{M'}$}
  \UnaryInfC{$\step{s}{\mathcal E[M]}{s'}{\mathcal E[M']}$}
  \DisplayProof
\end{center}
where reduction contexts $\mathcal E[-]$ are produced by the
grammar:
\[\begin{array}{rcl}
  \mathcal E[-] &::=& [-] \mid \mathcal E;N
  \mid (\mathcal E\,\parc\, N)
  \mid (M\,\parc\, \mathcal E)
  \mid {\mathcal E} N 
  \mid \arop{\mathcal E} 
  \mid \cond{\mathcal E}{N_1}{N_2}\\
  &&\mid {!}\mathcal E
  \mid \mathcal E\aasg m
  \mid M\aasg\mathcal E
  \mid \grb{\mathcal E}
  \mid \rls{\mathcal E}.
\end{array}\]
\begin{figure}[t]
\begin{center}
$\begin{array}{rclcrcl}
  \astep {s}{\skipcom\parc\skipcom}{s}{\skipcom} &\quad &   \astep {s}{\cond{i}{N_1}{N_2}}{s}{N_1},\quad i\neq 0\\
  \astep {s}{\skipcom;c}{s}{c} &&   \astep {s}{\cond{0}{N_1}{N_2}}{s}{N_2}\\
  \astep {s}{\arop{i}}{s}{\widehat{\mathbf{op}}(i)} &&    \astep{s}{(\lambda x.M) N}{s}{M[N/x]}\\
  \astep {s}{\newin{x\aasg i}c}{s}{c} &&    \astep {s\otimes(v\mapsto i)}{{!}v}{s\otimes(v\mapsto i)}{i}\\
  \astep {s}{\newsem{x\aasg i}c}{s}{c}  &&\astep {s\otimes(v\mapsto i)}{v\aasg i'}{s\otimes(v\mapsto i')}{\skipcom}
  \end{array}$
  \end{center}
  \begin{center}
  $\begin{array}{rcl}
  \astep {s\otimes(v\mapsto 0)}{\grb v}{s\otimes(v\mapsto 1)}{\skipcom}\\
  \astep {s\otimes(v\mapsto i)}{\rls v}{s\otimes(v\mapsto 0)}{\skipcom},\quad i\neq 0\\
   \astep {s}{\while{M}{N}}{s}{\cond{M}{(N;\while{M}{N})}{\skipcom}}
\end{array}$
\end{center}
\caption{Reduction rules for $\fica$}
\label{fig:os}
\end{figure}
$\seq{}{M:\comt}$ is said to terminate,  written $M\Downarrow$, if 
$\emptyset \vdash \emptyset,\,M\longrightarrow^\ast \emptyset,\skipcom$.

\bigskip

Idealized Concurrent Algol~\cite{GM08} also features variable and semaphore constructors, 
called $\textbf{mkvar}$ and $\textbf{mksem}$ respectively,
which play a technical role in the full abstraction argument, similarly to~\cite{AM97a}.
We omit them in the main body of the paper, because
they do not present technical challenges, but they are covered in the Appendix for the sake of completeness.

\paragraph{Typing rules}
\[
\AxiomC{$\Gamma\vdash M:\expt\rarr\comt$}
  \AxiomC{$\Gamma\vdash N:\expt$}
  \BinaryInfC{$\Gamma\vdash \mkvar{M}{N}:\vart$}
  \DisplayProof
\quad
 \AxiomC{$\Gamma\vdash M:\comt$}
  \AxiomC{$\Gamma\vdash N:\comt$}
  \BinaryInfC{$\Gamma\vdash  \mksem{M}{N}:\semt$}
  \DisplayProof
\]

\paragraph{Reduction rules}

\begin{align*}
  \step {s&}{(\mkvar{M}{N})\aasg M'}{s}{M M'}\\
  \step {s&}{{!}(\mkvar{M}{N}}{s}{N}\\
  \step {s&}{\grb {\mathbf{mksem}\,M N}}{s}{M}\\
  \step {s&}{\rls {\mathbf{mksem}\,M N}}{s}{N}
\end{align*}

\paragraph{$\eta$ rules for $\vart,\semt$}
\[\begin{array}{rcl}
M &\longrightarrow & \mkvar{(\lambda x^\expt. M\aasg x)}{!M}\\
M &\longrightarrow & \mksem{\grb{M}}{\rls{M}}
\end{array}\]

Using $\mathbf{mkvar}$ and $\mathbf{mksem}$,
one can define $\divcom_\theta$ as syntactic sugar using $\divcom=\divcom_\comt$ only.
\[
\divcom_\theta=\left\{
\begin{array}{lcl}
\divcom && \theta=\comt\\
\divcom;0 && \theta=\expt\\
\mkvar{\lambda x^\expt.\divcom}{\divcom_\expt} && \theta=\vart\\
\mksem{\divcom}{\divcom} & &\theta=\semt\\
\lambda x^{\theta_1}.\divcom_{\theta_2} & & \theta=\theta_1\rarr\theta_2\\
\end{array}\right.
\]

\section{Additional material for Section~\ref{sec:leafy}}
\label{apx:leafy}

\subsection{Proof of Lemma~\ref{lem:la1}}

We proceed by reducing from the halting problem for deterministic two-counter machines~\cite[pp.~255--258]{Min67}.

The input to the halting problem is a deterministic two-counter machine \\* $\mathcal{C} = (Q_\mathcal{C}, q_0, q_F, T)$, where $Q_\mathcal{C}$ is the set of states, $q_0, q_F \in Q_\mathcal{C}$ are the initial and final states respectively, and $T : Q_\mathcal{C} ~\setminus~\{q_F\} \rightarrow (\move{INC} \cup \move{JZDEC})$ is the step function. Steps in $\move{INC}$ are of the form $(i, q') \in \{1,2\} \times Q_\mathcal{C}$ (increment counter $i$ and go to state $q'$). Steps in $\move{JZDEC}$ are of the form $(i, q', q'') \in \{1,2\} \times Q_\mathcal{C} \times Q_\mathcal{C}$ (if counter $i$ is zero then go to state $q'$; else decrement counter $i$ and go to state $q''$). The question is whether, starting from $q_0$ with both counters zero, $\mathcal{C}$ eventually reaches $q_F$ with both counters zero.

\bigskip

We first construct a $1$-$\la$ that recognises the language of all data words such that:
\begin{itemize}
\item 
the underlying word (i.e., the projection onto the finite alphabet) encodes a path through the transition relation of $\mathcal{C}$ from the initial state to the final state, in other words a pseudo-run where the non-negativity of counters and the correctness of zero tests are ignored;
\item
the occurrences of the letters that encode increments and decrements of $\mathcal{C}$ form pairs that are labelled by the same level-$1$ data values, where each increment is earlier than the corresponding decrement, which assuming that both counters are zero initially ensures their non-negativity throughout the pseudo-run and their being zero finally.
\end{itemize}

The second $1$-$\la$ is slightly more complex.  It accepts data words that have the same properties as those accepted by the first $1$-$\la$, and in addition:
\begin{itemize}
\item there exists some increment followed by a zero test of the same counter before a decrement with the same data value has occurred, in other words there is at least one incorrect zero test in the pseudo-run.
\end{itemize}

The two sets of accepted traces will be equal if and only if all pseudo-runs that satisfy the initial, non-negativity and final conditions necessarily contain some incorrect zero test, i.e.\ if and only if $\mathcal{C}$ does not halt as required. We give the formal construction below. 

\bigskip

\newcommand\Aone{\mathcal{A}_1(\mathcal{C})}
\newcommand\Atwo{\mathcal{A}_2(\mathcal{C})}

The two LAs we compute are $\Aone = \langle \Sigma, 1, Q, \delta_1 \rangle$ and $\Atwo = \langle \Sigma, 1, Q, \delta_2 \rangle$.
% such that $\mathcal{C}$ does not halt as specified above if and only if $\Aone$ and $\Atwo$ have the same language.

The alphabet, $\Sigma = \Sigma_\Q \cup \Sigma_\A$, is defined as follows:
\[
\Sigma_\Q = \{ \move{start}, \move{inc_1}, \move{inc_2}, \move{zero_1}, \move{zero_2} \}
\qquad
\Sigma_\A = \{ \move{end}, \move{dec_1}, \move{dec_2}, \move{zero'_1}, \move{zero'_2} \}
\]

Traces of $\Aone$ and $\Atwo$ represent pseudo-runs of $\mathcal{C}$, i.e.~sequences of steps of the machine. Aside from $\move{start}$ and $\move{end}$, each letter in the trace corresponds to the machine performing either an $\move{INC}$ step ($\move{inc}$), the ``then'' of a $\move{JZDEC}$ step ($\move{zero}$), or the ``else'' of a $\move{JZDEC}$ step ($\move{dec}$). The $\move{zero'}$ transition is a necessity which allows us to erase leaves added by $\move{zero}$. Each of $\move{inc}$, $\move{dec}$, $\move{zero}$, $\move{zero'}$ has two variants which encode $i$, the counter number in the corresponding step. We will say that two letters \emph{match} if they have the same data value.

By construction $\Aone$ will accept exactly the traces with the following properties, which correspond to the high-level description of our first $1$-$\la$:

\begin{itemize}
    \item The first letter in the trace is $\move{start}$ and the last is a matching $\move{end}$.
    \item For each occurrence of $\move{inc_i}$, there is a matching $\move{dec_i}$ later in the trace.
    \item For each occurrence of $\move{zero_i}$, there is a matching $\move{zero'_i}$ later in the trace.
    \item The letters in the trace (excluding $\move{start}$ and $\move{end}$) form a sequence $(a_0,\ldots,a_{n-1})$; there exists some sequence of states $(s_0,\ldots,s_n) \in Q_\mathcal{C}^{n+1}$ such that for all $i \in (0,\ldots,n-1)$, $s_{i+1}$ appears as the second or third component of~$T(s_{i})$, and $a_i$ is a step which may be performed at state $s_i$ (irrespective of counter values).
\end{itemize}

The state space of the root, $Q^{(0)} = Q_{\mathcal{C}} \times \{\circ, \star, \mathbf{1}, \mathbf{2} \}$, comprises pairs where the first component corresponds to a state of $\mathcal{C}$ and the second tracks an observation of some invalid sequence. The second component is only used in $\Atwo$. We denote the pair at the root by square brackets. The states of the leaves at level 1 are $Q^{(1)} = \bigcup \big\{~ \{ i, 0_{i}, i\star \} ~\big|~ i \in \{1, 2\} ~\big\}$, where $0_{i}$ denotes a temporary leaf generated by $\move{zero_i}$, $i$ denotes a counter, and $i\star$ denotes a counter being observed in $\Atwo$.

The transition function $\delta_1$ of $\Aone$ is defined as follows.

\[
\dagger \trans{\move{start}}_1 [q_0, \circ] 
\qquad 
[q_F, \circ] \trans{\move{end}}_1 \dagger
\qquad
\frac{q \trans{\move{INC}} (i,q')~\in T}{
[q, \circ] \trans{\move{inc_i}}_1 ([q', \circ], i)
}
\]
\[ 
\frac{q \trans{\move{JZDEC}} (i, q', q'')~\in T}{
([q, \circ], i) \trans{\move{dec_i}}_1 [q'', \circ]
\qquad
[q, \circ] \trans{\move{zero_i}}_1 ([q', \circ], 0_{i})
}
\qquad
\frac{q \in Q_\mathcal{C}}
{
([q, \circ], 0_{i}) \trans{\move{zero'_i}}_1 [q, \circ]
}
\]
\bigskip

By construction $\Atwo$ accepts exactly those traces of $\Aone$ where at least one $\move{zero_i}$ letter occurs in between an $\move{inc_i}$ letter and the matching letter $\move{dec_i}$. In other words, the ``then'' of a $\move{JZDEC}$ step has been taken while the counter was nonzero. This is not a legal step, and so such a trace does not represent a computation of $\mathcal{C}$. This implements the high-level description of our second $1$-$\la$.

In order to accept a word, $\Atwo$ must change the second component of the root's state from $\star$ to $\circ$. It does this by nondeterministically choosing to observe some $\move{inc}$ transition. From here, it proceeds as in $\Aone$ until either it meets the matching $\move{dec}$, in which case the automaton rejects, or it meets an $\move{ifz}$ transition on the same counter, at which point it marks the second component with $\circ$ and proceeds as in $\Aone$. 

The transition function $\delta_2$ of $\Atwo$ is defined as follows:
\[
\dagger \trans{\move{start}}_2 [q_0, \star] 
\qquad
[q_F, \circ] \trans{\move{end}}_2 \dagger
\qquad
\frac{
q \trans{\move{INC}} (i,q')~\in T
\qquad
x \in \{\circ, \star, \mathbf{1}, \mathbf{2}\}
}{
[q, x] \trans{\move{inc_i}}_2 ([q', x], i)
\qquad
[q, \star] \trans{\move{inc_i}}_2 ([q, \mathbf{i}], i\star)
}
\]

\[
\frac{
q \trans{\move{JZDEC}} (i,q',q'')~\in T
\qquad
x \in \{\circ, \star, \mathbf{1}, \mathbf{2}\}
}{
[q, x] \trans{\move{zero_i}}_2 ([q', x], 0_{i})
\qquad
[q, \mathbf{i}] \trans{\move{zero_i}}_2 ([q', \circ], 0_{i})
}
\qquad
\frac{
q \in Q_\mathcal{C}
\qquad
x \in \{\circ, \star, \mathbf{1}, \mathbf{2}\}
}{
([q, x], 0_{i}) \trans{\move{zero'_i}}_2 [q, x]
}
\]

\[
\frac{
q \trans{\move{JZDEC}} (i,q',q'')~\in T
\qquad
x \in \{\circ, \star, \mathbf{1}, \mathbf{2}\}
}{
([q, x], i) \trans{\move{dec_i}}_2 [q'', x]
\qquad
([q, \circ], i\star) \trans{\move{dec_i}}_2 [q'', \circ]
}
\]
\bigskip

% For each $\move{JZDEC}$ step $(q, i, q', q'')$ in $T$, the following transitions are in $\delta_2^\move{Steps}$:
% \[
% [q, \circ, x] \trans{\move{zero_i}} ([q', \mathbf{i}, x], \mathbf{i})
% \qquad
% [q, \circ, \mathbf{i}] \trans{\move{zero_i}} ([q', \mathbf{i}, \circ], \mathbf{i})
% \qquad
% ([q', \mathbf{i}, x], \mathbf{i}) \trans{\move{zero'_i}} [q', \circ, x]
% \]\[
% ([q, \circ, x], i) \trans{\move{dec_i}} [q'', \circ, x]
% \qquad 
% ([q, \circ, \circ], i\star) \trans{\move{dec_i}} [q'', \circ, \circ]
% \]

$\Aone$ captures every correctness condition for halting computations of $\mathcal{C}$ except the legality of $\move{zero}$ steps. Hence, $\Atwo$ accepts exactly those accepted traces of $\Aone$ which are \emph{not} halting computations of $\mathcal{C}$, and so $\mathcal{C}$ performs a halting computation if and only if $\Aone \neq \Atwo$.

{
\newcommand{\Ss}{\mathcal{S}}
\newcommand{\Test}{\mathtt{TEST}}
\newcommand{\quu}{\ensuremath{q^{(2i-2)}}}
\newcommand{\qu}{\ensuremath{q^{(2i-1)}}}
\newcommand{\qi}{\ensuremath{q^{(2i)}}}
\newcommand{\qii}{\ensuremath{q^{(2i+1)}}}
\newcommand{\qiii}{\ensuremath{q^{(2i+2)}}}
\newcommand{\qpuu}{\ensuremath{q^{\prime(2i-2)}}}
\newcommand{\qpu}{\ensuremath{q^{\prime(2i-1)}}}
\newcommand{\qpi}{\ensuremath{q^{\prime(2i)}}}
\newcommand{\qpii}{\ensuremath{q^{\prime(2i+1)}}}
\newcommand{\qpiii}{\ensuremath{q^{\prime(2i+2)}}}

\newcommand{\dq}{\delta_\Q}
\newcommand{\da}{\delta_\A}

\newcommand{\Aad}{\Aa^{\downarrow}}
\newcommand{\Aau}{\Aa^{\uparrow}}
\newcommand{\Bbu}{\Bb^{\uparrow}}
\newcommand{\qinit}{q_{init}}
\newcommand{\maxrank}{\mathit{maxrank}}
\newcommand{\summary}{\mathit{Summary}}
\newcommand{\maxdom}[1]{\max(\dom{#1})}
\newcommand{\inc}{\mathit{inc}}
\newcommand{\dec}{\mathit{dec}}
\newcommand{\add}{\mathit{add}}
\newcommand{\del}{\mathit{del}}
\newcommand{\wf}{\widehat{f}}

\newcommand{\state}{\mathit{state}}
\newcommand{\children}{\mathit{children}}
\newcommand{\wrr}{\widehat{r}}
\newcommand{\If}{{\sffamily if}}
\newcommand{\Then}{{\sffamily then}}

\newcommand{\Aa}{\ensuremath{\mathcal{A}}}
\newcommand{\Nat}{\mathbb{N}}
\newcommand{\Bb}{\ensuremath{\mathcal{B}}}
\newcommand{\set}[1]{\{#1\}}
\newcommand{\s}{s}

\section{Additional material for Section~\ref{sec:lla}}
\label{apx:sla}

\subsection{Proof of Theorem \ref{thm:sla-decidable}}
\label{apx:sla-decidable}

We present a proof of  decidability  of the
emptiness problem for $\lla$, Theorem \ref{thm:sla-decidable}.
There are two main steps in the proof. 
The first step uses a notion of summary for some even layer $2i$. 
This allows to restrict an automaton to first $2i$ layers.
The second step is a method for computing a summary for layer $2i$ from a
summary for layer $2i+2$.

\subsection*{Summaries}
The structure of transitions of $\sla$ provides a notation of a domain for data
values. 
The \emph{domain} of a data value $d \in \D$ is the set of data values whose
associated state may be modified by a transition that adds or removes $d$, i.e.,
when reading a letter annotated by $d$.
\[
\dom{d} = \begin{cases}
\set{\predc^2(d), \predc(d), d} & \text{if d is at an even level} \\
\set{\predc^3(d), \predc^2(d), \predc(d), d} & \text{if d is at an odd level}
\end{cases}
\]

Domains give us a notion of independence: Two letters $(t_1, d_1)$, $(t_2, d_2)$ are \emph{independent} if the domains of $d_1$ and $d_2$ are disjoint. We remark that if $w$ is a trace of some $\sla$ then every sequence obtained by permuting adjacent independent letters of $w$ is also a trace of the same $\sla$ ending in the same configuration.

%The form of moves of LLA gives a notion of a domain for every data value $d$.
%It consists of data values whose states can change when reading a letter with
%data $d$. 
% We have $\dom{d}=\{\predc^2(d),\predc^1(d),d\}$ if $d$ is on an even layer, and 
% $\dom{d}=\{\predc^3(d),\predc^2(d),\predc^1(d),d\}$ if $d$ is on an odd layer.
% The domains in turn induce independence on letters with data: $(t,d)$ is
% independent from $(t',d')$ if the domains of $d$ and $d'$ are disjoint.
% It is easy to see that if $w$ is a trace of an LLA then 
% every sequence obtained from $w$ by permuting adjacent independent letters is
% also a trace ending in the same configuration. 

Let us fix an $k$-$\sla$ automaton $\Aa = \abra{\Sigma_\Aa, k_\Aa, Q_\Aa,
\delta_\Aa}$, and let $b$ be its even-layer bound.

Suppose, on an accepting trace on $\Aa$, we encounter some data value $d$ at
even layer $2i$. 
On an accepting trace value $d$ occurs twice: the first occurrence corresponds to 
adding $d$, the second to deleting $d$. Let
$w$ be the part of the trace in between, and including, these two occurrences of $d$.

We can classify letters $(t',d')$ in $w$ into one of three categories:

\begin{enumerate}
    \item \emph{$d$-internal}, when $\dom{d'}$ is included in the subtree rooted at $d$;
    \item \emph{$d$-external}, when $\dom{d'}$ is disjoint from the subtree rooted at $d$;
    \item \emph{$d$-frontier}, when $\dom{d'}$ contains  $d$ and its parent.
\end{enumerate}

Note that these three categories partition the set of all letters in $w$.
The frontier letters are the ones with data value $d$, as well as those with
children of $d$. 
The later are from layer $2i+1$.
Letters with data values from bigger layers are either $d$-internal or
$d$-external.

At this point we use branching bound $b$ of the automaton. 
The number of children of $d$ is bounded by $b$, and every child of $d$ appears
twice in $w$. 
Hence, the number of $d$-frontier letters in $w$ is at most $b+2$, counting the letters with~$d$. 

% The notion of summary for a node at level $2i$ comes from examining the form of
% an accepting trace of $\Aa$ modulo independence.
% Suppose we have a trace of an automaton, and that this trace adds at some
% moment a node $d$ of layer $2i$. 
% Since the trace is accepting, the node must be deleted some time later. 
% Let $w$ be the part of the trace between these two events.
% We classify letters $(t',d')$ from $w$ into one of three types: (i) \emph{$d$-internal} 
% when $\dom{d'}$ is included in the subtree rooted in $d$; (ii) \emph{$d$-external}
% when $\dom{d'}$ in disjoint from the subtree rooted in $d$; (iii) 
% \emph{$d$-frontier} when $d'$ is a child of $d$.
% The number of frontier letters in $w$ is bounded due to the form of LLA.

The $d$-frontier letters divide $w$ into subwords, giving us a sequence of
transitions:
\begin{equation}
\label{eqn:expanded}
	\kappa_1\trans{m_1}\kappa'_1\trans{w_1}\kappa_{2}\trans{m_{2}}\kappa'_{2}\trans{w_2}\dots\kappa_{l}\trans{m_l}\kappa'_l\trans{w_l}\kappa_{l+1}\trans{m_{l+1}}\kappa'_{l+1}
\end{equation}
where $m_1,\dots,m_l$ are $d$-frontier letters; $m_1$ adds node $d$ while $m_{l+1}$ deletes  $d$.
% \adnote{By the definition above, $w$ should not contain the transition which creates or deletes $d$ (though it is adjacent).}\igwnote{You are right. This needs adjustement one way or another but I am still not sure where it is the best to cut.}

Configuration $\kappa'_1$ is the first in which $d$ appears in the tree, so $d$
is a leaf node in $\kappa'_1$. Likewise, $\kappa_l$ is the last configuration in
which $d$ appears, as it is removed by $m_{l+1}$, so $d$ is a leaf node in
$\kappa_{l+1}$. 

We now use independence properties.
Every word $w_j$ contains only $d$-internal and $d$-external letters. 
Due to independence, $w_j$ is equivalent to some $u_jv_j$, with $u_j$
containing only  $d$-internal letters of $w_j$ and $v_j$ containing only  the
$d$-external letters of $w_j$.
(Actually $u_1$ and $u_l$ are empty but we do not need to make a case
distinction in the rest of the argument) 
% So $w_i$ are composed only from $d$-internal and $d$-external letters. 
% Configuration $\kappa_0$ is the one where node $d$ is created; so $d$ is a leaf
% labelled  with some state, say $q$.
% Configuration $\kappa'_{l+1}$ is the one where $d$ is removed; so $d$ is a leaf in
% $\kappa_{l+1}$, and is not in the domain of $\kappa'_{l+1}$.
% This implies that in $w_1$ and $w_{l+1}$ there are no internal actions as the
% first node of level $2i+1$ below $d$ is created in $\kappa'_1$, and the last
% node of level $2i+1$ below $d$ is removed in $\kappa'_l$.
% The important point is that every $w_i$ is equivalent
% modulo independence to $u_iv_i$ with $u_i$ consisting only of $d$-internal letters,
% and $v_i$ only of $d$-external letters.

From here, we can see that the $d$-internal parts
$u_1,\cdots,u_l$ of $w$ only interact with the $d$-external parts at a bounded
number of positions, and those positions exactly correspond to the frontier
transitions $m_2,\cdots,m_l$. Hence, if we could characterize the interactions
that can occur at level $2i$, then we could replace the sequences of transitions
on every $u_j$ by a single short-cut transition. 
This would eliminate the need for levels $\geq 2i$ in the automaton.

We introduce a notion of a summary to implement the idea of short-cut
transitions.
A \emph{summary} for level $2i$ is a function $f \colon \set{1,\dots, 2(l+1)} \to Q^{2i-2} \times
Q^{2i-1}$; for some $l\leq b+1$.
Intuitively, from some trace $w$ expanded as in Equation
\ref{eqn:expanded}, we can extract $f$ such that $f(2j-1)$ is a pair of states
labelling $\predc^2(d)$ and $\predc(d)$ in $\kappa_j$, while $f(2j)$ is a pair
of states labelling these nodes in $\kappa'_{j}$.
This is only the intuition because we do not have runs of $\Aa$ at hand to
compute $f$. 

% From this discussion it follows that the interactions of the internal part of
% the trace $u_1,\dots, u_l$ with the rest of the trace are limited to a bounded
% number of moves: $m_1\dots,m_{l+1}$. 
% Thus if we knew what interactions at level $i$ the automaton can
% realize, we could replace computations $u_1,\dots, u_l$  by single transitions.
% This simplifies the automaton as it removes access to all data layers $\geq 2i$ in the tree.

% We present the interaction with the outside in a form of a \emph{summary} that is a function
% $f:[1,\dots,2d]\to Q^{2i-2}\times Q^{2i-1}$.
% Intuitively, from a trace as above, we extract $f$ where $f(2j-1)$ is a pair of states 
% labelling $\predc^2(d)$ and $\predc(d)$ in $\kappa_j$, while $f(2j)$ is a pair
% of states labelling these nodes in $\kappa'_{j}$.
% The formal definition of the summary actually does not refer to runs of the
% whole automaton.

To formalise the idea of summaries for a given automaton, we will introduce the notion of a \emph{cut
automaton}. Intuitively, the behaviour of a cut automaton $\Aad(2i,  f)$ will
represent the behaviours of $\Aa$ contained within some subtree rooted in a
data value at layer $2i$. 

The states and transitions of $\Aad(2i,  f)$ are those of $\Aa$ but lifted
up so that level $2i$ becomes the root level:
\begin{equation*}
	\Q^{\downarrow (l-2i)} = \Q^{(l)} \qquad \dq^{\downarrow (l-2i)}=\dq^{(l)}\qquad \da^{\downarrow (l-2i)}=\da^{(l)}\qquad\text{for $l\geq2i+2$}
\end{equation*}
The two to layers, $0$ and $1$, are special as just lifting transitions would
make them stick above the root. 
Here is also the place where we use the summary $f$.
\begin{equation*}
    Q^{\downarrow(0)}=Q^{(2i)}\times\dom{f}\qquad Q^{\downarrow(1)}=Q^{(2i+1)}
\end{equation*}
The extra component at layer $0$ will be used for layer $1$ transitions.

% To formalize a notion of a summary for a level $2i$ of an automaton  $\Aa$ we introduce the notion of a \emph{cut automaton}, denoted $\Aad(2i,\qinit,f)$. Its behavior is the behavior of $\Aa$ restricted to some subtree rooted in a data value at layer $2i$. 
% The allowed interactions of this automaton with the rest of the tree will be provided as a summary $f$. 
% The initial state at the root of this automaton will be $\qinit$.
% We will define the set of summaries as those $f$ for which 
%  $\Aad(2i,\qinit,f)$ has an accepting run.

Before defining transitions we introduce some notation. 
For a summary $f$ we write $\maxdom{f}$ for the maximal element in the domain of
$f$.
We use an abbreviated notation for transitions. 
If $f(j)=(\quu,\qu)$, and $f(j+1)=(\qpuu,\qpu)$ then we write
\begin{equation*}
    f(j)\trans{a}(f(j+1),\qpi)\ \text{instead of}\ (\quu,\qu)\trans{a}(\qpuu,\qpu,\qpi)\ .
\end{equation*}

Transitions at levels $0$ and $1$ are adaptations of those of levels $2i$ and $2i+1$
in the original automaton.
A node that was at level $2i$ is now the root so it has no predecessors anymore.
The initial and final moves of $\Aad(2i, f)$ create and destroy the root. 
They use $f$ to predict what are states of predecessors in a corresponding move
of $\Aa$.
\begin{align*}
    \dq^{\downarrow (0)} \text{ contains }& \ \dagger\trans{a}(\qpi,1)\\
    & \text{\quad if there is a transition $f(1)\trans{a} (f(2),\qpi)$ in $\dq^{(2i)}$}\\
    \da^{\downarrow (0)} \text{ contains }& \ (q,r)\trans{a}\dagger\\
    &\text{\quad if $r=\maxdom{f}-1$ and there is $(f(r),q)\trans{a} f(r+1)$ in $\da^{(2i)}$}
\end{align*}
Finally, we have transitions that add and delete nodes on level $1$:
\begin{align*}
	\text{in } \dq^{\downarrow (1)} \text{ we have }& 
	(\qi,r)  \trans{a} ((\qpi,r+2),\qpii)\\
	&\text{\qquad if } (f(r),\qi) \trans{a}(f(r+1),\qpi,\qpii) \in \dq^{(2i+1)}\\
	\text{in } \da^{\downarrow (1)} \text{ we have }&
	((\qi,r),\qii)  \trans{a} ((\qpi,r+2))\\
	& \text{\qquad if } (f(r),\qi,\qii)  \trans{a} (f(r+1),\qpi) \in \da^{(2i+1)}
\end{align*}

% The moves of $\Aad(2i,\qinit,f)$ are the same as of $\Aa$ but lifted by $2i$ levels
% \begin{equation*}
% 	\addd_l=\add_{l-2i}\quad \deld_l=\del_{l-2i}\qquad\text{for $l>2i+2$}
% \end{equation*}
% It remains to define the moves for levels $0$ and $1$, as they "stick-out' of the tree. 
% This is where $f$ and $\qinit$ are used.
% We first have initial and final moves of $\Aad(2i,\qinit,f)$ that create a root,
% and delete it. 
% \begin{align*}
% 	\addd_0 &\text{ contains $\dagger\trans{a}(\qinit,1)$}\\
% 	\deld_0 &\text{ contains $(q,r)\trans{a}\dagger$}\\
% 	&\text{for $r=\maxdom(f)-1$ and some $(f(r),q)\trans{a} f(r+1)$ in $\del_{2i}$}
% \end{align*}
% Then we have transitions adding and deleting nodes on level $1$
% \begin{multline*}
% 	\text{in $\addd_1$ we have $(\qi,r)\trans{a} ((\qpi,r+2),\qpii)$}\\
% 	\text{if $(f(r),\qi)\trans{a}(f(r+1),\qpi,\qpii)$ is in $\add_{2i+1}$.}
% \end{multline*}
% \begin{multline*}
% 	\text{in $\deld_1$ we have $((\qi,r),\qii)\trans{a} ((\qpi,r+2),\qpii)$}\\
% 	\text{if $(f(r),\qi,\qii)\trans{a}(f(r+1),\qpi)$ is in $\del_{2i+1}$.}
% \end{multline*}

\renewcommand{\s}{\sigma}

We can now formally define the set of summaries for an even layer $2i$:
\begin{equation*}
	\summary(\Aa,2i)=\set{f \colon \Aad(2i,f) \text{ accepts some trace}}
\end{equation*}

The next step is to define an automaton that uses such a set of summaries.
The idea is that when a node of layer $2i$ is created it is assigned a summary
from the set of summaries. 
Then all moves below this node are simulated by consulting this summary.
So we will never need layers below $2i$.

Let $\Ss$ be a set of summaries at level $2i$.
We will now define $\Aau(2i, \Ss)$.
It will be $(2i+1)$-$\lla$ automaton.
The states and transitions of $\Aau(2i, \Ss)$
are exactly the states and transitions of $\Aa$ for levels $0$ to $2i-1$. 
The set of states at level $2i$ is
\[ 
    Q^{(2i)} = \set{(f,r) \colon f\in\Ss, r\in\dom{f} }\ .
\]
So a state at layer $2i$ is a summary function and a \emph{use counter}
indicating the part of the summary that has been used. 

For technical reasons we will also need one state at layer $2i+1$. We set
$Q^{(2i+1)}=\set{\bullet}$. 

The transitions $\dq^{\uparrow (2i)}$ and $\da^{\uparrow (2i)}$ are defined as follows.
\begin{align*}
	\text{in } \dq^{\uparrow (2i)} \text{ we have }& f(1) \trans{a} (f(2),(f,3))&\quad
	\text{if }  f \in \Ss\\
	\text{in } \da^{\uparrow (2i)} \text{ we have }& (f(r),(f,r)) \trans{a} f(r+1)&\quad
	\text{if } r = \maxdom{f}-1
\end{align*}

These transitions imply that for every node created at level $2i$, the automaton
guesses a summary and sets the summary's use counter to $3$. 
It is $3$ and not $1$ because the first two values of $f$ are used for the
creation of the node.
The node can be deleted once this bounded counter value is maximal. 

% Let $\set{\s_q}_{q\in Q_{2i}}$ be a set of summaries of level $2i$ indexed by the states from $Q_{2i}$. 
% We define $\Aau(2i,\set{\s_q}_{q\in Q_{2i}})$.
% The states of $\Aau(2i,\set{\s_q}_{q\in Q_{2i}})$ are the states of $\Aa$ for
% levels $0,\dots, 2i-1$. 
% The set of states at level $2i$ is $\set{(f,r) : r\in s_q\text{ for some $q$},
% r\in\dom{f}}$.

% This time, the moves up to level $2i-1$ are the same as in $\Aa$.
% For level $2_i$ we define
% \begin{multline*}
% 	\text{in $\addu_{2i}$ we have $(\quu,\qu)\trans{a} ((\quu,\qu,(f,1))$}\\
% 	\text{if $(\quu,\qu)\trans{a}(\qpuu,\qpu,\qi)$ is in $\add_{2i+1}$ and $f\in\s_{\qi}$.}
% \end{multline*}
% \begin{multline*}
    % 	\text{if $r=\maxdom(f)-1$}
% 	\text{in $\delu_{2i}$ we have $(f(r),(f,r))\trans{a} f(r+1)$}\\
% \end{multline*}
% These transitions mean that for every $2i$ node created automaton guesses a
% summary, and initiates its use index to $1$.
% When the use index gets to maximal value, the node can be deleted. 

Finally, we define the transitions in $\dq^{\uparrow (2i+1)}$ and $\da^{\uparrow (2i+1)}$:
\begin{align*}
	\text{In } \dq^{\uparrow (2i+1)} \text{ we have }& (f(r),(f,r))\trans{a} (f(r+1),(f,r+2),\bullet)\\
	& \qquad\text{if } r<\maxdom{f}-1\\
    \text{In } \da^{\uparrow (2i+1)} \text{ we have }& (f(r),(f,r),\bullet)\trans{a} (f(r),(f,r))\\
    & \text{if }  r=\maxdom{f}-1
\end{align*}
So the automaton creates a child node whenever it uses a summary.
The use counter is increased by $2$ at such a transition.
Once the use counter cannot be increased anymore, $\da^{\uparrow (2i+1)}$
provides transitions for deleting children at layer $2i+1$.
No other transitions are applicable at this point. 
Once there are no children, the root can be removed by a $\da^{\uparrow(2i)}$ transition.

% \begin{multline*}
% 	\text{in $\addu_{2i+1}$ we have $(f(r),(f,r))\trans{a} (f(r+1),(f,r+2),\bullet)$}\\
% 	\text{if $r<\maxdom(f)-1$}
% \end{multline*}
% \begin{multline*}
% 	\text{in $\delu_{2i+1}$ we have $(f(r),(f,r),\bullet)\trans{a} (f(r),(f,r))$}\\
% 	\text{if $r=\maxdom(f)-1$}
% \end{multline*}
% With these transitions, each time automaton uses a summary it creates a node of
% level $2i+1$. 
% Every time the use index is incremented by $2$.
% Then there are transitions allowing to delete all these nodes when the maximal
% value of a use index is reached. 

The next lemma states formally the relation between the two automata we have
introduced and the original one.
Recall that $\Aad$ is used to define a set of summaries.
The lemma is proved by stitching runs of $\Aau$ and $\Aad$.

\begin{lemma}
	For every $k$-level automaton $\Aa$ and level $2i<k$,
	$\Aa$ accepts a trace iff $\Aau(2i,\summary(\Aa,2i))$
	accepts a trace. 
\end{lemma}

The next lemma shows how to use summaries of level $2i+1$ to compute summaries
at level $2i$. 

\begin{lemma}\label{lem:summary-step}
	Take  a summary $f$ of some level $2i$, and consider $\Bb=\Aad(2i,f)$.
	Then $\Bb$ accepts some trace iff $\Bbu(2,\summary(\Aa,2i+2))$ accepts some trace.
\end{lemma}
\begin{proof}
	Follows from $\summary(\Aa,2i+2)=\summary(\Bb,2)$ and the previous lemma.
\end{proof}

The lemma reduces the task of computing summaries to checking emptiness of
automata with $3$ layers. 
In the next subsection we show how to reduce the later problem to the
reachability problem in VASS.
With this lemma we can compute $\summary(\Aa,2i)$ inductively. 
Once we compute $\summary(\Aa,2)$, we can reduce testing emptiness of
$\Aau(2,\summary(\Aa,2))$ to VASS reachability.
This turns out to be degenerate case of computing summaries, so the same
technique as for computing summaries applies.

\subsection*{Computing summaries}

% A summary for a node of level $2i$ is a set of possible sequences of
% interactions with nodes of levels $2i-1$ and $2i-2$ from the moment when node
% $2i$ is created with state $q$ to the moment when it is deleted.
% Such a sequence of interactions is encoded as a function $f:\set{1,\dots,2d}\to
% Q_{2i-2}\times Q_{2i-1}$ for some $d<\maxrank$.
% Value $f(2r-1)=(q_{-2},q_{-1})$ says that $r$-th interaction of the node at
% level $2i$ tests if the states at levels $2i-2$ and $2i-1$ are $q_{-2}$ and
% $q_{-1}$, respectively. 
% Value $f(2r)=(q'_{-2},q'_{-1})$ says what are the new states at levels $2i-2$
% and $2i-1$  after the interaction.
% So our task is to compute the set $\summary(2i,q)$ of all summaries coming from
% $(2i)$-runs (to be defined). 

We compute $\summary(\Aa, 2i)$ assuming that we know $\summary(\Aa, 2i+2)$.
For this we use Lemma~\ref{lem:summary-step}. 
We reduce testing emptiness of $\Bbu(2,\summary(\Aa,2i+2))$ from that lemma to VASS
reachability.
Since presenting a VASS directly would be quite unreadable, we present a
nondeterministic program that will use variables ranging 
over bounded domains and some fixed set of non-negative counters. 
By construction, every counter will be tested for $0$ only at the end of the
computation.
This structure allows us to emulate our nondeterministic program in a VASS, such that acceptance by the program is equivalent to reachability of a particular configuration in the VASS.

% The part of the data tree we are looking at is 
%\includegraphics[scale=.7]{data-tree.png}

% More precisely we assume that all the computations involving nodes on level
% $(2i+2)$ are captured by summaries: $\summary(2i+2,q')$ for $q'\in Q_{2i+2}$.

We fix a summary $\wf$ of level $2i$. Observe that the number of summaries at
level $2i$ is bounded, and so it is sufficient to check whether a given
candidate summary $\wf$ is a valid summary. 
% So it is enough to show how to check if a potential summary $\wf$ does indeed
% belong to $\summary(2i,q)$.

The variables of the program are as follows:
\begin{align*}
	\wrr\in~&\dom{\wf}\\
	\state\in~&Q_{2i}\cup\set{\bot}\\
	\state[j]\in~&Q_{2i+1}\cup\set{\bot,\top} & j \in \set{1,\dots,b}\\
	\children[j,f,r]\in~&\mathbb{N} &\text{$f$ summary at level $(2i+2)$, $r \in\dom{f}$}
\end{align*}
Intuitively, $\state$ and $\wrr$ represent a state from $Q^{(0)}$ of
$\Bbu(2,\summary(\Aa,2i+2))$.
The initial configuration is empty so $\state=\bot$.
Variable $\state[j]$, represents the state of  $j$-th child of the root.
By boundedness, the root can have at most $b$ children. 
Value $\state[j]=\bot$ means that the child has not yet been yet created, and
$\state[j]=\top$  that the child has been deleted.
Counter $\children[j,f,r]$ indicates the number of children of the $j$-th child of
the root with a particular summary $f$ of level $2i+2$ and usage counter $r$.

Following these intuitions the initial values of the variables are $\wrr=1$,
$\state=\bot$, $\state[j]=\bot$ for every $j$,  and $\children[j,f,r]=0$ for
every $j$, $f$ and $r$.

% We consider every transition that involves a node at level $2i$. These are all one of the following shapes:
% \begin{align*}
% (\quu, \qu) & \xrightarrow{t} (\qpuu, \qpu, \qpi) \\
% (\quu, \qu, \qi) & \xrightarrow{t} (\qpuu, \qpu, \qpi, \qpii) \\
% (\qi, \qii) & \xrightarrow{t} (\qpi, \qpii, \qpiii)
% \end{align*}

The program $\Test(\wf)$ we are going to write is a set of rules that are executed
nondeterministically.
Either the program will eventually \texttt{accept}, or it will block with no further rules that can be applied. 
We later show that the program has an accepting run for $\wf$ iff
$\wf\in\summary(\Aa,2i)$. The rules of the program refer to transitions of $\Aa$
and simulate the definition of $\Bbu(2,\summary(\Aa,2i+2))$ from Lemma~\ref{lem:summary-step}.
They are defined as follows.

\makeatletter
\newcommand{\alist}[1]{%
    \begin{aligned}[t] & #1\checknextarg}
\newcommand{\checknextarg}{\@ifnextchar\bgroup{\gobblenextarg}{ \end{aligned} }}
\newcommand{\gobblenextarg}[1]{ \\ & #1\@ifnextchar\bgroup{\gobblenextarg}{ \end{aligned} }}
\makeatother

\newcommand{\vassrule}[2]{
    \begin{aligned}
    \textbf{if}~~ & #1 \\
    \textbf{then}~~ & #2
    \end{aligned}
}

\paragraph{Initializing the root}
We have a rule
\begin{equation*}
    \vassrule{\state=\bot}
    {\alist
        {\state=\qpi}
        {\wrr=3}
    }
\end{equation*}
for every transition $f(1)\trans{a} (f(2),\qpi)$ in $\dq^{(2i)}$.

\paragraph{Removing the root and accepting.} The program is able to accept when
it has completed all of its interaction with the outside world. Observe that
this is the only time that the counters are tested for zero.
Since this occurs at the end of the program, it can be easily checked by VASS reachability.
\[
\vassrule
{\alist
    {\state=\qi}
    {\wrr  = \maxdom{\wf}-1}
    {\forall j \colon \state[j] = \top}
    {\forall (j,f,r) \colon \children[j,f,r] = 0}
}
{\texttt{accept}}
\]
for every $(f(\wrr),\qi)\trans{a} f(\wrr+1)$ in $\da^{(2i)}$.

\paragraph{Adding a node at level $2i+1$.} We ensure that we are in the correct
state and ensure that the summary we are testing aligns with some transition
from the automaton.
\begin{equation*}
    \vassrule%
    {\alist
        {\state = \qi}
        {\wf(\wrr) = (\quu, \qu)}
        {\wf(\wrr+1) = (\qpuu, \qpu)}
        {\wrr + 2 < \maxdom{\wf}}
        {\exists j \colon \state[j] = \bot}
    }%
    {\alist
        {\state := \qpi}
        {\state[j] := \qpii}
        {\wrr = \wrr + 2}
    }
\end{equation*}
for every transition 
\begin{equation*}
   (\quu, \qu, \qi) \xrightarrow{t} (\qpuu, \qpu, \qpi, \qpii) \in \delta_\Q^{(2i+1)}
\end{equation*}

\paragraph{Removing a node at level $2i+1$.} We delete  a child according to
some transition from $\delta_\Q^{(2i+1)}$. While the zero test (ensuring $j$ is
a leaf) is not performed here directly, no further operations will be made on
children counters of this child and hence the zero test performed at the end of
the simulation does the job.
\begin{equation*}
    \vassrule
    {\alist
        {\state = \qi}
        {\wf(\wrr) = (\quu, \qu)}
        {\wf(\wrr+1) = (\qpuu, \qpu)}
        {\wrr + 2 < \maxdom{\wf}}
        {\exists j \colon \state[j] = \qii}
    }
    {\alist
        {\state := \qpi}
        {\state[j] := \top}
        {\wrr = \wrr + 2}
    }
\end{equation*}
for every transition
\begin{equation*}
    (\quu, \qu, \qi, \qii) \xrightarrow{t} (\qpuu, \qpu, \qpi) \in \delta_\A^{(2i+1)} 
\end{equation*}

\paragraph{Adding a node at level $2i+2$.} Firstly we ensure that there is some
child $j$ where such a node can be appended. We simulate creation of a child by
nondeterministically choosing a summary and increasing the corresponding
unbounded counter. Index $3$ in $\children[j,f,3]$ means that this child is
after the first interaction with its ancestors at levels $2i$ and $2i+1$, that
happened at its creation.
\begin{equation*}
    \vassrule
    {\alist
        {\state      = \qi}
        {\exists j \colon \state[j]   = \qii}
    }
    {\alist
        {\state      = \qpi}
        {\state[j]   = \qpii}
        {\children[j,f,3] \text{ += } 1} 
        {\text{ for some $f \in \summary(2i+2)$ s.t.}}
        {\text{\qquad \qquad \qquad $f(1)=(\qi,\qii)$ and $f(2)=(\qpi,\qpii)$}}
    }
\end{equation*}
% for every 
% \begin{equation*}
%     (\qi, \qii) \xrightarrow{t} (\qpi, \qpii, \qpiii)) \in \delta_Q^{(2i+2)}    
% \end{equation*}

% \noindent\If\ $\state=q_0$ and $\exists b. \state[b]=q_1$\\ 
% \Then\ $\state:=q'_0$; $\state[b]:=q'_1$; $\inc(\children[b,f,1])$ for some
% $f\in\summary(2i+2,q'_2)$ 

% So we check if there is a child $b$ of the node where the rule can be applied,
% and then simulate creation of a new child by guessing its summary and increasing
% the corresponding counter.
% Index $1$ in $\children[b,f,1]$ means that the child has not yet interacted with
% its ancestors.

% \includegraphics[scale=.5]{add-2ip3.png}

\paragraph{Progressing a child at level $2i+2$.} We identify an appropriate
child $j$ which itself has a child in state $(f, r)$. We use the
test $r + 2 < \maxdom{f}$ to ensure that the last interaction of the node is
reserved for deletion of our root node.
\begin{equation*}
    \vassrule
    {\alist
        {\state = \qi}
        {
        \exists (j,f,r) \colon \state[j] = \qii\\
               & \hspace{15mm} \text{ and } f(r) = (\qi, \qii) \\
               & \hspace{15mm} \text{ and } f(r+1) = (\qpi, \qpii)\\
               & \hspace{15mm} \text{ and } (r+2) < \maxdom{f} \\
               & \hspace{15mm} \text{ and } \children[j,f,r] \geq 1
        }
    }
    {\alist
        {\state := \qpi}
        {\state[j] := \qpii}
        {\children[j,f,r+2] \text{ += } 1}
        {\children[j,f,r] \text{ -= } 1}
    }
\end{equation*}
Observe that the test $\children[j,f,r] \geq 1$ can be simulated by a VASS because
we have $\children[j,f,r] \text{ -= } 1$ in the statement that follows.

% \noindent\If\ $\state=q_0$ and $\exists b,f,r.\ \state[b]=q_1$  and
% $f(r)=(q_0,q_1)$ and $f(r+1)=(q'_0,q'_1)$ and
% $(r+2)<\maxdom(f )$ and $\dec(\children[b,f,r])$\\
% \Then\ $\state:=q'_0$; $\state[b]:=q'_1$; $\inc(\children[b,f,r+2])$

% Here we find an appropriate child $b$ that has itself an appropriate child
% represented by a non-zero value of a counter $\children[b,f,r]$. 
% The test $(r+2)<\maxdom(f)$ is need to ensure that we do not use the last
% interaction of the node. 
% The last interaction is by definition, a deletion of the node, so it cannot be
% used in this case.
% In this case we change the sate, decrease the counter, and increase the counter
% representing the change of state of the child of $b$.

\paragraph{Removing a node at level $2i+2$.} We find a child which has completed
its summary to the point that it can now be removed. We use the last values in
$f$ to determine how to remove the node. 
\begin{equation*}
    \vassrule
    {\alist
        {\state = \qi}
        {
        \exists (j,f,r) \colon \state[j] = \qii\\
               & \hspace{15mm} \text{ and } f(r) = (\qi, \qii) \\
               & \hspace{15mm} \text{ and } f(r+1) = (\qpi, \qpii)\\
               & \hspace{15mm} \text{ and } (r+1) = \maxdom{f}\\
               & \hspace{15mm} \text{ and } \children[j,f,r] \geq 1
        }
    }
    {\alist
        {\state := \qpi}
        {\state[b] := \qpii}
        {\children[j,f,r] \text{ -= } 1}
    } 
\end{equation*}

\begin{lemma}
    Program $\Test(\wf)$ accepts iff
    $\wf\in\summary(\Aa,2i)$. 
\end{lemma}
\begin{proof}
    By definition, $\wf\in\summary(\Aa,2i)$ if automaton $\Bb=\Aad(2i,\wf)$ accepts a
    trace.
    By Lemma~\ref{lem:summary-step} this is equivalent to
    $\Bbu(2,\summary(\Aa,2i+2))$ accepting some trace.
    It can be checked that the instructions of $\Test(\wf)$ correspond
    one-to-one to transitions of $\Bbu(2,\summary(\Aa,2i+2))$. 
    So an accepting run of $\Test(\wf)$ can be obtained from a trace accepted by
    $\Bbu(2,\summary(\Aa,2i+2))$, and vice versa.
\end{proof}

% \includegraphics[scale=.5]{del-2ip2.png}

% \noindent\If\ $\state=q_0$ and $\exists b,f,r.\ \state[b]=q_1$  and
% $f(r)=(q_0,q_1)$ and $f(r+1)=(q'_0,q'_1)$ and $r+1=\maxdom(f)$ and
% $\dec(\children[b,f,r])$\\
% \Then\ $\state:=q'_0$; $\state[b]:=q'_1$

% \includegraphics[scale=.5]{del-2ip3.png}

% [ This already appears above ]

% \noindent\If\ $\state=q_0$ and $\exists b,f,r.\ \state[b]=q_1$  and
% $f(r)=(q_0,q_1)$ and $f(r+1)=(q'_0,q'_1)$ and $r+2<\maxdom(f)$ and
% $\dec(\children[b,f,r])$\\
% \Then\ $\state:=q'_0$; $\state[b]:=q'_1$; $\inc(\children[b,f,r+2])$

% \includegraphics[scale=.5]{add-2ip1.png}

% \noindent\If\ $\wf(\wrr)=(q_{-2},q_{-1})$ and $\state=q_0$ and
% $\wf(\wrr+1)=(q'_{-2},q'_{-1})$ and $\wrr+2<\maxdom(\wf)$\\
% \Then\ $\state:=q'_0$; $\wrr=\wrr+2$, $\state[b]:=q'_1$ for some $b$ s.t.\
% $\state[b]=\bot$

% \includegraphics[scale=.5]{del-2i.png}

% \noindent\If\ $\wf(\wrr)=(q_{-2},q_{-1})$ and $\state=q_0$ and
% $\wf(\wrr+1)=(q'_{-2},q'_{-1})$ and $\wrr+1=\maxdom(\wf)$\\
% \Then\ ACCEPT

% \includegraphics[scale=.5]{del-2ip1.png}

% \noindent\If\ $\wf(\wrr)=(q_{-2},q_{-1})$ and $\state=q_0$ and
% $\wf(\wrr+1)=(q'_{-2},q'_{-1})$ and $r+2<\maxdom(f)$ and $\exists b$.\
% $\state[b]=q_{1}$ and $\forall f,r.$\ $\children[b,f,r]=0$\\	
% \Then\ $\state:=q'_0$; $\wrr=\wrr+2$, $\state[b]:=\top$

% This is the only place where we check counters for $0$. Observe that it happens
% only when the value of $\state[b]$ is set to $\top$. 
% As this can happen only once for each $b$ this shows that every counter is
% tested for $0$ only once. 

}
%\input{sections-apx/02-from-algol}
% !TEX root =  main.tex

\section{Additional material for Section~\ref{sec:tola}}
\label{apx:tola}

\subsection{Proof of Theorem~\ref{thm:trans}} 

Because every $\fica$-term can be converted to $\beta\eta$-normal form, we use induction on the structure of such normal forms. 
The base cases are:
\begin{itemize}
\item $\seq{\Gamma}{\skipcom:\comt}$: $Q^{(0)}= \{0\}$, $\dagger \trans{\mrun} 0$,\,\, $0 \trans{\mdone} \dagger$;
\item $\seq{\Gamma}{\divcom_\comt:\comt}$: $Q^{(0)}= \{0\}$, $\dagger \trans{\mrun} 0$;
\item $\seq{\Gamma}{\divcom_\theta:\theta}$: $Q^{(0)}= \{0\}$, $\dagger \trans{m} 0$,
assuming $\theta=\theta_l\rarr\cdots\rarr\theta_1\rarr\beta$ and $m$ ranges over
question-moves from $M_{\sem{\beta}}$;
\item $\seq{\Gamma}{i:\expt}$: $Q^{(0)}= \{0\}$, $\dagger \trans{\q} 0$,\,\, $0 \trans{i} \dagger$.
\end{itemize}
Observe that they are clearly even-ready, because only one node is ever created.

The remaining cases are inductive. 
Note that  we will use $\mm$ to range over $\alp{\seq{\Gamma}{\theta}}+\{\eq,\ea\}$, i.e. not only $M_{\sem{\seq{\Gamma}{\theta}}}$, and
recall our convention that $m\in M_{\sem{\seq{\Gamma}{\theta}}}$ stands for $m^{(\epsilon,0)}$.

When referring to the inductive hypothesis, i.e. the automaton constructed for some subterm $M_i$,
we will use the subscript $i$ to refer to its components, e.g. $Q_i^{(j)}$,  $\trans{\mm}_i$ etc.
In contrast, we shall use $Q^{(j)}$, $\trans{\mm}$ to refer to the automaton that is being constructed.
The construction will often use inference lines $\frac{\qquad}{\qquad}$ to indicate that the transitions listed under the line should be added
to the new automaton as long as the transitions listed above the line are present in an automaton given by the inductive hypothesis.
Sometimes we will invoke the inductive hypothesis for several terms, which can provide several automata of different depths.
Without loss of generality, we will then assume that they all have the same depth $k$, because an automaton of lower depth can be
viewed as one of higher depth. 

\begin{itemize}
\item $\seq{\Gamma}{\arop{M_1}:\expt}$: $Q^{(j)}=Q_1^{(j)}$ ($0\le j\le k$). In order to interpret unary operators it suffices
to modify transitions  carrying the final answer in the automaton for $M_1$. Formally, this is done as follows.
\[
\frac{(q_1^{(0)}, \cdots, q_1^{(j)}) \trans{\mm}_1 (r_1^{(0)},\cdots, r_1^{(j')}) \qquad \mm\neq i}{(q_1^{(0)}, \cdots, q_1^{(j)}) \trans{\mm} (r_1^{(0)},\cdots, r_1^{(j')})}
\qquad
\frac{q_1^{(0)} \trans{i}_1 \dagger}{q_1^{(0)} \trans{\widehat{\mathbf{op}}(i)} \dagger}
\]
Above, $j$ ranges over $\{-1,0,\cdots, k\}$, so that $(q_1^{(0)}, \cdots, q_1^{(j)})$ can also stand for $\dagger$.
Even-readiness is preserved by the construction, because the 
configuration graph of the original automaton is preserved.

\item $\seq{\Gamma}{M_1|| M_2}:\comt$:  $Q^{(0)} =  Q_1^{(0)} \times Q_2^{(0)}$,
$Q^{(j)}= Q_1^{(j)}+Q_2^{(j)}$  $(1\le j\le k)$.
The first group of transitions activate and terminate the two components respectively:
\[
\frac{\dagger\trans{\mrun}_1 q_1^{(0)}\qquad \dagger\trans{\mrun}_2 q_2^{(0)}}{\dagger\trans{\mrun}(q_1^{(0)},q_2^{(0)}) }\qquad
\frac{q_1^{(0)}\trans{\mdone}_1 \dagger\qquad q_2^{(0)}\trans{\mdone}_2 \dagger}{(q_1^{(0)},q_2^{(0)})\trans{\mdone}\dagger}.
\]
The remaining transitions allow each component to progress.
\[
\frac{(q_1^{(0)}, \cdots, q_1^{(j)}) \trans{\mm}_1 (r_1^{(0)},\cdots, r_1^{(j')})\qquad q_2^{(0)}\in Q_2^{(0)}\qquad \mm\neq\mrun,\mdone}{((q_1^{(0)},q_2^{(0)}), \cdots, q_1^{(j)}) \trans{\mm} ((r_1^{(0)},q_2^{(0)}),\cdots, r_1^{(j')})}\qquad 
\]
\[
\frac{q_1^{(0)}\in Q_1^{(0)}\qquad (q_2^{(0)}, \cdots, q_2^{(j)}) \trans{\mm}_2 (r_2^{(0)},\cdots, r_2^{(j')})\qquad \mm\neq\mrun,\mdone
}{((q_1^{(0)},q_2^{(0)}), \cdots, q_2^{(j)}) \trans{\mm} ((q_1^{(0)},r_2^{(0)}),\cdots, r_2^{(j')})}\]
Even-readiness at even levels different from $0$ follows from even-readiness of the automata obtained in IH, because the construction simply runs them
concurrently without interaction at these levels. For level $0$, we observe that, whenever the root reaches state $(q_1^{(0)},q_2^{(0)})$,
even-readiness of the two automata implies that each of them has removed all nodes below the root, i.e. the root will be a leaf.

\item $\seq{\Gamma}{M_1;M_2}:\comt$: $Q^{(i)} = Q^{(i)}_1 + Q^{(i)}_2$ ($0\le i\le k$).
We let the automaton for $M_1$ run first (except for the final step
$\mdone$):
\[
\frac{\dagger\trans{\mrun}_1 q_1^{(0)}}{\dagger\trans{\mrun}q_1^{(0)}}\qquad
\frac{(q_1^{(0)}, \cdots, q_1^{(j)}) \trans{\mm}_1 (r_1^{(0)},\cdots, r_1^{(j')})\qquad \mm\neq\mdone}{(q_1^{(0)}, \cdots, q_1^{(j)}) \trans{\mm} (r_1^{(0)},\cdots, r_1^{(j')})}.
\]
Whenever the automaton $M_1$ can terminate, we pass control to
the automaton for $M_2$ via
\[
\frac{q_1^{(0)}\trans{\mdone}_1 \dagger \qquad
\dagger \trans{\mrun}_2 q_2^{(0)} \qquad 
q_2^{(0)} \trans{\mm}_2 (r_2^{(0)},\cdots, r_2^{(j')})\qquad \mm\neq \mrun}{
q_1^{(0)} \trans{\mm} (r_2^{(0)},\cdots, r_2^{(j')})}
\]
and allow it to continue
\[
\frac{
(q_2^{(0)}, \cdots, q_2^{(j)}) \trans{\mm}_2 (r_2^{(0)},\cdots, r_2^{(j')})\qquad \mm\neq \mrun}{(q_2^{(0)}, \cdots, q_2^{(j)}) \trans{\mm} (r_2^{(0)},\cdots, r_2^{(j')})}.
\]
Note that the construction relies crucially on even-readiness 
of the automaton for $M_1$, because we move to the automaton for $M_2$
as soon as the automaton $M_1$ arrives at a configuration with level-$0$
state $q_1^{(0)}$ such that $q_1^{(0)}\trans{\mdone}_1 \dagger$.
Thanks to even-readiness, we can conclude that the root will be the only
node in the configuration then and the transition can indeed fire, i.e. $M_1$ is really  finished.

Even-readiness of the new automaton follows from the fact that the original automata were even-ready, because we are re-using their transitions (and when the automaton for $M_2$ is active, 
that for $M_1$ has not left any nodes).

\item $\seq{\Gamma}{M_1;M_2:\beta}$

The general case is nearly the same as the $\comt$ case presented
above except that we need to keep track of what initial move has been played
in order to  perform the transition to $M_2$ correctly.
This is especially important for $\beta=\vart,\semt$, where there are multiple initial moves.
This extra information will be stored at level $0$, while the automaton 
corresponding to $M_1$ is active. Below we present a general construction
parameterized by the set $I$ of initial moves.
The set $I$ is defined as follows.
\begin{itemize}
    \item $\beta=\comt$: $I=\{\mrun\}$
    \item $\beta=\expt$: $I=\{\mq\}$
    \item $\beta=\vart$: $I=\{\mread,\mwrite{0},\cdots, \mwrite{\imax}\}$
    \item $\beta=\semt$: $I=\{\mgrb,\mrls\}$
\end{itemize}

\bigskip

States
\[\begin{array}{rcl}
Q^{(0)} &=& (Q^{(0)}_1\times I) + Q^{(0)}_2\\
Q^{(i)} &=& Q^{(i)}_1 + Q^{(i)}_2\qquad (0 < i\le k)
\end{array}\]
Transitions
\[
\frac{\dagger\trans{\mrun}_1 q_1^{(0)}\qquad x\in I}{\dagger\trans{x} (q_1^{(0)},x)}
\]
\[
\frac{(q_1^{(0)}, \cdots, q_1^{(j)}) \trans{\mm}_1 (r_1^{(0)},\cdots, r_1^{(j')})\qquad \mm\neq\mdone\qquad x\in I}{((q_1^{(0)},x), \cdots, q_1^{(j)}) \trans{\mm} ((r_1^{(0)},x),\cdots, r_1^{(j')})}.
\]
\[
\frac{q_1^{(0)}\trans{\mdone}_1 \dagger \qquad
\dagger \trans{x}_2 q_2^{(0)} \qquad 
q_2^{(0)} \trans{\mm}_2 (r_2^{(0)},\cdots, r_2^{(j')})\qquad x\in I\qquad \mm\not\in I}{
(q_1^{(0)},x) \trans{\mm} (r_2^{(0)},\cdots, r_2^{(j')})}
\]
\[
\frac{
(q_2^{(0)}, \cdots, q_2^{(j)}) \trans{\mm}_2 (r_2^{(0)},\cdots, r_2^{(j')})\qquad \mm\not\in I}{(q_2^{(0)}, \cdots, q_2^{(j)}) \trans{\mm} (r_2^{(0)},\cdots, r_2^{(j')})}
\]
None of the $M_1;M_2$ cases requires an adjustment of pointers, because the inherited indices are accurate.

\item $\seq{\Gamma}{\newin{x\aasg i}{M_1}}:\beta$.  
By~\cite{GM08},
$\sem{\seq{\Gamma}{\newin{x\aasg i}{M_1}}}$ can be obtained by
\begin{itemize}
\item  first restricting $\sem{\seq{\Gamma,x}{M_1}}$ to plays in which
the moves $\mread^x$, $\mwrite{n}^x$ are followed immediately by answers, 
\item selecting only those plays in which each answer to a $\mread^x$-move is consistent
with the preceding $\mwrite{n}^x$-move (or equal to $i$, if no preceding $\mwrite{n}^x$ was made),
\item erasing all moves related to $x$, e.g. those of the form $m^{(x,\rho)}$. 
\end{itemize}

To implement the above recipe, we will lock the automaton after each $\mread^x$- or $\mwrite{n}^x$-move, so that only an answer to that move can be played next.
Technically, this will be done by annotating the level-$0$ state with a $\mathit{lock}$-tag. 
Moreover,  at level $0$, we will also keep track of the current value of $x$. This will help us ensure that
 answers to $\mread^x$  are consistent  with the stored value and that $\mwrite{n}^x$ transitions cause the right change.
Eventually, all moves with the $x$ subscript will be replaced with $\eq,\ea$  to model hiding.

Accordingly, 
we take $Q^{(0)}=(Q_1^{(0)} + (Q_1^{(0)}\times \{\mathit{lock}\})) \times\{0,\cdots,\imax\}$ and $Q^{(j)} = Q_1^{(j)}$ ($1\le j\le k$).
First, we make sure that the state component is initialised to $i$ and that it can be arbitrary at the very end:
\[
\frac{\dagger \trans{\mm}_1 q_1^{(0)}}{\dagger\trans{\mm} (q_1^{(0)},i)}
\qquad
\frac{q_1^{(0)} \trans{\mm}_1 \dagger\qquad 0\le n\le \imax}{(q_1^{(0)},n)\trans{\mm} \dagger}.
\]
Transitions involving moves different from $\mwrite{z}^x$, $\mok^x$, $\mread^x$, $z^x$ (and the moves handled above) progress unaffected 
while preserving  $n$ (the current value of $x$ recorded at level $0$):
\[
\frac{\begin{array}{lcl}
(q_1^{(0)},\cdots, q_1^{(j)})\trans{\mm}_1 (r_1^{(0)},\cdots, r_1^{(j')}) &\quad & \mm\neq \mread^x, z^x, \mwrite{z}^x,\mok^x\\
&& 0\le j,j'\qquad 0\le n\le \imax\qquad 
\end{array}
}{((q^{(0)}_1,n),\cdots, q_1^{(j)})\trans{\mm}((r_1^{(0)},n),\cdots, r_1^{(j')})}.
\]
Transitions using $\mread^x$, $\mwrite{z}^x$ add a lock at level $0$.
The lock can be lifted only if a corresponding answer is played (because of the lock, a unique $\mwrite{z}^x$ or $\mread^x$ will be pending).
Its value must be consistent with the value of $x$  recorded at level $0$.
\[
\frac{(q_1^{(0)},\cdots, q_1^{(j)})\trans{\mwrite{z}^{(x,\rho)}}_1 (r_1^{(0)},\cdots, r_1^{(j')})\qquad 0\le n,z\le \imax}{
((q_1^{(0)},n),\cdots, q_1^{(j)})\trans{\eq} ((r_1^{(0)},\lock, z),\cdots, r_1^{(j')})}
\]
\[
\frac{(q_1^{(0)},\cdots, q_1^{(j)})\trans{\mread^{(x,\rho)}}_1 (r_1^{(0)},\cdots, r_1^{(j')})\qquad 0\le n\le \imax}{
((q_1^{(0)},n),\cdots, q_1^{(j)}) \trans{\eq} ((r_1^{(0)},\lock,n),\cdots, r_1^{(j')}))}
\]

\[
\frac{(r_1^{(0)},\cdots, r_1^{(j')})\trans{\mok^x}_1 (t_1^{(0)},\cdots, t_1^{(j)})\qquad 0\le n\le \imax}{
((r_1^{(0)},\lock, n),\cdots, r_1^{(j')})\trans{\ea} ((t_1^{(0)},n),\cdots, t_1^{(j)})}
\]

\[
\frac{(r_1^{(0)},\cdots, r_1^{(j')})\trans{n^{x}}_1 (t_1^{(0)},\cdots, t_1^{(j)})\qquad 0\le n\le \imax}{
((r_1^{(0)},\lock, n),\cdots, r_1^{(j')})\trans{\ea} ((t_1^{(0)},n),\cdots, t_1^{(j)})}
\]
As the construction involves running the original automaton and transitions corresponding to P-answers
are not modified, even-readiness follows directly from IH. For the same reason, the indices corresponding
to justification pointers need no adjustment.

 \item The case of $\newsem{x\aasg i}{M_1}$  is similar to $\newin{x\aasg i}{M_1}$. We represent the state of the semaphore using an additional bit at level $0$,
where $0$ means free and $1$ means taken.
We let $Q^{(0)}=(Q_1^{(0)} + (Q_1^{(0)}\times \{\lock\})) \times\{0,1\}$ and $Q^{(j)} = Q_1^{(j)}$ ($1\le j\le k$).
First, we make sure the bit is initialised to $i$ and can be arbitrary at the very end.
\[
\frac{\dagger \trans{\mm}_1 q_1^{(0)}\qquad i=0}{\dagger\trans{\mm} (q_1^{(0)},0)}\qquad
\frac{\dagger \trans{\mm}_1 q_1^{(0)}\qquad i>0}{\dagger\trans{\mm} (q_1^{(0)},1)}\qquad
\frac{q_1^{(0)} \trans{\mm}_1 \dagger\qquad z\in\{0,1\}}{(q_1^{(0)},z)\trans{\mm} \dagger}
\]
Transitions involving moves other than $\mrls^{(x,\rho)}$, $\mgrb^{(x,\rho)}$ and $\mok^x$
proceed as before, while preserving the state of the semaphore.
\[
\frac{
(q_1^{(0)},\cdots, q_1^{(j)})\trans{\mm}_1(r_1^{(0)},\cdots, r_1^{(j')})\qquad z\in \{0,1\} \qquad \mm\neq \mrls^{(x,\rho)}, \mgrb^{(x,\rho)}, \mok^x
}{((q_1^{(0)},z),\cdots, q_1^{(j)})\trans{\mm}((r_1^{(0)},z),\cdots, r_1^{(j')})}
\]
Transitions using $\mrls^{(x,\rho)}$, $\mgrb^{(x,\rho)}$ proceed only if they are compatible with the current state of the semaphore,
as represented by the extra bit.
At the same time, each time $\mgrb^{(x,\rho)}$ or $\mrls^{(x,\rho)}$ is played, 
we lock the automaton so that the corresponding answer can be played next.
The moves are then hidden and replaced with $\eq$ and $\ea$.
\[
\frac{(q_1^{(0)},\cdots, q_1^{(j)})\trans{\mgrb^{(x,\rho)}}_1 (r_1^{(0)},\cdots, r_1^{(j')})}{
((q_1^{(0)},0),\cdots, q_1^{(j)})\trans{\eq} 
((r_1^{(0)},\lock,1),\cdots, r_1^{(j')})}\]

\[
\frac{(q_1^{(0)},\cdots, q_1^{(j)})\trans{\mrls^{(x,\rho)}}_1 (r_1^{(0)},\cdots, r_1^{(j')})}{
((q_1^{(0)},1),\cdots, q_1^{(j)})\trans{\eq} ((r_1^{(0)},\lock,0),\cdots, r_1^{(j')})}
\]

\[
\frac{(r_1^{(0)},\cdots, r_1^{(j')})\trans{\mok^x}_1 (t_1^{(0)},\cdots, t_1^{(j)})\qquad z\in\{0,1\}}{
((r_1^{(0)},\lock,z),\cdots, r_1^{(j')})\trans{\ea} ((t_1^{(0)},z),\cdots, t_1^{(j)})}
\]

\item $\seq{\Gamma}{f M_h \cdots M_1:\comt}$ with $(f: \theta_h\rarr\cdots\rarr\theta_1\rarr\comt)\in\Gamma$. 
Note that this also covers the case $f:\comt$.
$Q^{(0)} = \{0,1,2\}$, $Q^{(1)}=\{0\}$,  $Q^{(j+2)}= Q^{(j)}$ ($0\le j\le k$).
First we add transitions corresponding to calling and returning from $f$: $\dagger \trans{\mrun} 0$,\quad  $0\trans{\mrun^f} (1,0)$,\quad 
$(1,0)\trans{\mdone^f} 2$,\quad $2\trans{\mdone} \dagger$.

In state $(1,0)$ we want to enable the environment to spawn an unbounded number of copies of each of $\seq{\Gamma}{M_u:\theta_u}$ ($1\le u\le h$).
This is done through the following rules, which embed the actions of the automata for $M_u$ while relabelling the moves.
\begin{itemize}
\item Moves from $M_u$ corresponding to $\theta_u$ obtain an additional annotation $f u$, as they are now the $u$th argument of
$f:\theta_h\rarr\cdots\rarr\theta_1\rarr\comt$.
\[
\frac{ (q_u^{(0)},\cdots, q_u^{(j)}) \trans{m^{(\vec{i},\rho)}}_u (q_u^{(0)},\cdots, q_u^{(j')})}{(1,0,q_u^{(0)},\cdots, q_u^{(j)}) \trans{m^{(f u \vec{i},\rho)}} (1,0,q_u^{(0)},\cdots, q_u^{(j')})}
\]
Note that above we mean $j,j'$ to range over $\{-1,0,\cdots, k\}$, so that $(q_u^{(0)}, \cdots, q_u^{(j)})$ and 
$(q_u^{(0)},\cdots, q_u^{(j')})$ can also stand for $\dagger$.
The pointer structure is simply inherited  in this case, but an additional pointer needs to be created to $\mrun^f$ from the old initial move for $M_u$, i.e. $m^{(\epsilon,0)}$, 
which did not have a pointer earlier.
Fortunately, because we also use $\rho=0$ in initial moves to represent the lack of a pointer, by copying $0$ now
we indicate that  the move $m^{fu,\rho}$ points one level up, i.e. at the new $\mrun^f$ move, as required.

\item The moves from $M_u$ that originate from $\Gamma$, 
i.e. moves of the form $m^{(x_v\vec{i},\rho)}$ ($1\le v \le l$), where $(x_v\in \theta_v)\in\Gamma$, 
need no relabelling except for question moves that should point at the initial move.
These moves correspond to question-tags of the form $m^{(x_v,\rho)}$.
Leaving $\rho$ unchanged in this case would mean pointing at $m^{f u,0}$, whereas we need to point at $\mrun$ instead.
To readjust such pointers, we simply add $2$ to $\rho$, and
preserve $\rho$ in other moves.
\[
\frac{ (q_u^{(0)},\cdots, q_u^{(j)}) \trans{m^{(x_v,\rho)}}_u (q_u^{(0)},\cdots, q_u^{(j')})\qquad \textrm{$m$ is a question}}{(1,0,q_u^{(0)},\cdots, q_u^{(j)}) \trans{m^{(x_v,\rho+2)}} (1,0,q_u^{(0)},\cdots, q_u^{(j')})}
\]
\[
\frac{ (q_u^{(0)},\cdots, q_u^{(j)}) \trans{m^{(x_v\vec{i},\rho)}}_u (q_u^{(0)},\cdots, q_u^{(j')})\qquad\textrm{$\vec{i}\neq \epsilon$ or ($\vec{i}=\epsilon$ and $m$ is an answer)}}{(1,0,q_u^{(0)},\cdots, q_u^{(j)}) \trans{m^{(x_v \vec{i},\rho)}} (1,0,q_u^{(0)},\cdots, q_u^{(j')})}
\]
\end{itemize}
The construction clearly preserves even-readiness at level $0$.
For other even levels, this follows directly from IH as we are simply
running copies of the automata from IH.

\item $\seq{\Gamma}{f M_h \cdots M_1:\expt}$. Here we follow the same recipe as for $\comt$ except that the initial and final transitions need to be changed from 
\[
\dagger \trans{\mrun} 0\qquad 0\trans{\mrun^f} (1,0)\qquad (1,0)\trans{\mdone^f} 2\qquad 2\trans{\mdone} \dagger
\]
to
\[
\dagger \trans{\mq} 0\qquad 0\trans{\mq^f} (1,0)\qquad (1,0)\trans{i^f} 2^i\qquad 2^i\trans{i} \dagger.
\]

\item $\seq{\Gamma}{f M_h \cdots M_1:\vart}$. Here a slightly more complicated adjustment is needed to account for the two kinds of initial moves.
Consequently, we need to distinguish two copies of $1$, i.e. $1^r$ and $1^w$.
\[
\dagger \trans{\mread} 0\qquad 0\trans{\mread^f} (1^r,0)\qquad (1^r,0)\trans{i^f} 2^i\qquad 2^i\trans{i} \dagger.
\]
\[
\dagger \trans{\mwrite{i}} 0^i\qquad 0^i\trans{\mwrite{i}^f} (1^w,0)\qquad (1^w,0)\trans{\mok} 2\qquad 2\trans{\mok} \dagger.
\]
All the other rules allowing for transitions between states of the form $(1,0,\cdots)$ need to be replicated for $(1^r,0,\cdots)$ and $(1^w,0,\cdots)$.

\item $\seq{\Gamma}{f M_h \cdots M_1:\semt}$. This is similar to the previous case.
To account for the two kinds of initial moves, we use states $1^g$ and $1^r$.
\[
\dagger \trans{\mgrb} 0^g\qquad 0^g\trans{\mgrb^f} (1^g,0)\qquad (1^g,0)\trans{\mok^f} 2^g\qquad 2^g\trans{\mok} \dagger
\]
\[
\dagger \trans{\mrls} 0^r\qquad 0^r\trans{\mrls^f} (1^r,0)\qquad (1^r,0)\trans{\mok^f} 2^r\qquad 2^r\trans{\mok} \dagger
\]
All the other rules allowing for transitions between states of the form $(1,0,\cdots)$ need to be replicated for $(1^r,0,\cdots)$ and $(1^g,0,\cdots)$.

\item $\seq{\Gamma}{\lambda x.M_1: \theta_h\rarr \cdots\rarr \theta_1\rarr \beta}$: This is simply dealt with by renaming labels in the automaton for 
$\seq{\Gamma,x:\theta_h}{M_1: \theta_{h-1}\rarr\cdots\rarr\theta_1\rarr\beta}$: 
tags of the form $m^{(x\vec{i},\rho)}$ must be renamed as $m^{(h\vec{i},\rho)}$.

\item $\seq{\Gamma}{\cond{M_1}{M_2}{M_3}:\beta}$

This case is similar to $M_1;M_2$ except that $M_1$ of type $\expt$, so the associated move is $\mq$ rather than $\mrun$.
Morever, once $M_1$ terminates, the automaton for either $M_2$ or $M_3$ must be activated, as appropriate.

\bigskip

States
\[\begin{array}{rcl}
Q^{(0)} &=& (Q^{(0)}_1\times I) + Q^{(0)}_2 +Q^{(0)}_3\\
Q^{(i)} &=& Q^{(i)}_1 + Q^{(i)}_2 + Q^{(i)}_3\qquad (0 < i\le k)
\end{array}\]
Transitions
\[
\frac{\dagger\trans{\mq}_1 q_1^{(0)}\qquad x\in I}{\dagger\trans{x} (q_1^{(0)},x)}
\]
\[
\frac{(q_1^{(0)}, \cdots, q_1^{(j)}) \trans{\mm}_1 (r_1^{(0)},\cdots, r_1^{(j')})\qquad \mm\not\in\{0,\cdots,\imax\}\qquad x\in I}{((q_1^{(0)},x), \cdots, q_1^{(j)}) \trans{\mm} ((r_1^{(0)},x),\cdots, r_1^{(j')})}.
\]
\[
\frac{q_1^{(0)}\trans{i}_1 \dagger \qquad i>0 \qquad
\dagger \trans{x}_2 q_2^{(0)} \qquad 
q_2^{(0)} \trans{\mm}_2 (r_2^{(0)},\cdots, r_2^{(j')})\qquad x\in I\qquad \mm\not\in I}{
(q_1^{(0)},x) \trans{\mm} (r_2^{(0)},\cdots, r_2^{(j')})}
\]
\[
\frac{q_1^{(0)}\trans{0}_1 \dagger \qquad
\dagger \trans{x}_3 q_3^{(0)} \qquad 
q_3^{(0)} \trans{\mm}_3 (r_3^{(0)},\cdots, r_3^{(j')})\qquad x\in I\qquad \mm\not\in I}{
(q_1^{(0)},x) \trans{\mm} (r_3^{(0)},\cdots, r_3^{(j')})}
\]
\[
\frac{
(q_2^{(0)}, \cdots, q_2^{(j)}) \trans{\mm}_2 (r_2^{(0)},\cdots, r_2^{(j')})\qquad \mm\not\in I}{(q_2^{(0)}, \cdots, q_2^{(j)}) \trans{\mm} (r_2^{(0)},\cdots, r_2^{(j')})}
\]
\[
\frac{
(q_3^{(0)}, \cdots, q_3^{(j)}) \trans{\mm}_3 (r_3^{(0)},\cdots, r_3^{(j')})\qquad \mm\not\in I}{(q_3^{(0)}, \cdots, q_3^{(j)}) \trans{\mm} (r_3^{(0)},\cdots, r_3^{(j')})}
\]
None of the cases requires an adjustment of pointers, because the inherited indices are accurate.
Even-readiness follows directly from IH.

\item $\seq{\Gamma}{\while{M_1}{M_2}:\comt}$: 

\bigskip

States
\[
Q^{(j)}= Q_1^{(j)}+Q_2^{(j)}\qquad 0\le j\le k
\]

Transitions

\[
\frac{\dagger\trans{\mq}_1 q_1^{(0)}}{\dagger\trans{\mrun} q_1^{(0)}} \qquad
\frac{q_1^{(0)}\trans{0}_1 \dagger}{ q_1^{(0)} \trans{\mdone} \dagger}
\]

\[
\frac{(q_1^{(0)}, \cdots, q_1^{(j)}) \trans{\mm}_1 (r_1^{(0)},\cdots, r_1^{(j')})\qquad \mm\not \in\{\mq,0,\cdots,\imax\}}{(q_1^{(0)}, \cdots, q_1^{(j)}) 
\trans{\mm} (r_1^{(0)},\cdots, r_1^{(j')})}
\]

\[
\frac{q_1^{(0)}\trans{i}_1 \dagger \qquad i>0\qquad\dagger \trans{\mrun}_2 q_2^{(0)}\trans{\mm}_2 (r_2^{(0)}, r_2^{(1)})\qquad \mm\neq\mdone}{
q_1^{(0)}\trans{\mm} (r_2^{(0)}, r_2^{(1)})}
\]

\[
\frac{q_1^{(0)}\trans{i}_1 \dagger\qquad i>0 \qquad\dagger \trans{\mrun}_2 q_2^{(0)}\trans{\mdone}_2 \dagger\qquad \dagger\trans{\mq}_1 r_1^{(0)} 
\trans{\mm}_1 (u_1^{(0)}, u_1^{(1)}) \qquad \mm\not\in\{0,\cdots,\imax\}}{
q_1^{(0)}\trans{\mm} (u_1^{(0)},u_1^{(1)})}
\]

\[
\frac{q_1^{(0)}\trans{i}_1 \dagger\qquad i>0 \qquad\dagger \trans{\mrun}_2 q_2^{(0)}\trans{\mdone}_2 \dagger\qquad \dagger\trans{\mq}_1 r_1^{(0)} 
\trans{0}_1 \dagger }{
q_1^{(0)}\trans{\mdone} \dagger}
\]
---
\[
\frac{(q_2^{(0)}, \cdots, q_2^{(j)}) \trans{\mm}_2 (r_2^{(0)},\cdots, r_2^{(j')})\qquad \mm\not \in\{\mrun,\mdone\}}{(q_2^{(0)}, \cdots, q_2^{(j)}) 
\trans{\mm} (r_2^{(0)},\cdots, r_2^{(j')})}
\]

\[
\frac{q_2^{(0)}\trans{\mdone}_2 \dagger \qquad\dagger \trans{\mq}_1 q_1^{(0)}\trans{\mm}_1 (r_1^{(0)}, r_1^{(1)})\qquad \mm\not\in\{0,\cdots,\imax\}}{
q_2^{(0)}\trans{\mm} (r_1^{(0)}, r_1^{(1)})}
\]

\[
\frac{q_2^{(0)}\trans{\mdone}_2 \dagger \qquad\dagger \trans{\mq}_1 q_1^{(0)}\trans{0}_1 \dagger}{
q_2^{(0)}\trans{\mdone} \dagger}
\]

\[
\frac{q_2^{(0)}\trans{\mdone}_2 \dagger\qquad\dagger \trans{\mq}_1 q_1^{(0)}\trans{i}_1 \dagger\qquad i>0\qquad \dagger\trans{\mrun}_2 r_2^{(0)} 
\trans{\mm}_2 (u_2^{(0)},u_2^{(1)})\qquad \mm\neq\mdone}{
q_2^{(0)}\trans{\mm} (u_2^{(0)}, u_2^{(1)})}
\]
As before, no pointers need adjustment, even-readiness is inherited.

\item $\seq{\Gamma}{!M_1:\expt}$

To model dereferencing, it suffices to explore the plays that start with $\mread$ in the automaton for $M_1$, the $\mread$ gets relabelled to $\mq$.

\bigskip

States
\[
Q^{(j)}= Q_1^{(j)}\qquad (0\le j\le k)
\]

Transitions
\[
\frac{\dagger\trans{\mread}_1 q_1^{(0)}}{\dagger \trans{\mq} q_1^{(0)}}\qquad
\frac{(q_1^{(0)}, \cdots, q_1^{(j)}) \trans{\mm}_1 (r_1^{(0)},\cdots, r_1^{(j')})\qquad \mm\neq\mread,\mwrite{i}{},\mok}{(q_1^{(0)}, \cdots, q_1^{(j)}) \trans{\mm} (r_1^{(0)},\cdots, r_1^{(j')})}\qquad 
\]

Note that the second rule will also handle transitions 
with the tag $i$. No pointer readjustment is needed, as the inherited pointers are accurate.
Even-readiness follows from IH.

\item $\seq{\Gamma}{M_1\aasg M_2:\comt}$

For assignment, we first direct the computation into the automaton for $M_2$ and, depending on the final move $i$, continue
in the automaton for $M_1$ as if $\mwrite{i}$ was played.
This is similar to $M_1;M_2$.

\bigskip

States
\[\begin{array}{rcl}
Q^{(i)} &=& Q^{(i)}_1 + Q^{(i)}_2\qquad (0 \le i\le k)
\end{array}\]
Transitions
\[
\frac{\dagger\trans{\mq}_2 q_2^{(0)}}{\dagger\trans{\mrun} q_2^{(0)}}
\]
\[
\frac{(q_2^{(0)}, \cdots, q_2^{(j)}) \trans{\mm}_2 (r_2^{(0)},\cdots, r_2^{(j')})\qquad \mm\not\in\{0,\cdots,\imax\} }{
(q_2^{(0)}, \cdots, q_2^{(j)}) \trans{\mm} (r_2^{(0)},\cdots, r_2^{(j')})}
\]
\[
\frac{q_2^{(0)}\trans{i}_2 \dagger \qquad i\in\{0,\cdots,\imax\}\qquad
\dagger \trans{\mwrite{i}}_1 q_1^{(0)} \qquad 
q_1^{(0)} \trans{\mm}_1 (r_1^{(0)},\cdots, r_1^{(j')})\qquad \mm\neq\mok}{
q_2^{(0)} \trans{\mm} (r_1^{(0)},\cdots, r_1^{(j')})}
\]
\[
\frac{
(q_1^{(0)}, \cdots, q_1^{(j)}) \trans{\mm}_1 (r_1^{(0)},\cdots, r_1^{(j')})\qquad \mm\not\in \{\mread,\mwrite{0},\cdots,\mwrite{\imax},0,\cdots, \imax,\mok\}}{(q_1^{(0)}, \cdots, q_1^{(j)}) \trans{\mm} (r_1^{(0)},\cdots, r_1^{(j')})}
\]
\[
\frac{q_1^{(0)}\trans{\mok}_1 \dagger}{ q_1^{(0)}\trans{\mdone} \dagger}
\]

None of the cases requires an adjustment of pointers, because the inherited indices are accurate.

\item $\seq{\Gamma}{\grb{M_1}:\comt}$: $Q^{(j)}= Q_1^{(j)}$ ($0\le j\le k$). Here we simply need to direct the automaton to perform the same transitions as $M_1$ would, starting from $\mgrb$.
At the same time, $\mgrb$ and the corresponding answer $\mok$ have to be relabelled as $\mrun$ and $\mdone$ respectively.
\[
\frac{\dagger\trans{\mgrb}_1 q_1^{(0)}}{\dagger \trans{\mrun} q_1^{(0)}}\qquad
\frac{(q_1^{(0)}, \cdots, q_1^{(j)}) \trans{\mm}_1 (r_1^{(0)},\cdots, r_1^{(j')})\qquad\mm\neq\mgrb,\mrls,\mok}{(q_1^{(0)}, \cdots, q_1^{(j)}) \trans{\mm} (r_1^{(0)},\cdots, r_1^{(j')})}\qquad
\frac{q_1^{(0)}\trans{\mok}_1 \dagger}{q_1^{(0)} \trans{\mdone} \dagger}
\]

\item $\seq{\Gamma}{\rls{M_1}:\comt}$: $Q^{(j)}= Q_1^{(j)}$ ($0\le j\le k$). Here we simply need to direct the automaton to perform the same transitions as $M_1$ would, starting from $\mrls$.
At the same time, $\mrls$ and the corresponding answer $\mok$ have to be relabelled as $\mrun$ and $\mdone$ respectively.
\[
\frac{\dagger\trans{\mrls}_1 q_1^{(0)}}{\dagger \trans{\mrun} q_1^{(0)}}\qquad
\frac{(q_1^{(0)}, \cdots, q_1^{(j)}) \trans{\mm}_1 (r_1^{(0)},\cdots, r_1^{(j')})\qquad\mm\neq\mgrb,\mrls,\mok}{(q_1^{(0)}, \cdots, q_1^{(j)}) \trans{\mm} (r_1^{(0)},\cdots, r_1^{(j')})}\qquad
\frac{q_1^{(0)}\trans{\mok}_1 \dagger}{q_1^{(0)} \trans{\mdone} \dagger}
\]

{
\item $\seq{\Gamma}{\mkvar{M_1}{M_2}:\vart}$.  Recall that $\seq{\Gamma}{M_1:\expt\rarr\comt}$. 
Because we are using terms in normal form $M_1=\lambda x^\expt.M_1'$.
For $0\le i\le\imax$, consider $N_i=M_1'[i/x]$, which is of smaller size than $M_1$. Let us apply IH to $N_i$
and write $Q_{1i}^{(j)}$ and $\trans{}_{1i}$ for components of the resultant automaton.

Let $Q^{(j)}= \sum_{i=0}^\imax Q_{1i}^{(j)} +Q_2^{(j)}$ ($0<j\le k$).
In this case, after $\mwrite{i}$ we redirect transitions to the automaton for $N_i$, and after $\mread$ - to $M_2$,
relabelling the initial and final moves as appropriate.

\[
\frac{\dagger\trans{\mrun}_{1i} q_{1i}^{(0)}\qquad 0\le i\le \imax}{\dagger \trans{\mwrite{i}} q_{1i}^{(0)}}\qquad
\frac{q_{1i}^{(0) }\trans{\mdone}_{1i} \dagger \qquad 0\le i\le \imax}{q_{1i}^{(0)} \trans{\mok} \dagger}
\]
\[
\frac{(q_1^{(0)}, \cdots, q_1^{(j)}) \trans{\mm}_{1i} (r_1^{(0)},\cdots, r_1^{(j')})\qquad\mm\neq\mrun,\mdone}{(q_1^{(0)}, \cdots, q_1^{(j)}) \trans{\mm} (r_1^{(0)},\cdots, r_1^{(j')})}
\]

\[
\frac{\dagger\trans{\mq}_2 q_2^{(0)}}{\dagger \trans{\mread} q_2^{(0)}}\qquad
\frac{(q_2^{(0)}, \cdots, q_2^{(j)}) \trans{\mm}_2 (r_2^{(0)},\cdots, r_2^{(j')})\qquad\mm\neq\mq,i}{(q_2^{(0)}, \cdots, q_2^{(j)}) \trans{\mm} (r_2^{(0)},\cdots, r_2^{(j')})}\qquad
\frac{q_2^{(0)}\trans{i}_2 \dagger}{q_2^{(0)} \trans{i} \dagger}
\]

\item $\seq{\Gamma}{\mksem{M_1}{M_2}:\semt}$. $Q^{(j)}= Q_1^{(j)}+Q_2^{(j)}$ ($0\le j\le k$). In this case, after $\mgrb$ we redirect transitions to the automaton for $M_1$, and after $\mrls$ - to $M_2$.

\[
\frac{\dagger\trans{\mrun}_1 q_1^{(0)}}{\dagger \trans{\mgrb} q_1^{(0)}}\qquad
\frac{(q_1^{(0)}, \cdots, q_1^{(j)}) \trans{\mm}_1 (r_1^{(0)},\cdots, r_1^{(j')})\qquad\mm\neq\mrun,\mdone}{(q_1^{(0)}, \cdots, q_1^{(j)}) \trans{\mm} (r_1^{(0)},\cdots, r_1^{(j')})}\qquad
\frac{q_1^{(0)}\trans{\mdone}_1 \dagger}{q_1^{(0)} \trans{\mok} \dagger}
\]

\[
\frac{\dagger\trans{\mrun}_2 q_2^{(0)}}{\dagger \trans{\mrls} q_2^{(0)}}\qquad
\frac{(q_2^{(0)}, \cdots, q_2^{(j)}) \trans{\mm}_2 (r_2^{(0)},\cdots, r_2^{(j')})\qquad\mm\neq\mrun,\mdone}{(q_2^{(0)}, \cdots, q_2^{(j)}) \trans{\mm} (r_2^{(0)},\cdots, r_2^{(j')})}\qquad
\frac{q_2^{(0)}\trans{\mdone}_2 \dagger}{q_2^{(0)} \trans{\mok} \dagger}
\]
}
\end{itemize}
\subsection{Example}\label{apx:example}

Here is a worked example of Theorem \ref{thm:trans} for the term $t = $
\[
\seq{f:\comt\rarr\comt}{\newin{x\aasg 0}{(f (x\aasg 1 || x\aasg 13)\, ||\, \cond{!x=13}{\,\skipcom\,}{\,\divcom}})}
\]
\cutout{
We start with the following statement:

\begin{verbatim}
newvar x=0 in 
     f( write(x:=1) || write(x:=13) ) 
  || if read(x) = 13 then skip else div
\end{verbatim}
}

We will show some simple subterms of this term, and then how to combine them using $||$ and introduce $\textbf{newvar}$. We will first construct the sub-automaton representing the following subterm:
\[
f (x\aasg 1 || x\aasg 13)
\]
\cutout{
\begin{verbatim}
    f(write (x:=1) || write(x:=13))
\end{verbatim}
}

For convenience we will call this subterm $w$ as in ``write''. The states for $\mathcal{A}(w)$ are as follows:
\begin{align*}
Q_w^{(0)} &= \{0_w, 1_w, 2_w\} & \qquad Q_w^{(1)} &= \{0_w\} \\
Q_w^{(2)} &= \{0_1,1_1,2_1\} \times \{0_{13},1_{13},2_{13}\} & \qquad Q_w^{(3)} &= \{0_1,0_{13}\}
\end{align*}

Note: in the standard construction, the subterms will not be annotated with the subscripts given. We show them here to emphasise that the union operation performed by combining branches is the \emph{disjoint} union of the states from each side.

The transitions for $\mathcal{A}(w)$ are as follows. When we write transitions here, places where values are symbolic (e.g. $u$ or $v$) represent one transition for every possible value that may appear in those places.

%\adnote{Not sure how to nicely align the below - happy for it to overflow the margins but it would be nice to center it all!}

\[
\dagger \trans{\mathsf{run}} 0_w 
\qquad 
2_w \trans{\mathsf{done}} \dagger
\]
\[
0_w \trans{\mathsf{run}^{(f,0)}} (1_w, 0_w) 
\qquad 
(1_w, 0_w) \trans{\mathsf{done}^{(f,0)}} 2_w
\]
\[
(1_w, 0_w) \trans{\mathsf{run}^{(f1,0)}} (1_w, 0_w, (0_1, 0_{13}))
\qquad
(1_w, 0_w, (2_1, 2_{13})) \trans{\mathsf{done}^{(f1,0)}} (1_w, 0_w)
\]
\[
(1_w, 0_w, (0_1, v)) \trans{\mwrite{1}^{(x,2)}} (1_w, 0_w, (1_1, v), 0_1)
\qquad
(1_w, 0_w, (1_1, v), 0_1) \trans{\mathsf{ok}^{(x,0)}} (1_w, 0_w, (2_1, v))
\]
\[
(1_w, 0_w, (u, 0_{13})) \trans{\mwrite{1}^{(x,2)}} (1_w, 0_w, (u, 1_{13}), 0_{13})
\qquad
(1_w, 0_w, (u, 1_{13}), 0_{13}) \trans{\mok^{(x,0)}} (1_w, 0_w, (u, 2_{13}))
\]

where $u \in \{0_1,1_1,2_1\}$ and $v\in \{0_{13},1_{13},2_{13}\}$.
\vspace{3mm}

We now do the same for the following term, $r$ (for ``read''):
\[
\cond{!x=13}{\,\skipcom\,}{\,\divcom}
\]
\cutout{
\begin{verbatim}
    if read(x) = 13 then skip else div
\end{verbatim}
}

The states for $\mathcal{A}(r)$ are simpler, as this term is shallow.
\[
Q_r^{(0)} = \{ 0_r, 1_r, 2_r^0, \cdots, 2_r^{\imax} \}
\qquad
Q_r^{(1)} = \{ 0_r \}
\]

The transitions for $\mathcal{A}(r)$ are as follows.
\[
\dagger \trans{\mathsf{run}} 0_r 
\qquad
2_r^{13} \trans{\mdone} \dagger
\]
\[
0_r \trans{\mread^{(x,0)}} (1_r, 0_r)
\qquad
(1_r,0_r) \trans{z^{(x,0)}} 2_r^{z}
\]
where $z\in\{0,\cdots,\imax\}$. 
Observe that only reaching state $2_r^{13}$ (hence, reading a value $13$ from $x$) will allow this automaton to terminate.

\vspace{3mm}

Combining these two automata is relatively simple. We will first apply the procedure for parallel composition ($||$), and then apply the $\textbf{newvar}$ context. See Theorem \ref{thm:trans} for the precise workings of these steps. The final automaton $\mathcal{A}(t)$ for our term $t$ is as follows.

\vspace{1mm}

States:
\begin{align*}
Q^{(0)} = (Q^{\prime (0)}+ Q^{\prime (0)} \times \{lock\}) \times X\\
\text{ where } Q^{\prime (0)} = Q_w^{(0)} \times Q_r^{(0)} \text{ and } X = \{0,\cdots,\imax\}
\end{align*}
\[
Q^{(1)} = Q_r^{(1)} + Q_w^{(1)} \qquad \qquad
Q^{(2)} = Q_w^{(2)} \qquad \qquad
Q^{(3)} = Q_w^{(3)}
\]

\vspace{1mm}

Transitions:
\[
\dagger \trans{\mathsf{run}} ((0_r, 0_w), 0)
\qquad
((2_w, 2_r^{13}), n) \trans{\mathsf{done}} \dagger
\]
\[
((0_w, b), n) \trans{\mathsf{run}^{(f,0)}} (((1_w, b), n), 0_w)
\qquad
(((1_w, b), n), 0_w) \trans{\mathsf{done}^{(f,0)}} ((2_w, b), n)
\]
\begin{align*}
(((1_w, b),n), 0_w) 
& \trans{\mathsf{run}^{(f1,0)}} (((1_w, b), n), 0_w, (0_1, 0_{13}))
\\
(((1_w, b), n), 0_w, (2_1, 2_{13})) 
& \trans{\mathsf{done}^{(f1,0)}} (((1_w, b),n), 0_w)
\\\\
(((1_w, b), n), 0_w, (0_1, v)) 
& \trans{\eq} (((1_w, b), lock, 1), 0_w, (1_1, v), 0_1)
\\
(((1_w, b), lock, n), 0_w, (1_1, v), 0_1) 
& \trans{\ea} (((1_w, b), n), 0_w, (2_1, v))
\\
(((1_w, b), n), 0_w, (u, 0_{13})) 
& \trans{\eq} (((1_w, b), lock, 13), 0_w, (u, 1_{13}), 0_{13})
\\
(((1_w, b), lock, n), 0_w, (u, 1_{13}), 0_{13}) 
& \trans{\ea} (((1_w, b), n), 0_w, (u, 2_{13}))
\\\\
((a, 0_r), n) 
& \trans{\eq} (((a, 1_r), lock, n), 0_r)
\\
(((a, 1_r), lock, n),0_r) 
& \trans{\ea} ((a, 2_r^{n}), n)
\end{align*}

where $u \in \{0_1,1_1,2_1\}$, $v\in \{0_{13},1_{13},2_{13}\}$, $a \in \{ 0_w, 1_w, 2_w \}$ and $b \in \{ 0_r, 1_r, 2_r \}$.

\section{Additional material for Section~\ref{sec:tosla}}

\subsection{Proof of Theorem~\ref{thm:trans2}}

We start with a technical lemma that identifies
the level of moves corresponding to free variables
of type $\vart$ and $\semt$.
Given $x:\vart$, moves of the form $\mwrite{i}^{(x,\rho)}$
and $\mread^{(x,\rho)}$ (by P) will be referred to 
as the associated questions, while
$\mok^{(x,\rho)}$ and $i^{(x,\rho)}$ (by O) will be called
the associated answers.
We use analogous terminology for $x:\semt$: the associated questions
are $\mgrb^{(x,\rho)}$ and $\mrls^{(x,\rho)}$, while the associated answer is $\mok^{(x,\rho)}$.

\begin{lemma}\label{lem:ade}
Given a $\fica$-term $\seq{\Gamma}{M:\theta}$ in 
$\beta\eta$-normal form, let $\Aut_M$ be the automaton produced 
by Theorem~\ref{thm:trans}.
For any $x:\vart$ or $x:\semt$ such that $\ade{x}{M}=i$,
the transitions corresponding to 
the moves associated with $x$ add/remove leaves 
at odd levels $1,3,\cdots, 2i-1$.
\end{lemma}
\begin{proof}
We reason by induction on $M$, inspecting each construction in turn. 

For $M\equiv \skipcom,\divcom,i$, the result holds vacuously,
because there are no moves associated with $x$ ($i=0$).

In the following cases, $\ade{x}{M}$ is calculated
by taking the maximum of $\ade{x}{M'}$ for subterms
and the automata constructions
never modify the level of transitions in automata
obtained by IH.
Consequently, the lemma can be established 
by appeal to IH:
$M_1||M_2$, $M_1;M_2$, $\cond{M_1}{M_2}{M_3}$,
$\while{M_1}{M_2}$, $!M_1$, $M_1\aasg M_2$,
$\grb{M_1}$, $\rls{M_1}$, $\newin{y}{M_1}$,
$\newsem{y}{M_1}$.

The remaining case is $M\equiv f M_h\cdots M_1$.
\begin{itemize}
\item Note that this case also covers $f\equiv x$, in which
case $\ade{x}{M}=1$ and transitions associated with $x$
involved leaves at level $2\cdot 1-1=1$, as required.
\item If $f\not\equiv x$ then $\ade{x}{M}=1+\max(\ade{x}{M_1},\cdots, \ade{x}{M_h})$.
In this case, the automata construction lowers  transitions associated with $x$ by exactly two levels, 
so by IH, they will appear at levels $1+2,\cdots, (2i-1)+2$. Note that $(2i-1)+2 = 2(i+1)-1$, i.e. the lemma holds.
\end{itemize}
\end{proof}

Observe that subterms of $\lfica$ terms are in $\lfica$,
i.e. we can reason by structural induction.

\begin{lemma}
Suppose $\seq{\Gamma}{M:\theta}$ is from $\lfica$.
The automaton $\clg{A}_M$ obtained from the translation in Theorem~\ref{thm:trans} 
is presentable as a $\sla$.
\end{lemma}
\begin{proof}
\begin{description}
\item In many cases, the construction merely relabels the given automaton.
Then a simple appeal to the inductive hypothesis will suffice.
The relevant cases are: $!M_1, \arop{M_1}, \rls{M_1}, \grb{M_1}, \lambda x. M_1$.

\item[$M\equiv M_1 || M_2$] 
The case of parallel composition involves running copies of $M_1$ and $M_2$ in parallel without communication,
with their root states stored as a pair at level $0$. Note, though, that each of the automata transitions independently
of the state of the  other automaton, which means that, if the automata $M_1$ and $M_2$ are $\sla$, 
so will be the automaton for $M_1 || M_2$. 
The branching bound after the construction is the sum of the two bounds  for $M_1$ and $M_2$.

\item [$M\equiv M_1;M_2$] The construction schedules the automaton for $M_1$ first and there is a transition to (a disjoint copy of) the second one only after the configuration of the first automaton consists of the root only. Otherwise the automata never communicate. As the transition from the first to the second automaton happens at the root, it can  be captured as a $\sla$ transition. Consequently, 
if the automata for $M_1,M_2$ are $\sla$, so is the automaton for $M$. Here the branching bound
is simply the maximum of the bounds for $M_1$ and $M_2$.

The same argument applies to $\cond{M_1}{M_2}{M_3}$, $M_1\aasg M_2$.

\item[$M\equiv \newin{x\aasg i}{M_1}$] 

Transitions not associated with $x$ are embedded into the automaton 
for $M$ except that at level $0$, the new automaton keeps track of the current value stored in $x$.  Because these transitions proceed uniformly without ever depending on the value stored at the root, this is consistent with $\sla$ behaviour.

For transitions associated with $x$, we note that,
because $M$ is from $\lfica$, we have $\ade{x}{M_1}\le 2$. 
By Lemma~\ref{lem:ade}, this means that the transitions
related to $x$ correspond to creating/removing leaves at either level $1$ or $3$.
These transitions need to read/write the root but, because they concern nodes at level $0$ or $3$,
they will be consistent with the definition of a $\sla$.
All other transitions (not labelled by $x$) proceed as in $M$ and
need not consult the additional information about the current state
stored in the root (the extra information is simply propagated). Consequently, if $M$ is represented by a $\sla$ then the interpretation of $\newin{x\aasg i}{M}$ is also a $\sla$. The construction does not affect the branching bound, because
the resultant runs can be viewed as a subset of runs of the automaton for $M$, i.e. those in which reads and writes are related.

The case of $M\equiv \newsem{x\aasg i}{M_1}$ is analogous.

\item[$M\equiv f M_h\cdots M_1$]

For $f M_h\cdots M_1$, we observe that the construction first 
creates two nodes at levels $0$ and $1$, and the node at level $1$ is used
to run an unbounded number of copies of (the automaton for) $M_i$.
The copies do not need access to the states stored at levels $0$ and $1$, 
because they are never modified when the copies are running.
Consequently, if each $M_i$ can be translated into a $\sla$,
the outcome of the construction in Theorem~\ref{thm:trans} is also a $\sla$. The new branching bound is the maximum over bounds from $M_1,\cdots, M_h$, because at even levels children are produced as in $M_i$ and 
level $0$ produces only $1$ child.
\end{description}
\end{proof}
% !TEX root =  main.tex

\section{Additional material for Section~\ref{sec:tofica}}
\label{apx:tofica}

\paragraph{Word representation}

Let $\Aut=\abra{\Sigma,k,Q,\delta}$ be a leafy automaton. 
We shall assume that $\Sigma,Q\subseteq\{0,\cdots,\imax\}$ so that we can
encode the alphabet and states using type $\expt$.
First we discuss how to assign a play $\play{w}$ to a trace $w$ of $\Aut$.
The basic idea is to simulate each transition with two moves, by $O$ and $P$ respectively.
The child-parent links in $\D$ will be represented by justification pointers.
\begin{itemize}
\item Suppose $w=w' (t,d)$ with $t\in\Sigma_\Q$.
We will  represent $(t,d)$ by a segment of the form $\rnode{A}{{\mq}^{{\vec{i}}}}\quad \rnode{B}{\mrun}^{t\vec{i}}\justg{B}{A}$.
If $w'=\epsilon$, we let $\play{w}=\rnode{A}{\mq} \,\,\rnode{B}{\mrun^{t}}\justf{B}{A}$, i.e. $\vec{i}=\epsilon$.
If $w'\neq\epsilon$ then, because $w$ is a trace, $w'$ must contain a unique
occurrence of $(t',\pred{d})$ for some $t'\in\Sigma_\Q$.
Then, if $(t',\pred{d})$ was represented by $\rnode{A}{\mq^{\vec{i'}}} \rnode{B}{\mrun}^{t'\vec{i'}}\justg{B}{A}$ in $\play{w'}$, 
we let $\play{w}=\rnode{C}{\play{w'}}\,\,\,\rnode{K}{\rnode{D}{{\mq}}^{1 t'{\vec{i'}}}}\justn{D}{C}{100}\,\,\, \rnode{B}{\mrun}^{t 1 t'\vec{i'}}\justn{B}{K}{100}$,
where  $\mq^{1 t'\vec{i'}}$ points at $\mrun^{t' \vec{i'}}$.
\item Suppose $w=w'(t,d)$ with $t\in\Sigma_{\A}$.
Because $w$ is a trace, $w'$ must contain a unique occurrence $(t',d)$ for some $t'\in\Sigma_\Q$.
If $(t',d)$ is represented by the segment $\rnode{A}{\mq^{\vec{i}}} \rnode{B}{\mrun^{t'\vec{i}}}\justf{B}{A}$
in $\play{w'}$,
we set $\play{w}=\play{w'} \,\,\mdone^{t'\vec{i}}\,\, t^{\vec{i}}$, where the two answer-moves are justified by 
$\mrun^{t'\vec{i}}$ and ${\mq^{\vec{i}}}$ respectively. 
Because $w$ is a trace, we can be sure that after processing $w'$, $\Aut$ enters a configuration in which $d$ is a leaf.
Thus, the two answers will satisfy the game-semantic $\wait$ condition, and $\play{w}$ will be well-defined.
\end{itemize}
The $\fork$ condition is  satisfied for $\play{w}$, because reading an answer removes the corresponding data value from the configuration and, hence, it cannot be used as a justifier afterwards.
In what follows, we write $\theta^n\rarr\beta$ for $\underbrace{\theta \rarr\cdots\rarr\theta}_{n}\rarr\beta$ for $n\in\N$.
The lemma below identifies the types that correspond to our encoding of traces.
\begin{lemma}
Let $N=\imax+1$. Suppose $\Aut$ is a $k$-$\la$ and $w\in\trace{\Aut}$.
Then $\play{w}$ is a play in $\sem{\theta_k}$,
where $\theta_0=\comt^N\rarr\expt$ and $\theta_{i+1}=(\theta_i\rarr\comt)^N\rarr\expt$ ($i\ge 0$).
\end{lemma}

\cutout{
Before we state the main result, we recall from~\cite{GM08} that strategies corresponding to $\fica$ terms 
satisfy a closure condition known as~\emph{saturation}: swapping two adjacent moves in a play belonging 
to such a strategy yields another play from the same strategy, 
provided the swap yields a play and it is not the case
that the first move is an O-move and the second one a P-move.
Thus, saturated strategies express causal dependencies of P-moves on O-moves.
Consequently, one cannot expect to find a $\fica$-term such that the corresponding
strategy is the smallest strategy containing $\{\,\play{w}\,|\,w\in \trace{\Aut}\,\}$. 
Instead, the best one can aim for is the following result.
}

\subsection{Saturation}
{

The game model~\cite{GM08} of $\fica$ consists of \emph{saturated} strategies only: the saturation
condition stipulates that all possible (sequential) observations of
(parallel) interactions must be present in a strategy: actions of the
environment (O) can always be observed earlier if possible, actions of the
program (P) can be observed later. To formalize this, for any arena
$A$, we define a preorder $\preceq$ on $P_A$, as the least transitive
relation $\preceq$ satisfying 
$s\, o\, m\, s'\preceq s\, m\, o\, s'$ and $s\, m\, p\, s'\preceq s\, p\, m\, s'$
for all $s,s'$,
where $o$ and $p$ are an O- and  a P-move respectively (in the above pairs of plays 
moves on the left-hand-side of $\preceq$ are assumed to have the same justifiers as on the right-hand-side). 
\begin{definition}\label{def:sat}
A strategy $\sigma:A$ is \emph{saturated} iff, for all $s,s'\in P_A$,
if $s\in \sigma$ and $s'\preceq s$ then $s'\in\sigma$.
\end{definition}
\begin{remark}\label{rem:causal}
Definition~\ref{def:sat} states that saturated strategies are stable 
under certain rearrangements of moves.
Note that $s_0\,  p\, o\, s_1\not \preceq s_0\, o\, p\, s_1$, while other move-permutations are allowed.
Thus, saturated strategies express causal dependencies of P-moves on O-moves. This partial-order aspect 
is captured explicitly in concurrent games based on event structures~\cite{CCRW17}.
\end{remark}
}

\subsection{Proof of Theorem~\ref{thm:toalgol}}

\begin{proof}
Our assumption $Q\subseteq\{0,\cdots,\imax\}$ allows us to maintain $\Aut$-states in the memory of $\fica$-terms.
A question $t_{\Q}^{(i)}$ read by $\Aut$ at level $i$ is represented by the variable $f^{(i)}_{t_{\Q}^{(i)}}$,
the corresponding answers $t_{\A}^{(i)}$ are represented by constants $t_{\A}^{(i)}$ (using our assumption $\Sigma\subseteq\{0,\cdots,\imax\}$).
The level $i$ of the data tree is encoded by the order of the variable
$f^{(i)}_{t_{\Q}^{(i)}}$.
For $0\le i < k$,  the variables $f_t^{(i)}$ are meant to have type
$\theta_{k-i-1}\rarr\comt$ and $f_t^{(k)}:\comt$.
This ensures that questions and answers respect the tree structure on data.
To achieve nesting, we rely on a  higher-order structure of the term:
$\lambda f^{(0)}.f^{(0)}(\lambda f^{(1)}.f^{(1)}(\lambda f^{(2)}.f^{(2)}(\cdots \lambda f^{(k)}. f^{(k)})))$.
Recall that the semantics of $f M$ consists of an arbitrary number of interleavings of $M$. This feature is used to mimic the fact that a leafy
automaton can spawn unboundedly many offspring. 
Finally, instead of single variables $f^{(i)}$, we will actually use sequences
$f^{(i)}_0\cdots f^{(i)}_\imax$, which will be used to induce the right move $\mrun^{t\vec{i}}$ when representing  $t\in\Sigma_{\Q}\subseteq \{0,\cdots,\imax\}$. Additionally, the term contains state-manipulating code that enables $P$-moves only if they are
consistent with the transition function of $\Aut$.
To achieve this, every level is equipped with a local variable $X^{(i)}$ of type
$\expt$, so that 
states on a single branch are represented by $\vvec{X^{(i)}} = (X^{(0)},\cdots,X^{(i)})$.

Given $\alpha\in\{\Q,\A\}$ and $-1\le j\le k$, we write
$\vvec{r_{\alpha}^{(j)}}$ for a tuple of values $(r_{\alpha}^{(0)},\cdots, r_{\alpha}^{(j)})$ on
the understanding that $\vvec{r_{\alpha}^{(-1)}}=\dagger$. A similar convention will apply to  $\vvec{u_{\alpha}^{(j)}}$.
Then we use $\vvec{X^{(i)} [ {u_{\alpha}^{(j')}}/{r_{\alpha}^{(j)}}]}$, where $-1\le j,j' \le i$, as shorthand for $\fica$ code that checks componentwise whether the values of
$\vvec{X^{(j)}}$ equal $\vvec{r_{\alpha}^{(j)}}$ and, if so, updates $\vvec{X^{(j')}}$ to $\vvec{u_{\alpha}^{(j')}}$ (if the check fails, the code should diverge). For $j=-1$ (resp. $j'=-1$), 
there is nothing to check (resp. update). All occurrences of $\vvec{X^{(i)} [ {u_{\alpha}^{(j')}}/{r_{\alpha}^{(j)}}]}$ will be protected by a semaphore to ensure mutual exclusion. Consequently,
they will induce exactly the causal dependencies (cf. Remark~\ref{rem:causal})
consistent with sequences of $\Aut$-transitions, i.e. with the shape of $\play{w}$ for some $w\in\trace{\Aut}$.
To select transitions at each stage, we rely on non-deterministic choice $\bigoplus$, which can be encoded in $\fica$\footnote{
$M_1\oplus M_2 = 
\newin{X\aasg 0}{((X\aasg 0\, \parc\, X\aasg 1});\cond{!X}{M_1}{M_2})$.}.

Below we define inductively a family of terms $\seq{}{M_i:\theta_{k-i}}$ ($0\le
i\le k$). Term $M_\Aut$ is then obtained by making a simple change to $M_0$.
For any $0\le i\le k$, let $M_i$ be the term
\[\begin{array}{rl}
\lambda f_0^{(i)}\cdots f_\imax^{(i)}. &\newin{X^{(i)}\aasg 0}{}\\
\bigoplus\limits_{(\vvec{r_{\Q}^{(i-1)}}, {\displaystyle t_{\Q}^{(i)}}, \vvec{u_{\Q}^{(i)}})\in\delta_{\Q}^{(i)}}& \Big( \grb{s};  \vvec{X^{(i)} [ {u_{\Q}^{(i)}}/{r_{\Q}^{(i-1)}}]} ;\rls{s};\quad f_{t_{\Q}^{(i)}}^{(i)}\, M_{i+1}; \\[-4mm]
&\quad \bigoplus_{(\vvec{r_{\A}^{(i)}}, {\displaystyle t_{\A}^{(i)}}, \vvec{u_{\A}^{(i-1)}})\in\delta_{\A}^{(i)}} \big( \grb{s};  \vvec{X^{(i)} [ {u_{\A}^{(i-1)}}/{r_{\A}^{(i)}}]} ;\rls{s}; t_{\A}^{(i)}\big)\Big).
\end{array}\]
We  write $M_{k+1}$ for empty space (this is for a good reason, because
$f_t^{(k)}:\comt$).
The above term $M_i$ declares a new variable to store the state, and then makes a
non-deterministic choice for question transitions that create data values at
level $i$. 
The update of the state is protected by a semaphore.
Then the appropriate $f^{(i)}_t$ is applied to term $M_{i+1}$ that simulates moves of the
automaton on data in the subtree of the freshly created node.
This is followed by the code making a non-deterministic choice over all answer transitions. To define $M_{\Aut}$, it now suffices to declare the semaphore in $M_0$, i.e. given $M_0 = \lambda f_0^{(0)}\cdots f_\imax^{(0)}. {\newin{X^{(0)}\aasg 0}{M}}$
we let $M_\Aut$ be
\[
\lambda f_0^{(0)}\cdots f_\imax^{(0)}. \newsem{s\aasg 0}{\newin{X^{(0)}\aasg 0}{M}}.
\]
\end{proof}

\begin{example}
We illustrate the outcome of the construction from Theorem~\ref{thm:toalgol} for $k=1$.

\[\arraycolsep=1.4pt
\begin{array}{lll}
\lambda f_0^{(0)}\cdots f_\imax^{(0)}. &\multicolumn{2}{l}{\newsem{s\aasg 0}{ \newin{X^{(0)}\aasg 0}{}}}\\
\bigoplus\limits_{(\dagger, {t_{\Q}^{(0)}}, {u_{\Q}^{(0)}})\in\delta_{\Q}^{(0)}}&\multicolumn{2}{l}{ \Bigg( \grb{s}; \, \vvec{X^{(0)} [ {u_{\Q}^{(0)}}/\dagger]};\,\rls{s};}\\
 &\multicolumn{1}{r}{\quad f_{t_{\Q}^{(0)}}^{(0)}\,\, \bigg(\,\, \lambda f_0^{(1)}\cdots f_\imax^{(1)}.} & {\newin{X^{(1)}\aasg 0}{}}\\
&\multicolumn{1}{r}{\bigoplus\limits_{({r_{\Q}^{(0)}}, {t_{\Q}^{(1)}}, \vvec{u_{\Q}^{(1)}})\in\delta_{\Q}^{(1)}}}& {\Big( \grb{s};  \, \vvec{X^{(1)} [ {u_{\Q}^{(1)}}/{r_{\Q}^{(0)}}]};\, \rls{s};\,\, f_{t_{\Q}^{(1)}}^{(1)}; }\\
&&{\quad \bigoplus_{(\vvec{r_{\A}^{(1)}}, t_{\A}^{(1)}, {u_{\A}^{(0)}})\in\delta_{\A}^{(1)}}}  \big( \grb{s}; \, \vvec{X^{(1)} [ {u_{\A}^{(0)}}/{r_{\A}^{(1)}}]} ;\, \rls{s};\, t_{\A}^{(1)}\big)\Big)\bigg);\\
&\multicolumn{2}{l}{\quad \bigoplus_{({r_{\A}^{(0)}}, t_{\A}^{(0)}, \dagger)\in\delta_{\A}^{(0)}} \big( \grb{s}; \, \vvec{X^{(0)} [ \dagger/{r_{\A}^{(0)}}]} ;\,\rls{s}; \,t_{\A}^{(0)})\Bigg)}\\
\end{array}
\]
where
\[\begin{array}{rcl}
\vvec{X^{(0)} [ {u_{\Q}^{(0)}}/\dagger]} &=& X^{(0)}\aasg u_{\Q}^{(0)}\\
\vvec{X^{(1)} [ {u_{\Q}^{(1)}}/{r_{\Q}^{(0)}}]} &=& \cond{(X^{(0)}=r_{\Q}^{(0)})}{(X^{(0)}\aasg u_{\Q}^{(0)};X^{(1)}\aasg u_{\Q}^{(1)})}{\Omega}\\
\vvec{X^{(1)} [ {u_{\A}^{(0)}}/{r_{\A}^{(1)}}]} &=&  \cond{((X^{(0)}=r_{\A}^{(0)})  \wedge (X^{(1)}=r_{\A}^{(1)}))}{(X^{(0)}\aasg u_{\A}^{(0)})}{\Omega}\\
\vvec{X^{(0)} [ \dagger/{r_{\A}^{(0)}}]}  &=& \cond{(X^{(0)}=r_{\A}^{(0)})}{\skipcom}{\Omega}\\
\Omega &=&\while{1}{\skipcom}
\end{array}\]
\end{example}

\end{document}